\documentclass[envcountsect]{llncs}
\usepackage{graphics}
\usepackage{graphicx}
\usepackage{epsfig}
\usepackage{amssymb}
\usepackage[all]{xy}
\usepackage{subfigure}
\usepackage{latexsym}
\usepackage{amsfonts}
\usepackage{amsmath}

\newcommand{\modalb}[1]{\left[ #1 \right]}
\newcommand{\modald}[1]{\langle #1 \rangle}
\newcommand{\modalt}{\left[\langle\otimes\rangle\right]}
\newcommand{\mutrue}{\mathrm{tt}}
\newcommand{\mufalse}{\mathrm{ff}}
\newcommand{\denot}[1]{\lVert #1 \rVert}

\newtheorem{thm}{Theorem}[section]
\newtheorem{lem}[thm]{Lemma}
\newtheorem{cor}[thm]{Corollary}
\newtheorem{prop}[thm]{Proposition}
\newtheorem{conj}[thm]{Conjecture}
\newtheorem{fact}[thm]{Fact}
\newtheorem{defn}[thm]{Definition}
\newtheorem{notn}[thm]{Notation}

\newtheorem{rem}[thm]{Remark}

\newcommand{\petzero}{*=<10pt>[o][F]{}}
\newcommand{\petone}{*=<10pt>[o][F]{\bullet}}
\newcommand{\petsqa}{*=<10pt>[][F]{a}}
\newcommand{\petsqb}{*=<10pt>[][F]{b}}
\newcommand{\petsqc}{*=<10pt>[][F]{c}}
\newcommand{\petsqe}{*=<10pt>[][F]{}}

\def\cfs/{\mathrel{\mathsf{cfs}}}

\def\Sub{\mathit{Sub}}
\def\I{\mathrel{I}}
\def\ppar/{\mathrel{\mathsf{par}}}
\def\co/{\mathrel{\mathsf{co}}}
\def\Conf/{\mathit{Conf}}
\def\lmu/{\ensuremath{\mathcal{L}_{\mu}}}
\def\tfl/{\ensuremath{\mathbf{\mathbb{L}_{\mu}}}}
\def\tlmu/{\ensuremath{\mathcal{L}^{\mathfrak{\otimes}}_{\mu}}}
\def\clmu/{\ensuremath{\mathcal{L}^{c}_{\mu}}}
\def\deflogo{\ensuremath{\triangleleft}}
\def\eos{\allowbreak\null\hskip\parfillskip\quad\deflogo\nobreak\hskip-\parfillskip\penalty0\par}

\newcommand{\myfrac}[2]
{\begin{tabular}{@{}c@{}}
 \ensuremath{#1}\\ \hline \ensuremath{#2}
 \end{tabular}}

\begin{document}
\title{Logics and Games for True Concurrency}
\author{Julian Gutierrez}
\institute{LFCS. School of Informatics. University of Edinburgh}

\maketitle 

\begin{abstract}
We study the underlying mathematical properties of various partial order models of concurrency based on transition systems, Petri nets, and event structures, and show that the concurrent behaviour of these systems can be captured in a uniform way by two simple and general dualities of local behaviour. Such dualities are used to define new mu-calculi and logic games for the analysis of concurrent systems with partial order semantics. Some results of this work are: the definition of a number of \emph{mu-calculi} which, in some classes of systems, induce the same identifications as some of the best known bisimulation equivalences for concurrency; and the definition of (infinite) \emph{higher-order logic games} for bisimulation and model-checking, where the players of the games are given (local) monadic second-order power on the sets of elements they are allowed to play. More specifically, we show that our games are \emph{sound} and \emph{complete}, and therefore, \emph{determined}; moreover, they are \emph{decidable} in the finite case and underpin novel decision procedures for bisimulation and model-checking. Since these mu-calculi and logic games generalise well-known fixpoint logics and game-theoretic decision procedures for concurrent systems with interleaving semantics, the results herein give some of the groundwork for the design of a logic-based, game-theoretic framework for studying, in a \emph{uniform} way, several concurrent systems regardless of whether they have an interleaving or a partial order semantics.\\

\end{abstract}

\section{Introduction}\label{intro}
Concurrency theory studies the logical and mathematical foundations of parallel processes, i.e., of systems composed of independent components which can interact with each other and with an environment. These systems can be analysed by studying the formalisms (logics and methodologies) employed to specify and verify their properties as well as the mathematical structures used to represent their behaviour. Such formalisms and structures make use of models of two different kinds: \emph{interleaving} or \emph{partially ordered}. This semantic feature is particularly important as most logics, tools, and verification techniques for analysing the behaviour of concurrent systems have to take this difference into account. This is sometimes an undesirable situation since it obscures our understanding of concurrent computations and divide research efforts in two different directions. Here we report on some work towards the definition of theories and verification techniques for analysing different models for concurrency in a uniform way.

This study focuses on core issues related to mu-calculi (fixpoint extensions of modal logic, in this case) and infinite logic games for concurrency. In particular, using a game-theoretic approach, we study fixpoint modal logics with partial order models as well as their associated bisimulation and model-checking problems. Our results show that generalisations (to a partial order setting) of some of the theories and verification techniques for interleaving concurrency can be used to address, \emph{uniformly}, the analysis of concurrent systems with both interleaving and partial order semantics. Some of our particular contributions are as follows.

We first study the relationships between logics and equivalences for concurrent systems with partial order semantics purely based on observable `local dualities' between concurrency and conflict, on the one hand, and concurrency and causality on the other. These dualities, which can be found across several partial order models of concurrency, are mathematically supported in a beautiful way by a simple axiomatization of concurrent behaviour. Although the dualities and axiomatization are defined with respect to partial order models of concurrency, such dualities and axiomatization have a natural interpretation when considering concurrent systems with interleaving semantics such as transition systems (or their unfoldings) since they appear as particular instances of our framework. 

We also define a logical notion of equivalence for concurrency tailored to be model independent. We do so by defining a number of fixpoint modal logics whose semantics are given by an intermediate structure called a `process space', which is a mathematical structure intended to be used as a common bridge between the particular models of concurrency under consideration. Roughly speaking, a process space is a structure that contains the local partial order behaviour of a concurrent system, and is built using the local dualities mentioned above. Then, following this approach, two concurrent systems, possibly with models of different kinds, can be compared with each other within the same framework by comparing \emph{logically} their associated process spaces.

Moreover, some of the bisimulation equivalences induced by these logics coincide with the standard bisimilarities both for interleaving and for causal systems, namely with Milner's strong bisimilarity (sb \cite{hmljacm-milner}) and with history-preserving bisimilarity (hpb \cite{hpb-rav}), respectively. The latter result holds when restricted to a particular class of concurrent systems, which we currently call the class of $\Xi$-systems. We also define a new bisimulation equivalence, which (on $\Xi$-systems) is strictly stronger than hpb and strictly weaker than hereditary history-preserving bisimilarity (hhpb \cite{open-nielsen}), one of the specializations of the abstract notion of bisimulation equivalence defined by Joyal, Nielsen, and Winskel using open maps \cite{open-nielsen}. 

We also study the model-checking problem for these logics against the models for concurrency we consider here. The outcome of this is a generalisation of the local model-checking games defined by Stirling \cite{localmc-stirling} for the mu-calculus (\lmu/ \cite{lmutcs-kozen}). This new game-based decision procedure is used for the temporal verification of a class of regular event structures \cite{res-thia}, and thereby, we improve previous results in the literature \cite{mceslics-madhu,tacas-penczek} in terms of temporal expressive power. We do so by allowing a free interplay of \emph{fixpoint operators} and \emph{local monadic second-order power} on the sets of elements that can be described within the logics.

The distinctive feature of the (infinite) logic games we define in order to address the bisimulation and model-checking problems we have described is that through their formal definition we move from a traditional setting where both players, namely a ``Verifier'' Eve ($\exists$) and a ``Falsifier'' Adam ($\forall$), have first-order power on the elements available in the locality where they are to play, to a more complex setting in which the players are provided with higher-order power on the sets of elements they are allowed to play. From a more computational viewpoint, we show that despite their higher-order features both logic games are sound and complete, and therefore, determined; moreover, they are also decidable when played on finite systems allowing for possible practical implementations. 

The structure of the document is as follows. Section \ref{pre} introduces some background on the models for concurrency, fixpoint modal logics, and bisimulation and model-checking games of our interest. In Section \ref{mucalculi} we define the local dualities recognisable in several (partial order) models for concurrency as well as the fixpoint modal logics that can be extracted from such dualities; here we also study the bisimulation equivalences induced by some of the modal logics defined in this section making no use of any game-theoretic machinery. Then, in Sections \ref{bisgames} and \ref{mcgames}, we introduce the higher-order logic games that characterise, respectively, the bisimulation and model-checking problems of the logics defined in the previous section; we also show their correctness and applications as described before. Finally, in Section \ref{relwork} a summary of related work is given, and in Section \ref{conc} we provide some concluding remarks and directions for further work.

\section{Preliminaries}\label{pre}
In this section we study the models for concurrency of our interest, together with background material on the modal logics and games for verification that are relevant to the work presented in this document. We also discuss some relationships between the models for concurrency that are studied here as well as between the equivalences induced by the modal logics presented in this section and the equivalences for concurrency considered in this and forthcoming sections.

\subsection{Partial Order Models of Concurrency}\label{ch2-pom}
In concurrency there are two main semantic approaches to modelling concurrent behaviour, either using interleaving or partial order models for concurrency. On the one hand, \emph{interleaving} models represent concurrency as the nondeterministic combination of all possible sequential behaviours in the system. On the other hand, \emph{partial order} models represent concurrency explicitly by means of an independence relation on the set of actions, transitions, or events in the system that can be executed concurrently.

We are interested in partial order models for various reasons. In particular, because they can be seen as a generalisation of interleaving models as explained later. This feature allows us to define the logics and games developed in further sections in a uniform way for several different models for concurrency, regardless of whether they are used to provide interleaving or partial order semantics. 

In the following, we present the three partial order models for concurrency that we study here, namely Petri nets, transition systems with independence, and event structures. We also present some basic relationships between these three models, and how they generalise some models for interleaving concurrency. For further information on models for concurrency and their relationships the reader is referred to \cite{models-winskel,modelstcs-winskel} where one can find a more comprehensive presentation.

\subsubsection*{Petri Nets.}
A \emph{net} $\mathcal{N}$ is a tuple $(P, C, R, \theta,\Sigma)$, where $P$ is a set of places, $C$ is a set of actions, $R \subseteq {{(P \times C)} \cup {(C \times P)}}$ is a relation between places and actions, and $\theta$ is a labelling function ${\theta} : {C \rightarrow \Sigma}$ from actions to a set $\Sigma$ of action labels. Places and actions are called nodes; given a node $n \in P \cup C$, the set ${{}^{\bullet}n} = {\{ {x} \mid {(x,n) \in R} \}}$ is the preset of $n$ and the set ${n^{\bullet}} = {\{ {y} \mid {(n,y) \in R} \}}$ is the postset of $n$. These elements define the static structure of a net.\footnote{The reader acquainted with net theory may have noticed that we use the word `action' instead of `transition', more common in the literature on (Petri) nets. We have made this choice of notation in order to avoid confusion later on in the document.} The notion of computation state in a net (i.e., its dynamic part) is that of a `marking', which is a set or a multiset of places; in the former case such nets are called \emph{safe}. Hereafter we only consider safe nets.

\begin{defn}\label{ch2-petrinets}
\emph{
A \emph{Petri net} $\mathfrak{N}$ is a tuple $(\mathcal{N},M_0)$, where $\mathcal{N} = (P, C, R, \theta,\Sigma)$ is a net and ${M_0} \subseteq {P}$ is its initial marking.
}
\eos
\end{defn}

As mentioned above, markings define the dynamics of nets; they do so in the following way. We say that a marking $M$ enables an action $t$ iff ${{}^{\bullet}t} \subseteq {M}$. If $t$ is enabled at $M$, then $t$ can occur and its occurrence leads to a successor marking $M'$, where ${M'} = {{({M} \setminus {{}^{\bullet}t})} \cup {t^{\bullet}}}$, written as $M \xrightarrow{t} M'$. Let $\xrightarrow{t}$ be the relation between successor markings and let $\longrightarrow^{*}$ be its transitive closure. Given a Petri net $\mathfrak{N} = (\mathcal{N},M_0)$, the relation $\longrightarrow^{*}$ defines the set of reachable markings in the system $\mathfrak{N}$; such a set of reachable markings is fixed for any pair $(\mathcal{N},M_0)$, and can be constructed with the occurrence net unfolding construction defined by Nielsen, Plotkin, and Winskel \cite{pnesdom-winskel}.

Finally, let $\ppar/$ be the symmetric independence relation on actions such that $t_1 \ppar/ t_2$ iff ${{}^{\bullet}t_{1}^{\bullet}} \cap {{}^{\bullet}t_{2}^{\bullet}} = \emptyset$, where ${{}^{\bullet}t^{\bullet}}$ stands for the set ${{}^{\bullet}t} \cup {t^{\bullet}}$, and there exists a reachable marking $M$ such that both ${}^{\bullet}t_1 \subseteq M$ and ${}^{\bullet}t_2 \subseteq M$. Then, if two actions $t_1$ and $t_2$ can occur concurrently they must be independent, i.e., ${(t_1,t_2)} \in {\ppar/}$. 

\subsubsection*{Transition Systems with Independence.}
A \emph{labelled transition system} (LTS) is an edge-labelled graph structure. Formally, an LTS is a tuple $(S,T,\Sigma)$, where $S$ is a set of vertices called states, $\Sigma$ is a set of labels, and ${T} \subseteq {S \times \Sigma \times S}$ is a set of $\Sigma$-labelled edges, which are called transitions. A \emph{rooted} LTS is an LTS with a designated initial state $s_0 \in S$. A transition system with independence is a rooted LTS where independent transitions can be explicitly recognised. Formally:

\begin{defn}
\emph{
A \emph{transition system with independence} (TSI) $\mathfrak{T}$ is a tuple $(S,s_0,T,I,\Sigma)$, where $S$ is a set of states with initial state $s_0$, ${T} \subseteq {S \times \Sigma \times S}$ is a transition relation, $\Sigma$ is a set of labels, and ${I} \subseteq {T \times T}$ is an irreflexive and symmetric relation on independent transitions. The binary relation $\prec$ on transitions defined by
\begin{center}
$(s,a,s_{1}) \prec (s_{2},a,q) \Leftrightarrow \exists b. (s,a,s_{1}) I (s,b,s_{2}) \wedge (s,a,s_{1}) I (s_{1},b,q) \wedge (s,b,s_{2}) I (s_{2},a,q)$
\end{center}
expresses that two transitions are `instances' of the same action, but in two different interleavings. We let $\sim$ be the least equivalence relation that includes $\prec$, i.e., the reflexive, symmetric, and transitive closure of $\prec$. The equivalence relation $\sim$ is used to group all transitions that are instances of the same action in all its possible interleavings. Additionally, $I$ is subject to the following axioms:
\begin{itemize}
\item \textbf{A1}. $(s,a,s_{1}) \sim (s,a,s_{2}) \Rightarrow s_{1} = s_{2}$ 
\item \textbf{A2}. $(s,a,s_{1}) \I (s,b,s_{2}) \Rightarrow \exists q. (s,a,s_{1}) \I (s_{1},b,q) \wedge (s,b,s_{2}) \I (s_{2},a,q)$
\item \textbf{A3}. $(s,a,s_{1}) \I (s_{1},b,q) \Rightarrow \exists s_{2}. (s,a,s_{1}) \I (s,b,s_{2}) \wedge (s,b,s_{2}) \I (s_{2},a,q)$
\item \textbf{A4}. $(s,a,s_{1}) (\prec \cup \succ) (s_{2},a,q) \I (w,b,w') \Rightarrow (s,a,s_{1}) \I (w,b,w')$ \eos
\end{itemize}
}
\end{defn}

Axiom \textbf{A1} states that from any state, the execution of a transition 
leads always to a unique state. This is a determinacy
condition. Axioms \textbf{A2} and \textbf{A3} ensure that independent
transitions can be executed in either order. Finally, \textbf{A4}
ensures that the relation $I$ is well defined. More precisely,
\textbf{A4} says that if two transitions $t$ and $t'$ are independent,
then all other transitions in the equivalence class
$\left[t\right]_{\sim}$ (i.e., all other transitions that are
instances of the same action but in different interleavings) are
independent of $t'$ as well, and vice versa. Having said that, an alternative and 
possibly more intuitive definition for axiom \textbf{A4} can be given. 
Let $\mathfrak{I}(t)$ be the set $\{ t' \mid t \I t' \}$. Then, axiom \textbf{A4} 
is equivalent to this expression: \textbf{A4}'. $t \sim t_2 \Rightarrow \mathfrak{I}(t) = \mathfrak{I}(t_2)$.

This axiomatization of concurrent behaviour was defined by Winskel and Nielsen \cite{models-winskel}, but has its roots in the theory of traces \cite{tracesbook-maz}, notably developed by Mazurkiewicz for trace languages, one of the simplest partial order models for concurrency. As shown in Figure \ref{ch2-diamond}, this axiomatization can be used to generate a `concurrency diamond' for any two independent transitions $t$ and $t'$.
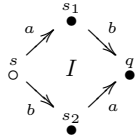
\begin{figure}[]
\begin{displaymath}
\xymatrix @R=5pt @C=10pt 
{
& \stackrel{s_1}{\bullet} \ar[rd]^{b} &  \\
\stackrel{s}{\circ} \ar[ru]^{a} \ar[rd]_{b} & \mathit{I} & \stackrel{q}{\bullet}  \\
& \stackrel{s_2}{\bullet} \ar[ru]_{a} &  \\
} 
\end{displaymath}
\caption{The `concurrency diamond' for $t \I t'$, where $t = (s,a,s_1)$ and $t' = (s,b,s_2)$. Concurrency is depicted with the symbol $I$ inside the square. The initial state is $\circ$.}
\label{ch2-diamond}
\end{figure}

\subsubsection*{Event Structures.} An event structure is a possibly labelled partially ordered set (poset) together with a binary relation on such a set. Formally:
\begin{defn}\label{esdef}
\emph{
A labelled \emph{event structure} $\mathcal{E}$ is a tuple $(E,\preccurlyeq,\sharp,\eta,\Sigma)$, where $E$ is a set of events that are partially ordered by $\preccurlyeq$, the causal dependency relation on events; events in a labelled event structure are `occurrences' of actions in a system. Moreover, ${\eta} : {E \rightarrow \Sigma}$ is a labelling function from events to a set of labels $\Sigma$, and ${\sharp} \subseteq  {E \times E}$ is an irreflexive and symmetric conflict relation such that the following two conditions hold: (1) if ${e_1, e_2, e_3} \in {E}$ and ${e_1 \sharp e_2} \preccurlyeq {e_3}$, then $e_1 \sharp e_3$; and moreover (2) ${\forall e} \in {E}$ the set $\{ {e' \in E} \mid {e' \preccurlyeq e} \}$ is finite.
}
\eos
\end{defn}

The independence relation on events is defined with respect to the causal relation $\preccurlyeq$ and conflict relation $\sharp$ on events. Two events $e_1$ and $e_2$ are said to be concurrent with each other, denoted by $e_1 \co/ e_2$, iff $e_1 \not \preccurlyeq e_2$ and $e_2 \not \preccurlyeq e_1$ and $\neg (e_1 \sharp e_2)$. The notion of computation state for event structures is that of a `configuration'. A configuration $C$ is a conflict-free set of events (i.e., if $e_1,e_2 \in C$, then $\neg (e_1 \sharp e_2)$) such that if $e \in C$ and $e' \preccurlyeq e$, then $e' \in C$. The initial configuration (or initial state) of an event structure $\mathcal{E}$ is by definition the empty configuration $\{\}$. Finally, a successor configuration $C'$ of a configuration $C$ is given by ${C'} = {{C} \cup {\{e\}}}$ such that $e \not \in C$. Write $C \xrightarrow{e} C'$ for this relation, and let $\longrightarrow^{*}$ be defined similar to the Petri net case.

\subsubsection*{A Uniform Representation.}
Despite being different informatic structures, the three models for concurrency just presented have a number of fundamental relationships between them, as well as with some models for interleaving concurrency. More precisely, TSI are noninterleaving transition-based representations of Petri nets, whereas event structures are unfoldings of TSI. This is analogous to the fact that LTS are interleaving transition-based representations of Petri nets while trees are unfoldings of LTS.

There are also simple relationships between TSI and LTS as well as between event structures and trees as follows: LTS are exactly those TSI with an empty independence relation $I$ on transitions, and trees are those event structures with and empty relation $\co/$ on events. In this way, partial order models can generalise the most important interleaving models in concurrency (and in program verification), namely LTS, trees, and Kripke structures (which are the vertex-labelled counterparts of LTS models).

Since the results presented in further sections are valid across all the models previously mentioned, it is convenient to fix some notations to refer unambiguously to any of them. To this end, we use the notation coming from the TSI model and present the maps that determine a TSI model based on the primitives of Petri nets and event structures. Also, with no further distinctions we use the word `system' when referring to any of these models or to sub-models of them, e.g., to LTS or Kripke structures.

The main reason for our choice of notation is that the basic components of a TSI can be easily and uniformly recognised in all the other models studied here. Thus, the translations are simple and direct. Also, this generic setting allows one to see more clearly that the axiomatization presented for TSI also holds for the other partial order models when analysing their local behaviour.

\paragraph*{Petri Nets and Event Structures as TSI Models.} A Petri net $\mathfrak{N} = (\mathcal{N},M_0)$, where $\mathcal{N} = (P, C, R, \theta,\Sigma)$ is a net as defined before and $M_0$ is its initial marking, can be represented as a TSI $\mathfrak{T} = (S,s_0,T,I,\Sigma)$ in the following way:
\begin{center}
\begin{tabular}{lcl}
$S$ & $=$ & $\{M \subseteq P \mid M_0 \longrightarrow^{*} M \}$\\
$T$ & $=$ & $\{ {(M,a,M')} \mid {\exists t \in C . ~  {a = \theta(t)} , M \xrightarrow{t} M'} \}$ \\
$I$ & $=$ & $\{ {((M_1,a,M'_1),(M_2,b,M'_2))} \mid {\exists (t_1,t_2) \in {\ppar/}} . $ \\ && ${ {a = \theta(t_1)}, {b = \theta(t_2)}, M_1 \xrightarrow{t_1} M'_1 , M_2 \xrightarrow{t_2} M'_2 } \}$
\end{tabular} 
\end{center} 
where the set of states $S$ of the TSI $\mathfrak{T}$ represents the set of reachable markings of $\mathfrak{N}$, the initial state $s_0$ is the initial marking $M_0$, the set of labels $\Sigma$ remains the same, and $T$ and $I$ have the expected derived interpretations. Similarly, an event structure $\mathcal{E} = (E,\preccurlyeq,\sharp,\eta,\Sigma)$ determines a TSI $\mathfrak{T} = (S,s_0,T,I,\Sigma)$ by means of the following mapping:
\begin{center}
\begin{tabular}{lcl}
$S$ & $=$ & $\{C \subseteq E \mid \{\} \longrightarrow^{*} C \}$ \\
$T$ & $=$ & $\{ {(C,a,C')} \mid {\exists e \in E . ~ a = \eta(e), C \xrightarrow{e} C'} \}$ \\
$I$ & $=$ & $\{ {((C_1,a,C'_1),(C_2,b,C'_2))} \mid { \exists (e_1,e_2) \in {\co/} . } $ \\ && ${ a=\eta(e_1), b=\eta(e_2), C_1 \xrightarrow{e_1} C'_1, C_2 \xrightarrow{e_2} C'_2 } \}$\\
\end{tabular} 
\end{center}
where the set of states $S$ is the set of configurations of $\mathfrak{E}$, the initial state $s_0$ is the initial configuration $\{\}$, and, as before, the set of labels $\Sigma$ remains the same in both models, and $T$ and $I$ have the expected derived TSI interpretations.

Notice that \emph{actions} in a Petri net, \emph{transitions} in a TSI and \emph{events} in an event structure are all different. As said before, transitions are \emph{instances} of actions, i.e., are actions relative to a particular interleaving. For instance, a Petri net composed of two independent actions ($a \parallel b$ in CCS notation \cite{ccs-milner}) is represented by a TSI with four different transitions, since there are two possible interleavings in such a system, namely $a_1.b_2$ and $b_1.a_2$. Therefore each action in the Petri net for $a \parallel b$ becomes two different transitions in the corresponding TSI.

On the other hand, events are \emph{occurrences} of actions, i.e., are actions relative to the causality relation. For instance, the Petri net representing the system defined by $(a+b).c$, where $a+b$ is the nondeterministic choice between actions $a$ and $b$, and $.$ is the sequential composition of such a choice with the action $c$, is represented by four events, instead of only three, because there are two different causal lines for the execution of action $c$, namely $a.c_1$ and $b.c_2$. Then, the Petri net action $c$ becomes two events $c_1$ and $c_2$ in the corresponding event structure.

\begin{notn}
\emph{
Given a transition $t=(s,a,s')$, also written as $s \xrightarrow{a} s'$ or $s \xrightarrow{t} s'$ if no confusion arises, we have that: state $s$ is called the source of $t$, and write $\sigma(t) = s$; state $s'$ is the target of $t$, and write $\tau(t) = s'$; and $a$ is the label of $t$, and write $\delta(t) = a$.}
\eos
\end{notn}

\begin{rem}
\emph{
The systems we study here may be finite or infinite, and this is always explicitly stated. However, they all are `image-finite', i.e., of finite branching.}
\eos
\end{rem}

\subsection{Modal Logic and the Mu-Calculus}\label{ch2-mlmc}
In this paper we study modal logics based on the mu-calculus (and the mu-calculus itself) since it can be used to express both linear-time and branching-time temporal properties. But first, we review Hennessy--Milner logic (HML \cite{hmljacm-milner}), a precursor modal language to the mu-calculus, which has played a major role in computer science, and especially in the specification of properties of concurrent systems. Then we turn our attention to the modal mu-calculus simply by adding fixpoint operators to HML. After that we look at the logical equivalences induced by these logics and how they have been used as equivalences for concurrency. See \cite{mucalculi-bradfield,processes-stirling} for further information on modal logics, the mu-calculus, or the equivalences induced by such logics.

\subsubsection*{Hennessy--Milner Logic.}
HML is a modal logic of actions that has its roots in the study of process algebras for concurrent and communicating systems. It was intended as an alternative approach to the formalisation of the notion of `observational equivalence' for concurrent systems. As usual for modal logics, HML formulae are interpreted over the set of states of a system.

\begin{defn}
\emph{
\emph{Hennessy--Milner logic} (HML \cite{hmljacm-milner}) has formulae $\phi$ is built from a set $\Sigma$ of labels $a,b,...$ by the following grammar:
\begin{center}
$\phi ::= \mufalse \mid \mutrue \mid \phi_1 \wedge \phi_2 \mid \phi_1 \vee \phi_2 \mid \modald{a} \phi_1 \mid \modalb{a} \phi_1$
\end{center}
where $\mufalse$ and $\mutrue$ are the false and true boolean constants, respectively, $\wedge$ and $\vee$ are  boolean operators, and $\modald{a} \phi_1$ and $\modalb{a} \phi_1$ are the modalities of the logic.}
\eos
\end{defn}

The meanings of $\mufalse$, $\mutrue$, $\wedge$, and $\vee$ are the usual ones. On the other hand, the semantics of the `diamond' modality $\modald{a} \phi_1$ is, informally, that at a given state it is possible to perform an $a$-labelled action to a state where $\phi_1$ holds; and dually for the `box' modality $\modalb{a} \phi_1$. Following \cite{processes-stirling}, we give the denotation of HML formulae inductively using an LTS. The semantics of HML is as follows:
\begin{defn}
\emph{
An \emph{HML model} $\mathfrak{T}$ of a formula $\phi$ is an LTS $(S,T,\Sigma)$. The denotation $\denot{\phi}^{\mathfrak{T}}$ of a formula $\phi$ is given as follows (omitting the superscript $\mathfrak{T}$):
\begin{center}
\begin{tabular}{lll}
$\denot{\mufalse}$&$=$&$\emptyset$ \\
$\denot{\mutrue}$&$=$&$S$ \\
$\denot{\phi_1 \wedge \phi_2}$&$=$&$\denot{\phi_1} \cap \denot{\phi_2}$ \\
$\denot{\phi_1 \vee \phi_2}$&$=$&$\denot{\phi_1} \cup \denot{\phi_2}$ \\
$\denot{\modald{a} \phi_1}$ & $=$ & $\{ s \in S \mid \exists s'. ~{s \xrightarrow{a} s'} \wedge {s' \in \denot{\phi_1}} \}$ \\
$\denot{\modalb{a} \phi_1}$ & $=$ & $\{ s \in S \mid \forall s'. ~{s \xrightarrow{a} s'} \Rightarrow {s' \in \denot{\phi_1}} \}$
\end{tabular} 
\end{center}
The \emph{satisfaction} relation $\models$ is defined in the usual way: $s \models \phi $ iff $s \in \denot{\phi}$.
}
\eos
\end{defn}

One of the most interesting properties of HML is that it characterises `bisimilarity' \cite{hmljacm-milner}, the equivalence relation induced by modal logic. A bisimulation between two rooted systems $\mathfrak{T}_1$ and $\mathfrak{T}_2$ with initial states $s_0$ and $q_0$, respectively, is an equivalence relation $\sim_{b}$ such that $s_0 \sim_{b} q_0$ if, and only if, they satisfy the same set of HML formulae. Then, we say that the two states $s_0$ and $q_0$ (or equivalently the two rooted systems $\mathfrak{T}_1$ and $\mathfrak{T}_2$) are bisimilar iff there is a bisimulation equivalence between them.

HML was initially defined as an alternative approach to understanding process equivalence in the context of CCS; Milner and Hennessy \cite{hmljacm-milner} showed that if two CCS processes are bisimilar, or in their words ``observationally equivalent'', then they satisfy the same set of HML formulae. They found, therefore, a correspondence between the logical equivalence induced by HML and an equivalence for concurrency (bisimilarity or observational equivalence in this case), modulo LTS, the class of models used for giving the semantics of CCS expressions.

Even though HML is quite a natural logic for studying process equivalences, it is not so much as a specification language, since it cannot express many temporal properties. Due to this, stronger logics have been studied. We now review one of such logics, the modal mu-calculus, which has strong connections to HML and a beautiful theory based on the addition of fixpoint operators to modal logic.

\subsubsection*{Fixpoints and the Modal Mu-Calculus.}
Fixpoint logics or mu-calculi \cite{mucalculi-bradfield} are logics that make use of fixpoint operators; in particular, the modal mu-calculus is a simple extension of modal logic with fixpoint operators. The mu-calculus as we use it nowadays was defined by Kozen \cite{lmutcs-kozen}, but it can also be seen as HML with fixpoint operators. The use of fixpoints in program logics was, however, not new by the time the mu-calculus was proposed. It actually dates back at least to Park \cite{fixpv-park} already in the context of program verification.

In informatics, and especially in concurrency and systems verification, the main motivation for extending a logic with fixpoint operators is the ability to express and study temporal properties of systems, this is their (possibly infinite) behaviour. In the reminder of this section we describe the mu-calculus, but before giving a formal presentation of it let us state some concepts and results that relate to fixpoints in general and their ubiquity in lattices and ordered structures.

\paragraph*{Fixpoints in Ordered Structures.} Fixpoints can be seen as equilibrium points. Their definition is simple: given a function $f$, we say that $x$ is a fixpoint of $f$ iff $x = f(x)$; it is a pre-fixpoint of $f$ if $f(x) \leq x$ and a post-fixpoint if $x \leq f(x)$. As we shall see, fixpoint theory is rather useful in logic when $f$ is monotonic and its domain is a complete lattice. Before stating one of the results on fixpoints that is relevant to this work, let us introduce some ordered structures. 

A \emph{partially ordered set} (poset) $(A,\leq)$ is a set $A$ together with a reflexive, transitive and anti-symmetric relation $\leq$ on its elements. A \emph{lattice} $\mathfrak{A} = (A,\leq)$ is a poset where for every two elements $x$ and $y$ in $A$, arbitrary meets (written $x \times y$) and joins (written $x + y$) exist. If, moreover, arbitrary meets and joins exist for any subset $B \subseteq A$, then $\mathfrak{A}$ is a complete lattice. 

\begin{thm}
\emph{
\textbf{(Knaster-Tarski fixpoint theorem \cite{ktft-tarski})} Let $f : A \rightarrow A$ be a monotone mapping on a complete lattice $\mathfrak{A} = (A, \leq)$. Then $f$ has a least fixpoint $x_{\mu}$ and a greatest fixpoint $x_{\nu}$ determined, respectively, by the pre-fixpoints and post-fixpoints of $f$:
\begin{center}
$x_{\mu} = \bigotimes \{ x \in A \mid f(x) \leq x \}$\\
$x_{\nu} = \bigoplus \{ x \in A \mid x \leq f(x) \}$
\end{center}
where $\bigotimes$ and $\bigoplus$ are the generalisations to arbitrary sets of the operators $\times: {A^{2}} \rightarrow A$ and $+: {A^{2}} \rightarrow A$ on pairs of elements as described before.
}
\eos
\end{thm}

\paragraph*{The Modal Mu-Calculus.} With these concepts in mind we are now ready to present the modal mu-calculus in full as well as some properties of mu-formulae.

\begin{defn}
\emph{
The \emph{modal mu-calculus} (\lmu/ \cite{lmutcs-kozen}) has formulae $\phi$ built from a set $\mathrm{Var}$ of variables $Y,Z,...$ and a set $\Sigma$ of labels $a,b,...$ by the following grammar:
\begin{center}
$\phi ::= Z \mid \phi_1 \wedge \phi_2 \mid \phi_1 \vee \phi_2 \mid \modald{a} \phi_1 \mid \modalb{a} \phi_1 \mid \mu Z . \phi_1 \mid \nu Z . \phi_1$
\end{center}
Also, define the boolean constants as $\mufalse \overset{\underset{\mathrm{def}}{}}{=} \mu Z. Z $ and $\mutrue \overset{\underset{\mathrm{def}}{}}{=} \nu Z . Z$; and assume these abbreviations: $\modald{K} \phi_1$ for $\bigvee_{a \in K} \modald{a} \phi_1$ and $\modalb{K} \phi_1$ for $\bigwedge_{a \in K} \modalb{a} \phi_1$, where $K \subseteq \Sigma$, as well as $\left[ - \right] \phi_1$ for $\left[ \Sigma \right] \phi_1$ and $\left[ - K \right] \phi_1$ for $\left[ {{\Sigma} \setminus {K}} \right] \phi_1$, and similarly for the diamond modality.
}
\eos
\end{defn}

The meaning of the boolean and modal operators is as for HML. The two additional operators of \lmu/, namely $\mu Z . \phi$ and $\nu Z . \phi$ are, respectively, the minimal and maximal fixpoint operators of the logic. The denotations of mu-calculus formulae are given over the set of states of a system as follows:

\begin{defn}
\emph{
A \emph{mu-calculus model} $\mathfrak{M} = (\mathfrak{T},\mathcal{V})$ is an LTS $\mathfrak{T} = (S,T,\Sigma)$ together with a valuation $\mathcal{V} : \mathrm{Var} \rightarrow 2^{S}$. The denotation  $\denot{\phi}^{\mathfrak{T}}_{\mathcal{V}}$ of a formula $\phi$ in the model $\mathfrak{M}$ is a subset of $S$ given as follows (omitting the superscript $\mathfrak{T}$):
\begin{center}
\begin{tabular}{lll}
$\denot{Z}_{\mathcal{V}}$&$=$&$\mathcal{V}(Z)$ \\
$\denot{\phi_1 \wedge \phi_2}_{\mathcal{V}}$&$=$&$\denot{\phi_1}_{\mathcal{V}} \cap \denot{\phi_2}_{\mathcal{V}}$ \\
$\denot{\phi_1 \vee \phi_2}_{\mathcal{V}}$&$=$&$\denot{\phi_1}_{\mathcal{V}} \cup \denot{\phi_2}_{\mathcal{V}}$ \\
$\denot{\modald{a} \phi_1}_{\mathcal{V}}$ & $=$ & $\{ s \in S \mid \exists s' \in S . ~ s \xrightarrow{a} s' \wedge s' \in \denot{\phi_1}_{\mathcal{V}} \}$ \\
$\denot{\modalb{a} \phi_1}_{\mathcal{V}}$ & $=$ & $\{ s \in S \mid \forall s' \in S . ~ s \xrightarrow{a} s' \Rightarrow s' \in \denot{\phi_1}_{\mathcal{V}} \}$ \\
$\denot{\mu Z . \phi}_{\mathcal{V}}$ & $=$ & $\bigcap \{ Q \in 2^{S} \mid \denot{\phi}_{\mathcal{V}\left[Z:=Q\right]} \subseteq Q \}$ \\
$\denot{\nu Z . \phi}_{\mathcal{V}}$ & $=$ & $\bigcup \{ Q \in 2^{S} \mid Q \subseteq \denot{\phi}_{\mathcal{V}\left[Z:=Q\right]} \}$ \\
\end{tabular} 
\end{center}
where $\mathcal{V}\left[Z:=Q\right]$ is the valuation $\mathcal{V'}$ which agrees with $\mathcal{V}$ save that $\mathcal{V'}(Z) = Q$.
}
\eos
\end{defn}

Note that the denotation of the fixpoint operators is given by the Knaster-Tarski fixpoint theorem where $f$ is the mapping ${\denot{\phi}^{\mathfrak{T}}_{\mathcal{V}}}$, the order relation $\leq$ is the subset inclusion relation $\subseteq$, and $\bigotimes$ and $\bigoplus$ are $\bigcap$ and $\bigcup$, respectively.

Also, let us define the `subformulae' of a mu-calculus formula $\phi$; formally, the \emph{subformula set} $Sub(\phi)$ of an \lmu/ formula $\phi$ is given by the \emph{Fischer--Ladner closure} of \lmu/ formulae \cite{phdthesis-lange} in the following way:
\begin{center}
\begin{tabular}{lll}
$Sub(Z)$                     &  $=$ & $\{Z\}$ \\
$Sub(\phi_1 \wedge \phi_2)$  &  $=$ & $\{ \phi_1 \wedge \phi_2 \} \cup Sub(\phi_1) \cup Sub(\phi_2)$ \\
$Sub(\phi_1 \vee \phi_2)$    &  $=$ & $\{ \phi_1 \vee \phi_2 \} \cup Sub(\phi_1) \cup Sub(\phi_2)$ \\
$Sub(\modald{a} \phi_1)$   &  $=$ & $\{\modald{a} \phi_1 \} \cup Sub(\phi_1)$ \\
$Sub(\modalb{a} \phi_1)$   &  $=$ & $\{\modalb{a} \phi_1 \} \cup Sub(\phi_1)$ \\
$Sub(\mu Z . \phi_1)$          &  $=$ & $\{\mu Z . \phi_1 \} \cup Sub(\phi_1)$ \\
$Sub(\nu Z . \phi_1)$          &  $=$ & $\{\nu Z . \phi_1 \} \cup Sub(\phi_1)$ \\
\end{tabular} 
\end{center}

We finish this presentation of the mu-calculus with a note on its expressive power. One of the most interesting features of the mu-calculus is that many interesting temporal logics used for program verification can be embedded into \lmu/. The translation of CTL is straightforward, e.g., as shown in \cite{phdthesis-lange}; other mappings, such as the one for CTL$^{*}$ and thus for LTL as well, are not so simple but still possible \cite{tltolmu-dam}. The source of the immense expressiveness of the mu-calculus comes from the freedom to mix (or alternate) minimal and maximal fixpoint operators arbitrarily. In fact, Bradfield showed that this alternation defines a strict hierarchy \cite{alt-bradfield}, one the most remarkable results regarding the expressivity of \lmu/. These results, amongst many others, have made the mu-calculus one of the most important and studied logics in informatics.

\subsection{Logic Games for Verification}\label{ch2-lgv}
A logic game \cite{lg-jvb} is played by two `players', a ``Verifier" ($\exists$) and a ``Falsifier" ($\forall$), in order to verify the truth or falsity of a given property. In these games the Verifier tries to show that the property holds, whereas the Falsifier wants to refute such an assertion. Solving these games amounts to answering the question of whether the Verifier has a `strategy' to win all plays in the game. Usually the `board' where the game is played is a graph structure in which each position of the board belongs to only one of the two players. Due to this, the games are sequential since at any moment only one of them can play. A play can be of finite or infinite length, and the winner is determined by a set of winning conditions.

There are different questions that can be asked in a verification game. For instance, if a logic formula has at least one model (a satisfiability problem), if a model satisfies a temporal property (a model-checking problem), or whether two systems are equivalent with respect to some notion of equivalence (an equivalence-checking problem). In this paper we are interested in two problems: \emph{bisimulation} and \emph{model-checking} for concurrent systems with partial order semantics. There are some aspects of the games of our interest I should remark.

Traditionally the \emph{players} have been given names depending on the kind of verification game that is being played. For instance, in a bisimulation game the Verifier is called Duplicator whereas the Falsifier is called Spoiler. Similarly, in other kinds of games, the Verifier and Falsifier have been called, respectively, Eloise and Abelard, Player $\exists$ and Player $\forall$, Builder and Critic, Player $\Diamond$ and Player $\Box$, Proponent and Opponent, Eve and Adam, or simply Player I and Player II. In order to have a uniform notation, we choose to call them ``Eve'' and ``Adam'' regardless of the game they play. 

The \emph{boards} where the games are played also have different structures depending on the kind of verification problem that one wants to solve. In a bisimulation game the board is made up with the elements of the two systems that are being analysed, e.g., each position in the board is an element of the Cartesian product of the state sets of the two systems. On the other hand, in a model-checking game the board is composed of pairs of elements where one of the components is an element of the model being checked and the other component relates to the temporal property under consideration. These game features are formally defined whenever new bisimulation or model-checking games are presented.

Finally, notice that by playing a logic game the two players jointly define sequences of positions of the game board. Such sequences are called plays of the game. Let $\Gamma$ be the set of plays of a game and $\mathfrak{B}$ be a game board, i.e., a set of positions in the game. Then, a deterministic \emph{strategy} is a function $\lambda: {\Gamma \rightarrow \mathfrak{B}}$ from plays to positions of the game board, so that such strategies define the next move a player makes. But in some cases, in order for a player to make a move he or she only needs to know their current position. In these cases, their strategies can be defined as functions on the set of positions of the board, rather then on the set of plays of the game. These strategies are called `history-free'---positional or memoryless. Formally, a history-free strategy is a function $\lambda : \mathfrak{B} \rightarrow \mathfrak{B}$. Finally, a \emph{winning} strategy is a strategy that guarantees that the player that uses it can win all plays of the game. Here, we only deal with history-free winning strategies.

\subsubsection*{Bisimulation Games.}
Bisimulation games are formal and interactive characterisations of a family of equivalence relations called bisimulation equivalences. One of the simplest bisimulation equivalences is `bisimilarity', the equivalence relation induced by modal logic. This equivalence was defined, independently, by Johan van Benthem \cite{jvbthesis-jvb} while studying the semantics of modal logic, and a few years later by Milner and Park \cite{ccs80-milner,bis-park} while studying the behaviour of concurrent systems with interleaving semantics.

More precisely, a bisimulation game $\mathcal{G}(\mathfrak{T}_1,\mathfrak{T}_2)$ is a formal representation of a bisimulation equivalence $\sim_{eq}$ between two systems $\mathfrak{T}_1$ and $\mathfrak{T}_2$. Whereas Eve believes that $\mathfrak{T}_1 \sim_{eq} \mathfrak{T}_2$, Adam wants to show that $\mathfrak{T}_1 \not \sim_{eq} \mathfrak{T}_2$. All plays start in the initial position $(s_0,q_0)$ consisting of the initial states of the systems, and the players take alternating turns---although Adam always plays first and chooses where to play. Thus, in every round of the game Adam makes the first move in either system according to a set of rules, and then Eve must make a corresponding $\sim_{eq}$-equivalent move on the other system; the game can proceed in this way indefinitely. Thus, the plays of the game can be of finite or infinite length. All plays of infinite length are winning for Eve; in the case of plays of finite length, the player who cannot make a move loses the game. These winning conditions apply to all the bisimulation games we study here.

In concurrency, bisimulation games are often used to show that two concurrent systems interact equivalently (with respect to $\sim_{eq}$) with an arbitrary environment. Since the exact definition of a particular bisimulation equivalence $\sim_{eq}$ can be altered (strengthened or weakened) by the kinds of properties that one wants to analyse, then the set of rules for playing a bisimulation game can be different in each game. The best known bisimulation game for interleaving concurrency is the one that characterises `strong' bisimilarity (sb \cite{hmljacm-milner}), the bisimulation equivalence induced by HML.

However, in order to capture properties of partial order models rather than of interleaving ones, equivalences finer than strong bisimilarity have been defined as well as their associated bisimulation games. Two of the most important bisimulation games for partial order models are the ones that characterise `history-preserving' bisimilarity (hpb \cite{hpb-rav}) and `hereditary history-preserving' bisimilarity (hhpb \cite{open-nielsen}). Both (history-preserving) bisimulation equivalences, together with a deep study of their applications to concurrency, can be found in \cite{hpb-glabbeek}. Let us now introduce some concepts needed to present the bisimulation games that characterise sb, hpb, and hhpb.

\paragraph{Strong Bisimulation Games.}
A bisimulation game for strong bisimilarity is played on a board $\mathfrak{B}$ composed of pairs $(s,q)$ of states $s$ and $q$ of two systems $\mathfrak{T}_1$ and $\mathfrak{T}_2$, respectively. Such a pair is a position of the board $\mathfrak{B}$ and is called a `configuration' of the game. The position $(s_0,q_0)$, where $s_0$ and $q_0$ are the initial states of $\mathfrak{T}_1$ and $\mathfrak{T}_2$, is the initial configuration. Since the strategies of the game are history-free, then a strategy $\lambda$ is a partial function on ${\mathfrak{B}} \subseteq {S \times Q}$, where $S$ and $Q$ are the state sets of the two systems $\mathfrak{T}_1$ and $\mathfrak{T}_2$, respectively.

\begin{notn}
\emph{
Since a system has only one initial state, a bisimulation game can be unambiguously presented as either $\mathcal{G}(\mathfrak{T}_1,\mathfrak{T}_2)$ or $\mathcal{G}(s_0,q_0)$ if the two systems are obvious from the context. Also, since bisimulation games are symmetric, we omit the subscript in $\mathfrak{T}$ whenever referring to either system.
}
\eos
\end{notn}

\begin{defn}\label{ch2-sbgame}
\emph{
\textbf{(Strong bisimulation games)} Let ${(s,q)}$ be a configuration of the game $\mathcal{G}(\mathfrak{T}_1,\mathfrak{T}_2)$. There are two players, Adam and Eve, and Adam always plays first and chooses where to play. The equivalence relation $R_{sb}$ is a strong bisimulation, $\sim_{sb}$, between $\mathfrak{T}_1$ and $\mathfrak{T}_2$ if:
\begin{itemize}
\item (Base case) The initial configuration $(s_0,q_0)$ is in $R_{sb}$.
\item ($\sim_{sb}$ rule) If ${(s,q)}$ is in $R_{sb}$ and Adam chooses a transition in $\mathfrak{T}$, say a transition $s \xrightarrow{a} s'$ of $\mathfrak{T}_1$, then Eve must choose a transition in the other system (any $q \xrightarrow{a} q'$ of $\mathfrak{T}_2$ in this case), such that the new configuration $(s',q')$ is in $R_{sb}$ as well.
\end{itemize}
$\mathfrak{T}_1 \sim_{sb} \mathfrak{T}_2$ iff Eve has a winning strategy for the sb game  $\mathcal{G}(\mathfrak{T}_1,\mathfrak{T}_2)$.}
\eos
\end{defn}

This bisimulation game do not capture any information of partial order models that is not already present in their interleaving counterparts. For this reason, games for strong bisimilarity are considered games for interleaving concurrency rather than for partial order concurrency. In order to capture properties of partial order models, one has to recognise at least when two transitions of a system are independent and hence executable in parallel. This feature is captured by the following finer games.

\paragraph{History-Preserving Bisimulation Games.}
A game for history-preserving bisimilarity is a bisimulation game as presented before with a further `synchronisation' requirement on transitions. Such a synchronisation requirement makes the selection of transitions by Eve more restricted. Let us first define this notion of synchronisation on transitions before making a formal presentation of the game.

A possibly empty sequence of transitions $\pi = {\left[t_1,...,t_k \right]}$ is a \emph{run} of a system $\mathfrak{T}$. Let $\Pi_{\mathfrak{T}}$ be the set of runs of $\mathfrak{T}$ and $\varrho(\pi)$ be the last transition of $\pi$. Define $\epsilon = \varrho(\left[ ~ \right])$ and $s_0 = \sigma(\epsilon) = \tau(\epsilon)$, for an empty sequence $\left[ ~ \right]$. Given a run $\pi$ and a transition $t$, the sequence $\pi . t$ denotes the run $\pi$ extended with $t$. Let $\pi_1 \in {\Pi_{\mathfrak{T}_1}}$ and $\pi_2 \in {\Pi_{\mathfrak{T}_2}}$ for two systems $\mathfrak{T}_1$ and $\mathfrak{T}_2$. We say that the pair of runs ($\pi_1 . u , \pi_2 . v$) is \emph{synchronous} iff ${{({\varrho(\pi_1)},u)} \in I_1} \Leftrightarrow {{({\varrho(\pi_2)},v)} \in I_2}$, where $I_1$ and $I_2$ are the independence relations of $\mathfrak{T}_1$ and $\mathfrak{T}_2$, and the posets induced by $\pi_1 . u$ with $I_1$ and $\pi_2 . v$ with $I_2$ are isomorphic.\footnote{Given a run $\pi$ and an independence relation $I$, there is a poset $(E,\leq_{E})$ induced by $\pi$ with $I$ such that $E$ has as elements the event occurrences associated with the transitions in $\pi$ and where the partial order relation $\leq_{E}$ is defined by the event structure unfolding of the system whose independence relation is $I$.} By definition $(\epsilon,\epsilon)$ is synchronous. As it is more convenient to define hpb games on pairs of runs rather than on pairs of states, a configuration of the game will be a pair of runs. 

\begin{defn}\label{ch2-hpbgame}
\emph{
\textbf{(History-preserving bisimulation games)} Let $(\pi_1 , \pi_2)$ be a configuration of the game $\mathcal{G}(\mathfrak{T}_1,\mathfrak{T}_2)$. The initial configuration of the game is $(\epsilon,\epsilon)$. The relation $R_{hpb}$ is a history-preserving (hp) bisimulation, $\sim_{hhpb}$, between $\mathfrak{T}_1$ and $\mathfrak{T}_2$ iff it is a strong bisimulation relation between $\mathfrak{T}_1$ and $\mathfrak{T}_2$ and:
\begin{itemize}
\item (Base case) The initial configuration $(\epsilon,\epsilon)$ is in $R_{hpb}$.
\item ($\sim_{hpb}$ rule) If $(\pi_1,\pi_2)$ is in $R_{hpb}$ and Adam chooses a transition $u$ in either system, say in $\mathfrak{T}_1$, such that $u = \tau(\varrho(\pi_1)) \xrightarrow{a} s'$, then Eve must choose a transition $v$ in the other system such that $v = \tau(\varrho(\pi_2)) \xrightarrow{a} q'$ and the new configuration $(\pi_1 . u , \pi_2 . v)$ is synchronous, i.e., ${(\pi_1 . u , \pi_2 . v)}$ is in $R_{hpb}$ as well.
\end{itemize}
$\mathfrak{T}_1 \sim_{hpb} \mathfrak{T}_2$ iff Eve has a winning strategy for the hpb game  $\mathcal{G}(\mathfrak{T}_1,\mathfrak{T}_2)$.}
\eos
\end{defn}

\paragraph{Hereditary History-Preserving Bisimulation Games.}
A bisimulation game for hereditary history-preserving bisimilarity is an hpb game extended with backtracking moves. These backtracking moves are restricted to transitions that are `backwards enabled'. More specifically, let $\pi(i)$ be the $i$-th transition in $\pi$. Given a run $\pi = \left[ t_1,...,t_k \right] $, a transition $\pi(i)$ is backwards enabled if, and only if, it is independent of all transitions $t_j$ that appear after it in $\pi$, i.e., iff ${\forall t_j} \in {\{\pi(i+1),...,\pi(k) \} . ~{{\pi(i)} \I {t_j}}} $. 

This definition captures the fact that backwards enabled transitions are the terminal elements of the partial order induced by the independence relation $I$ on the transitions in $\pi$. Now, let $\pi - \pi(i)$ be the sequence of transitions $\pi$ without its $i$-th element $\pi(i)$. It should be clear that if $\pi(i)$ is backwards enabled, then the partial order induced by $I$ on those transitions in $\pi - \pi(i)$ is just the same partial order induced by $I$ on $\pi$ without the terminal element or transition $\pi(i)$. Formally, an hhpb game is defined as follows:

\begin{defn}\label{ch2-hhpbgame}
\emph{
\textbf{(Hereditary history-preserving bisimulation games)} Let $(\pi_1 , \pi_2)$ be a configuration of the game $\mathcal{G}(\mathfrak{T}_1,\mathfrak{T}_2)$. The initial configuration of the game is $(\epsilon,\epsilon)$. The equivalence relation $R_{hhpb}$ is a hereditary history-preserving (hhp) bisimulation, $\sim_{hhpb}$, between $\mathfrak{T}_1$ and $\mathfrak{T}_2$ iff it is a history-preserving bisimulation between $\mathfrak{T}_1$ and $\mathfrak{T}_2$ and:
\begin{itemize}
\item (Base case) The initial configuration $(\epsilon,\epsilon)$ is in $R_{hhpb}$.
\item ($\sim_{hhpb}$ rule) If $(\pi_1,\pi_2)$ is in $R_{hhpb}$ and Adam deletes, say from $\pi_1$, a transition $\pi_1(i)$ that is backwards enabled, then Eve must delete the transition $\pi_2(i)$ from the history of the game in the other system, provided that $\pi_2(i)$ is also backwards enabled and that the new configuration $(\pi_1 - \pi_1(i) , \pi_2 - \pi_2(i))$ is in $R_{hhpb}$ as well.
\end{itemize}
$\mathfrak{T}_1 \sim_{hhpb} \mathfrak{T}_2$ iff Eve has a winning strategy for the hhpb game  $\mathcal{G}(\mathfrak{T}_1,\mathfrak{T}_2)$.
}
\eos
\end{defn}

Unlike the game for $\sim_{sb}$, which is a game for interleaving concurrency, both history-preserving games presented here can capture properties of partial order models and differentiate them from their interleaving counterparts. The simplest example is the case of two processes $a \parallel b$ and $a.b + b.a$, which are equivalent from an interleaving viewpoint, but different if considering partial order semantics.

\subsubsection*{Model-Checking Games.}
Model-checking games \cite{mcgames-gradel,games-wal}, also called Hintikka evaluation games, are logic games played in a formula $\phi$ and a mathematical model $\mathfrak{M}$. In a game $\mathcal{G}(\mathfrak{M},\phi)$ the goal of Eve is to show that $\mathfrak{M} \models \phi$, while Adam believes that $\mathfrak{M} \not \models \phi$. In program verification, most usually $\phi$ is a modal or a temporal formula and $\mathfrak{M}$ is a Kripke structure or an LTS, and the two players play the game $\mathcal{G}(\mathfrak{M},\phi)$ by picking single elements of $\mathfrak{M}$, according to the game rules defined by $\phi$. For now, let us consider model-checking games played on interleaving models and on formulae given as mu-calculus specifications.

The game we are about to describe is the local model-checking procedure for the mu-calculus as defined by Stirling \cite{localmc-stirling}. It is a game interpretation of the tableau technique for mu-calculus model-checking introduced by Stirling and Walker \cite{lmctab-stirling}. Although the game is naturally played on interleaving models of concurrency, it can also be used to model-check partial order models, such as Petri nets, if one considers their one-step interleaving semantics, e.g., as in \cite{localtcs-bradfield}.

\paragraph{Local Model-Checking Games in the Mu-Calculus.}
A local model-checking game $\mathcal{G}(\mathfrak{M},\phi)$ is
played on a mu-calculus model $\mathfrak{M} = (\mathfrak{T},\mathcal{V})$, where
$\mathfrak{T} = (S,s_0,T,\Sigma)$ is an interleaving system, and on a mu-calculus formula
$\phi$. Since the game is local, this is, it answers to the question of whether the initial state $s_0$ satisfies $\phi$, then it can also be presented as $\mathcal{G}_{\mathfrak{M}}(s_0,\phi)$, or even as $\mathcal{G}(s_0,\phi)$ whenever the model $\mathfrak{M}$ is clear from the context. The board in which the game is played has the form
$\mathfrak{B} \subseteq S \times Sub(\phi)$, where $Sub (\phi)$ is the set of subformulae of a mu-calculus formula $\phi$ as defined by the Fischer--Ladner closure of mu-calculus formulae.

A play is a possibly infinite sequence of configurations $C_0, C_1, ...$; each $C_i = (s,\psi)$ is an element of the board $\mathfrak{B}$. i.e., it is a position of the game. Every play starts in $C_0 = (s_0,\phi)$, and proceeds according to the rules of the game, given below. Two deterministic rules control the unfolding of fixpoint operators. Moreover, given a configuration $(s,\phi)$, the rules for $\vee$ and $\wedge$ make, respectively, Eve and Adam choose a next configuration $(s,\psi)$ which is determined by the subformula set of $\phi$. Similarly, the rules for $\modald{~}$ and $\modalb{~}$ make, respectively, Eve and Adam choose a next configuration $(q,\psi)$ which is determined by those transitions $t$ such that $s = \sigma(t)$ and $q = \tau(t)$. These conditions can be captured in the following way. Let $(s,\phi)$ be the current configuration of the game; the next configuration of the game is defined by the following  game rules:
\begin{itemize}
\item if $\phi = \mu Z . \varphi$ (resp. $\phi = \nu Z . \varphi$), then Eve (resp. Adam) replaces $\mu Z . \varphi$ (resp. $\nu Z . \varphi$) by its associated variable $Z$ and the next configuration is $(s,Z)$.
\item if $\phi = Z$ such that $\psi = \mu Z . \varphi$ (resp. $\psi = \nu Z . \varphi$) for some formula $\psi$, then Eve (resp. Adam) unfolds the fixpoint and the next configuration is $(s,\varphi)$.
\item if $\phi = \psi_1 \vee \psi_2$ (resp. $\phi = \psi_1 \wedge \psi_2$), then Eve (resp. Adam) chooses some $\psi_i$, for $i \in \{1,2\}$, and the next configuration is $(s,\psi_i)$.
\item if $\phi = \modald{a} \psi$ (resp. $\phi = \modalb{a} \psi$), then Eve (resp. Adam) chooses a transition $s \xrightarrow{a} s'$ and the next configuration is $(s',\psi)$.
\end{itemize}

Finally the following rules are the winning conditions that determine a unique winner for every finite or infinite play $C_0, C_1, ... $ in a game $\mathcal{G}(s_0,\phi)$. Adam wins a finite play $C_0, C_1, ..., C_k$ or an infinite play $C_0, C_1, ...$ iff:
\begin{enumerate}
\item $C_k = (s,Z)$ and $s \not \in \mathcal{V}(Z)$.
\item $C_k = (s,\modald{a} \psi)$ and $\{{s'} \mid {s \xrightarrow{a} s'} \} = \emptyset$.
\item The play is of infinite length and there exists a mu-calculus formula $Z$ which is both the least fixpoint of some subformula $\mu Z . \psi$ and the syntactically outermost variable in $\phi$ that occurs infinitely often in the game.
\end{enumerate}

Dually, Eve wins a finite play $C_0, C_1, ..., C_n$ or an infinite play $C_0, C_1, ...$ iff:
\begin{enumerate}
\item $C_k = (s,Z)$ and $s \in \mathcal{V}(Z)$.
\item $C_k = (s,\modalb{a} \psi)$ and $\{{s'} \mid {s \xrightarrow{a} s'} \} = \emptyset$.
\item The play is of infinite length and there exists a mu-calculus formula $Z$ which is both the greatest fixpoint of some subformula $\nu Z . \psi$ and the syntactically outermost variable in $\phi$ that occurs infinitely often in the game.
\end{enumerate}

\noindent Then $s_0 \models \phi$ iff Eve has a winning strategy in the model-checking game $\mathcal{G}(s_0,\phi)$. 

\section{Mu-Calculi with Partial Order Semantics}\label{mucalculi}

In this section we study the underlying mathematical properties of the partial order models of concurrency presented before, and show that the behaviour of these systems can be captured in a uniform way by two simple and general dualities of local behaviour. We use these dualities of local behaviour to define a number of mu-calculi, or fixpoint modal logics, with partial order semantics. This work delivers a logical approach to defining a notion of equivalence for concurrency tailored to be abstract or model independent, setting the grounds for a logic-based framework for studying different models for concurrency uniformly.

\subsection{Local Dualities in Partial Order Models} \label{dual}

We present two ways in which concurrency can be regarded as a dual concept to \emph{conflict} and \emph{causality}, respectively. These two ways of observing concurrency will be called \emph{immediate concurrency} and \emph{linearised concurrency}. Whereas immediate concurrency is dual to conflict, linearised concurrency is dual to causality. These local dualities were first defined in \cite{fos09-gut}.

The intuitions behind these two observations are the following. Consider a concurrent system and any two different transitions $t$ and $t'$ with the same source node, i.e., $\sigma(t) =  \sigma(t')$. These two transitions are either immediately concurrent, and therefore independent, i.e., $(t,t') \in I$, or dependent, in which case they must be in conflict. Similarly, consider any two transitions $t$ and $t'$ where $\tau(t) = \sigma(t')$. Again, the pair of transitions $(t,t')$ can either belong to $I$, in which case the two transitions are concurrent, yet have been linearised, or the pair does not belong to $I$, and therefore the two transitions are causally dependent. In both cases, the two conditions are exclusive and there are no other possibilities.

The local dualities just described are formally defined in the following way:
\begin{center}
\begin{tabular}{lcl}
$\otimes$ & $\overset{\underset{\mathrm{def}}{}}{=} $ &$\{(t,t') \in T \times T \mid \sigma(t) = \sigma(t') \wedge t \I t' \}$ \\
$\#$ & $\overset{\underset{\mathrm{def}}{}}{=} $ &$\{ (t,t') \in T \times T \mid \sigma(t) = \sigma(t') \wedge \neg ( t \I t' ) \}$ \\
$\ominus$ & $\overset{\underset{\mathrm{def}}{}}{=} $ &$\{ (t,t') \in T \times T \mid \tau(t) = \sigma(t') \wedge t \I t' \}$ \\
$\leq$ & $\overset{\underset{\mathrm{def}}{}}{=} $ &$\{ (t,t') \in T \times T \mid \tau(t) = \sigma(t') \wedge \neg ( t \I t' ) \}$ \\
\end{tabular} 
\end{center}

Notice the dual conditions between $\otimes$ and $\#$ and between $\ominus$ and $\leq$ with respect to the independence relation, if assuming valid the locality requirement.

\begin{defn}
\emph{
Let $t$ and $t'$ be two transitions. We say that $t$ and $t'$ are \emph{immediately concurrent} iff $(t,t') \in \otimes$, \emph{in conflict} iff $(t,t') \in \#$, \emph{linearly concurrent} iff $(t,t') \in \ominus$, or \emph{causally dependent} iff ${(t,t')} \in {\leq}$.
}
\eos
\end{defn}

\subsubsection*{Sets in a Local Context.}
The relation $\otimes$ defined on pairs of transitions, can be used to recognise \emph{sets} where every transition is independent of each other and hence can all be executed concurrently. Such sets are said to be \emph{conflict-free} and belong to the same `trace'. 

\begin{defn}\label{cfsets}
\emph{A \emph{conflict-free} set of transitions $P$ is a set of transitions with the same source node, where $t \otimes t'$ for each two elements in $P$.}
\eos
\end{defn}

Notice that by definition empty sets and singleton sets are 
trivially conflict-free. Given a system $\mathfrak{T}$, all conflict-free
sets of transitions at a state $s$ can be defined locally from
the \emph{maximal set} of transitions $\mathfrak{X}(s)$, where
$\mathfrak{X}(s)$ is the set of all transitions $t$ such that
$\sigma(t) = s$. We simply write $\mathfrak{X}$ when the state $s$ is 
defined elsewhere or is implicit from the context. 
Moreover, all maximal sets and conflict-free sets of
transitions are fixed given a particular system $\mathfrak{T}$. 

\begin{defn}
\emph{
Given a system $\mathfrak{T}$, a \emph{support} set $R$ in $\mathfrak{T}$ is either a maximal set of transitions in $\mathfrak{T}$ or a non-empty conflict-free set of transitions in $\mathfrak{T}$.
}
\eos
\end{defn}

Given a system $\mathfrak{T}$, the set of all its support sets is denoted by $\mathfrak{P}$. As can be seen from the definition, support sets can be of two kinds, and one of them provide us with a way of doing local reasoning. More precisely, doing local reasoning on sets of independent transitions becomes possible when considering conflict-free sets since they can be decomposed into smaller sets, where every transition is, as well, independent of each other. Using standard notation on sets, we write $P_1 \uplus P_2$ to denote the disjoint union of sets $P_1$ and $P_2$. If both $P_1$ and $P_2$ are support sets then we have that $P_1 \neq \emptyset$ and $P_2 \neq \emptyset$, and hence $P_1 \uplus P_2 \neq \emptyset$.

\begin{defn}
\emph{
Given a support set $R$, a \emph{complete trace} $M$ of $R$, denoted by $M \sqsubseteq R$, is a support set $M \subseteq R$ such that $\neg \exists t \in {R \setminus M} . ~ \forall t' \in M . ~ t \otimes t'$.
}
\eos
\end{defn}

Note that if $R$ is a conflict-free support set, then $M=R$. Otherwise, $R$ necessarily is a maximal set $\mathfrak{X}$ and $M$ must be a proper subset of $R$. Therefore, if $R = \mathfrak{X}$, then the sets $M$ such that $M \sqsubseteq \mathfrak{X}$ are the biggest conflict-free support sets that can be recognised in a particular state $s$ of a system $\mathfrak{T}$; we call them \emph{maximal traces}. Since all complete and maximal traces are support sets, then they are also fixed and computable given a particular system $\mathfrak{T}$.

\subsection{Fixpoint Logics with Partial Order Models} \label{synsem}

The local dualities and sets defined in the previous section can be used to build the semantics of a number of fixpoint modal logics which capture that behaviour of partial order models that is not present in interleaving ones. As a consequence, these logics may be more adequate languages for expressing properties of partial order systems such as Petri nets, event structures, or TSI.

The naturality of these logics is reflected by the bisimulation equivalences they induce, since in several cases they either coincide with standard bisimilarities for concurrency, e.g., with sb or with hpb, or have better decidability properties than other already known bisimulation equivalences for partial order models, e.g., with respect to hhpb, which is undecidable even in finite systems. 

The semantics of the logics we define here are based on the recognition of the dualities that can be defined in a partial order model for concurrency. The logic we introduce here is called Trace Fixpoint Logic (\tfl/). As defined by its semantics, \tfl/ captures the duality between concurrency and causality by refining the usual modal operator of \lmu/. On the other hand, the duality between concurrency and conflict is captured by a second-order modality that recognises maximal traces in the system. Such a modality enjoys beautiful mathematical properties; in particular, not only it is a monotonic, but also an \emph{idempotent} operator, which informally means that it delivers as much information as possible whenever used. \tfl/ is a more simple, purely modal logic for reasoning about partial order systems alternative to the fixpoint logic introduced in \cite{fos09-gut}.

\subsubsection*{Process Spaces.} In order to define the semantics of \tfl/ we construct an intermediate structure into which any of the systems we consider here can be mapped. Such a structure determines a `space of processes', which are simple abstract entities representing pieces of isolated (i.e., local and independent) behaviour.

\begin{defn}
\emph{
Let $\mathfrak{T} = (S,s_0,T,\Sigma,I)$ be a system. A \emph{Process Space} $\mathbb{S}$ is the lattice $\mathfrak{P} \times \mathfrak{A}$, where $\mathfrak{P}$ is the set of support sets of $\mathfrak{T}$ and $\mathfrak{A}$ is the set of transitions $T \cup \{t_{\epsilon}\}$, such that $t_{\epsilon}$ is the empty transition satisfying that for all $t \in T$, if $s_0 = \sigma(t)$ then $t_{\epsilon} \leq t$. A tuple $(R,t) \in \mathbb{S}$ is called a process, and the initial process of $\mathbb{S}$ is the tuple $(\mathfrak{X}_0,t_{\epsilon})$, where $\mathfrak{X}_0 = \mathfrak{X}(s_0)$. 
}
\eos
\end{defn}

Notice that for any process it is always possible to infer the particular state in $\mathfrak{T}$ to which such a process relates. Since a process does not represent explicitly the states of a systems we say that a process space is `stateless'. Also, let $\mathcal{X}$ be the subset of $\mathfrak{P}$ that contains only maximal sets $\mathfrak{X}$ and maximal traces $M$. Call $\mathfrak{S} = \mathcal{X} \times \mathfrak{A}$ a \emph{stateless maximal} process space.

\subsubsection*{Trace Fixpoint Logic.} Having defined local dualities in partial order models and a process space upon them, we are now ready to present a a modal logic that is sensitive to causal information and that allows for reasoning on the traces of a concurrent system with a partial order semantics.

\begin{defn}
\emph{
\emph{Trace Fixpoint Logic} (\tfl/) has formulae $\phi$ built from a set $\mathrm{Var}$ of variables $Y,Z,...$ and a set $\Sigma$ of labels $a,b,...$ by the following grammar:
\begin{center}
$\phi ::= Z \mid \neg \phi_1 \mid \phi_1 \wedge \phi_2 \mid \modald{a}_{c} \phi_1 \mid \modald{a}_{nc} \phi_1 \mid \modald{\otimes} \phi_1 \mid \mu Z . \phi_1$
\end{center}
where $Z \in \mathrm{Var}$ and $\mu Z . \phi_1$ has the restriction that any free occurrence of $Z$ in $\phi_1$ must be within the scope of an even number of negations. 
Dual boolean, modal, and fixpoint operators are defined in the usual way: 
\begin{center}
\begin{tabular}{lcl}
$\phi_1 \vee \phi_2 $&$\overset{\underset{\mathrm{def}}{}}{=}$&$ \neg (\neg \phi_1 \wedge \neg \phi_2)$\\
$\modalb{a}_{c} \phi_1 $&$\overset{\underset{\mathrm{def}}{}}{=}$&$ \neg \modald{a}_{c} \neg \phi_1 $\\
$\modalb{a}_{nc} \phi_1 $&$\overset{\underset{\mathrm{def}}{}}{=}$&$ \neg \modald{a}_{nc} \neg \phi_1 $\\
$\modalb{\otimes} \phi_1 $&$\overset{\underset{\mathrm{def}}{}}{=}$&$ \neg \modald{\otimes} \neg \phi_1$\\
$\nu Z . \phi_1 $&$\overset{\underset{\mathrm{def}}{}}{=}$&$ \neg \mu Z . \neg \phi_1 \left[ \neg Z / Z \right]$
\end{tabular}
\end{center}
Boolean constants and other abbreviations are defined as for \lmu/. Moreover, `plain' modalities, i.e., HML modalities, can be represented as follows:
\begin{center}
\begin{tabular}{lcl}
$\modald{a} \phi_1 $&$\overset{\underset{\mathrm{def}}{}}{=}$&$ \modald{a}_{c} \phi_1 \vee \modald{a}_{nc} \phi_1$ \\ $\modalb{a} \phi_1 $&$\overset{\underset{\mathrm{def}}{}}{=}$&$ \modalb{a}_{c} \phi_1 \wedge \modalb{a}_{nc} \phi_1$
\end{tabular}
\end{center}
}
\eos
\end{defn}

We say that a formula is in `positive form' if negations are applied only to variables. Any formula built with the language given above, together with the dual operators, can be converted into positive form; it is moreover in `positive normal form' if there are no clashes of bound variables. Again, any formula can be converted into an equivalent one in positive normal form. Then, without loss of generality, hereafter we only consider formulae in positive normal form.

\begin{defn}\label{semtfl}
\emph{
A \emph{\tfl/ model} $\mathfrak{M}$ is a system $\mathfrak{T} = (S,s_0,T,I,\Sigma)$ together with a valuation $\mathcal{V} : \mathrm{Var} \rightarrow 2^{\mathfrak{S}}$, where $\mathfrak{S} = \mathcal{X} \times \mathfrak{A}$ is the stateless maximal process space associated with $\mathfrak{T}$. The denotation $\denot{\phi}^{\mathfrak{T}}_{\mathcal{V}}$ of a formula $\phi$ in the model $\mathfrak{M} = (\mathfrak{T},\mathcal{V})$ is a subset of $\mathfrak{S}$, given by the following rules (omitting the superscript $\mathfrak{T}$):
\begin{center}
\begin{tabular}{lll}
$\denot{Z}_{\mathcal{V}}$&$=$&$\mathcal{V}(Z)$ \\
$\denot{\neg \phi_1}_{\mathcal{V}}$&$=$&$\mathfrak{S} - \denot{\phi_1}_{\mathcal{V}}$\\
$\denot{\phi_1 \wedge \phi_2}_{\mathcal{V}}$&$=$&$\denot{\phi_1}_{\mathcal{V}} \cap \denot{\phi_2}_{\mathcal{V}}$ \\
$\denot{\modald{a}_{c} \phi_1}_{\mathcal{V}}$ & $=$ & $\{ (R,t) \in \mathfrak{S} \mid \exists r \in R . ~ t \leq r \wedge (\mathfrak{X},r) \in \denot{\phi_1}_{\mathcal{V}} \}$ \\
$\denot{\modald{a}_{nc} \phi_1}_{\mathcal{V}}$ & $=$ & $\{ (R,t) \in \mathfrak{S} \mid \exists r \in R . ~ t \ominus r \wedge (\mathfrak{X},r) \in \denot{\phi_1}_{\mathcal{V}} \} $ \\
$\denot{\modald{\otimes} \phi_1}_{\mathcal{V}}$ &    $=$ & $\{ (R,t) \in \mathfrak{S} \mid \exists M \in \mathcal{X} . ~ M \sqsubseteq R \wedge (M,t) \in \denot{\phi_1}_{\mathcal{V}} \}$ \\
\end{tabular} 
\end{center}
such that $a = \delta(r)$ and $\mathfrak{X}$ is the maximal set at $\tau(r)$. Also, given the usual restriction on free occurrences of variables imposed in order to obtain monotone operators in the complete lattice $(2^{\mathfrak{S}},\subseteq)$, the powerset lattice of $\mathfrak{S}$, it is possible to define the denotation of the least fixpoint operator in the standard way according to the Knaster-Tarski fixpoint theorem: 
\begin{center}
\begin{tabular}{lll}
$\denot{\mu Z . \phi}_{\mathcal{V}}$ & $=$ & $\bigcap \{ Q \subseteq \mathfrak{S} \mid \denot{\phi}_{\mathcal{V}\left[Z:=Q\right]} \subseteq Q \}$ \\
\end{tabular} 
\end{center}
where $\mathcal{V}\left[Z:=Q\right]$ is the valuation $\mathcal{V'}$ which agrees with $\mathcal{V}$ save that $\mathcal{V'}(Z) = Q$. Since positive normal form is assumed henceforth, the semantics of the dual boolean, modal, and fixpoint operators can be given in the usual way. Finally, the \emph{satisfaction} relation $\models$ is defined in the usual way: given a process $P$ and a formula $\phi$, we have that $P \models \phi $ iff $P \in \denot{\phi}$.
}
\eos
\end{defn}

Informally, the meaning of the basic \tfl/ operators is the following: boolean constants and operators are interpreted in the usual sense; the semantics of the `causal' diamond modality $\modald{a}_c \phi_1$ (resp. of the `non-causal' diamond modality $\modald{a}_{nc} \phi_1$) is that a process $(R,t)$ satisfies $\modald{a}_c \phi_1$  (resp. $\modald{a}_{nc} \phi_1$) if it can perform an $a$-labelled action $r$ that causally depends on $t$ (resp. that is independent of $t$) and move through $r$ into a process where $\phi_1$ holds; and dually for the causal and non-causal box modalities $\modalb{a}_c \phi_1$ and $\modalb{a}_{nc} \phi_1$. The modality $\modald{\otimes} \phi_1$ provides local second-order power on conflict-free sets of transitions, i.e., on sets of independent transitions. This modality allows one to restrict, locally, the behaviour of a system to those executions that can actually happen concurrently at a given state. Finally, the meaning of the fixpoint operators is as for \lmu/.

\begin{prop}\label{idem}
$\modald{\otimes}$ is an idempotent operator.
\end{prop}
\begin{proof}
Let $H = \denot{\modald{\otimes}\phi}$ and $G = \denot{\phi}$. $G$ can be split into two disjoint sets of stateless maximal processes $G^{\otimes} \uplus G^{\#}$ (called simply processes in the sequel), where the former is the set of processes in $G$ whose support sets are conflict-free, and the latter those processes whose support sets are not, i.e., $G \setminus G^{\otimes}$. Similarly, $H$ can be represented as the disjoint union of sets of processes $H^{\otimes}$ and $H^{\#}$. 

Notice that $H^{\otimes} = G^{\otimes}$ because for any process $P_{H^{\otimes}} = (R,t)$ in $H^{\otimes}$ there is a process $P_{G^{\otimes}} = (R,t)$ in $G^{\otimes}$, since $R \sqsubseteq R$ for any conflict-free support set $R$. However, this is clearly not the case for the processes in $G^{\#}$ and $H^{\#}$, because there may be a process $P_{G^{\#}}$ in $G^{\#}$ (whose support set is necessarily maximal and not conflict-free) such that there is no process $P_{G^{\otimes}}$ in $G^{\otimes}$ to which the support set of $P_{G^{\#}}$ can be related using $\sqsubseteq$. Therefore, whereas $G^{\#}$ would contain such a processes, $H^{\#}$ would not. Similarly, there may be new processes in $H^{\#}$ whose support sets can be related to support sets of processes in $G^{\otimes}$ (and of course in $H^{\otimes}$ as well) using $\sqsubseteq$, but that were not in $G^{\#}$.

Now, let $F = F^{\otimes} \uplus F^{\#} = \denot{\modald{\otimes} \modald{\otimes} \phi}$. For the same reason as before, $F^{\otimes} = H^{\otimes}$. However, in this case $F^{\#} = H^{\#}$ since now for every process in both $F^{\#}$ and $H^{\#}$, there must be a process in $H^{\otimes}$ (and of course in $F^{\otimes}$) to which their support sets can be related using $\sqsubseteq$. So, since applying $\modald{\otimes}$ only once always leads to a fixpoint, then $\modald{\otimes}$ is an idempotent operator. 
\qed
\end{proof}

\begin{fact}
$\modald{\otimes}$ is not an extensive operator.
\end{fact}
\begin{proof}
Let $H = H^{\otimes} \uplus H^{\#} = \denot{\modald{\otimes}\phi}$ and $G = G^{\otimes} \uplus G^{\#} = \denot{\phi}$, where $H^{\otimes}$ and $H^{\#}$ as well as $G^{\otimes}$ and $G^{\#}$ are defined as before. As shown in the proof of Proposition \ref{idem}, it is possible that $G^{\#}$ contains processes that are not in $H^{\#}$. Therefore $G \nsubseteq H$. 
\qed
\end{proof}

\begin{cor}
$\modald{\otimes}$ is not a closure operator.
\end{cor}
\begin{proof}
$\modald{\otimes}$ is monotonic and idempotent, but is not extensive.
\qed
\end{proof}

\subsection{Logical and Concurrent Equivalences}\label{equi}
We now turn our attention to the study of the relationships between the explicit notion of independence in concurrent systems with partial order semantics (a \emph{model} independence), and the explicit notion of independence in the logics we have defined (a \emph{logical} independence). We do so by relating well-known equivalences for concurrency, namely $\sim_{sb}$, $\sim_{hpb}$ and $\sim_{hhpb}$, with the equivalences induced by different \tfl/ sublogics where the interplay between concurrency and conflict, and concurrency and causality is \emph{syntactically} restricted.

\begin{defn} \label{sflequiv}
\emph{
{\bf ($\mathfrak{L}$ equivalence $\sim_{\mathfrak{L}}$)} Given a logic $\mathfrak{L}$, two processes $P$ and $Q$ associated with two systems $\mathfrak{T}_1$ and $\mathfrak{T}_2$, respectively, are $\mathfrak{L}$-equivalent, $P \sim_{\mathfrak{L}} Q$, if and only if, for every $\mathfrak{L}$ formula $\phi$ in $\mathfrak{F}_{\mathfrak{L}}$, $P \models^{\mathfrak{T}_1} \phi \Leftrightarrow Q \models^{\mathfrak{T}_2} \phi$, where $\mathfrak{F}_{\mathfrak{L}}$ is the set of all fixpoint-free closed formulae of $\mathfrak{L}$.
}
\eos
\end{defn}
\begin{rem}
\emph{
The previous definition delivers a logical, abstract notion of equivalence that can be used across different models for concurrency, i.e., tailored to be model independent. With this logical notion of equivalence two systems $\mathfrak{T}_1$ and $\mathfrak{T}_2$, possibly of different kinds, are equivalent with respect to some equivalence $\sim_{\mathfrak{L}}$ if, and only if, their associated process spaces cannot be differentiated by any $\mathfrak{L}$-logical formula.}
\eos
\end{rem}

Recall that in order to obtain an exact match between finitary modal logic and bisimulation, all models considered here are image-finite \cite{hmljacm-milner}, i.e., of finite branching. Moreover, since the semantics of \tfl/ is based on action labels, we only consider models without `auto-concurrency' \cite{models-winskel}, a common restriction when studying either modal logics or equivalences for (labelled) partial order models.

More precisely, auto-concurrency is the phenomenon by which multiple instances of various concurrent transitions are equally labelled. In other words, auto-concurrency can be seen as nondeterminism inside a set of independent transitions. In many cases auto-concurrency is regarded as an undesirable situation on partial order models since it can be easily avoided in practice and makes slightly counter-intuitive the analysis of behavioural properties of concurrent processes with partial order semantics.

As a matter of fact, on finite systems, auto-concurrency is formally, but not actually, a further restriction since any bounded branching system that has auto-concurrency can be effectively converted into a system that does not have auto-concurrency by a suitable relabelling of auto-concurrent transitions without changing the concurrent behaviour of the model. Notice that no auto-concurrency is a real further restriction for infinite systems as image-finiteness does not imply branching boundedness on infinite models.

Having said that, let us turn to the study of some syntactic fragments of \tfl/. They are called the \emph{natural} syntactic fragments of \tfl/ because such sub-logics arise as the languages where the dualities between concurrency and causality as well as concurrency and conflict are syntactically manipulated. As we will see the equivalences induced by all such fragments are decidable and in some cases coincide with well-known bisimilarities for interleaving and for partial order models of concurrency. We start this study of logical and concurrent equivalences by analysing a syntactic fragment of \tfl/ that is oblivious to causal information.

\subsubsection*{The Modal Mu-Calculus.} 
The first sublogic is obtained from \tfl/ by disabling the sensitivity of this logic to both dualities. On the one hand, insensitivity to the duality between concurrency and causality can be captured by considering only modalities without subscript, i.e., HML modalities, using the abbreviations given previously in Section \ref{synsem}. On the other hand, insensitivity to the duality between concurrency and conflict can be captured by considering the $\modalt$-free \tfl/ sublanguage, where $\modalt$ means $\{\modald{\otimes},\modalb{\otimes}\}$. The resulting logic has the same syntax of \lmu/. This fragment is the purely-modal $\modalt$-free fragment of \tfl/.
\begin{prop}\label{sflmults}
The syntactic purely-modal $\modalt$-free fragment of \tfl/ is semantically equivalent to the modal mu-calculus.
\end{prop}
\begin{proof}
Recall the semantics of the operators of \tfl/. Without loss of generality, we can only consider the case of the modal operators. 
\begin{center}
\begin{tabular}{lll}
$\denot{\modald{a} \phi_1}_{\mathcal{V}}$ &  $=$ & $\denot{\modald{a}_c \phi_1 \vee \modald{a}_{nc}\phi_1}_{\mathcal{V}} = \denot{\modald{a}_c \phi_1}_{\mathcal{V}} \cup \denot{\modald{a}_{nc} \phi_1}_{\mathcal{V}}$ \\
& $=$ & $\{ (R,t) \in \mathcal{X} \times \mathfrak{A} \mid \exists r \in R . ~ t \leq r \wedge (\mathfrak{X},r) \in \denot{\phi_1}_{\mathcal{V}} \} ~\cup$ \\
&  & $\{ (R,t) \in \mathcal{X} \times \mathfrak{A} \mid \exists r \in R . ~ t \ominus r \wedge (\mathfrak{X},r) \in \denot{\phi_1}_{\mathcal{V}} \}$ \\
\end{tabular} 
\end{center}
The first observation to be made is that the $\modalt$-free fragment of \tfl/ only considers maximal sets in the semantics. Therefore, support sets can be disregarded and only the state $\sigma(r)$ associated with a transition $r \in R$ need to be kept. Then, the expressions above can be modified in the following way:
\begin{center}
\begin{tabular}{lll}
$\denot{\modald{a} \phi_1}_{\mathcal{V}}$ &  $=$ & $\{ (s,t) \in S \times \mathfrak{A} \mid \exists q \in S . ~t \leq r \wedge (q,r) \in \denot{\phi_1}_{\mathcal{V}} \} ~\cup$ \\
&  & $\{ (s,t) \in S \times \mathfrak{A} \mid \exists q \in S . ~t \ominus r \wedge (q,r) \in \denot{\phi_1}_{\mathcal{V}} \} $ \\
\end{tabular} 
\end{center}
where $r = s \xrightarrow{a} q$. The second observation is that when computing the semantics of the combined operator $\modald{a}$, the conditions $t \leq r$ and $t \ominus r$ complement each other and become trivially true since there are no other possibilities. Therefore, the second component of every pair $(s,t) \in S \times \mathfrak{A}$ can also be disregarded.
\begin{center}
\begin{tabular}{lll}
$\denot{\modald{a} \phi_1}_{\mathcal{V}}$ &  $=$ & $\{ s \in S \mid \exists q \in S . ~s \xrightarrow{a} q \wedge q \in \denot{\phi_1}_{\mathcal{V}} \}$ \\
\end{tabular} 
\end{center}

The case for the box operator $\modalb{a}$ is similar. As a consequence, the semantics of all the operators of this \tfl/ sublogic and the mu-calculus coincide.
\qed
\end{proof}

\begin{rem}
\emph{
The mu-calculus cannot recognise pairs of transitions in $I$ and therefore sees any partial order model as its interleaving counterpart, or what is equivalent, a partial order model with an empty relation $I$. As a consequence, although using a partial order model of concurrency, it is possible to retain in \tfl/ all the joys of a logic with an interleaving model, and so, nothing is lost with respect to the main interleaving approaches to concurrency.
}
\eos
\end{rem}

Regarding logical and concurrent equivalences, it is now easy to see that $\sim_{sb}$, the bisimulation equivalence induced by modal logic, is captured by the fixpoint-free fragment of this sublogic, which we can denote by $\sim_{\lmu/}$. Hence, the \emph{logical correspondence} ${\sim_{\lmu/}} \equiv {\sim_{sb}}$ follows from Proposition \ref{sflmults} and the fact that modal logic characterises bisimulation on image-finite models.

\subsubsection*{The Trace Modal Mu-Calculus.} 
The second sublogic we study is the `trace modal mu-calculus', \tlmu/. This logic is obtained from \tfl/ by allowing only the recognition of the duality between concurrency and conflict with its idempotent operator. The syntax of \tlmu/ is: $\phi ::= Z \mid \neg \phi_1 \mid \phi_1 \wedge \phi_2 \mid \modald{a} \phi_1 \mid \modald{\otimes} \phi_1 \mid \mu Z . \phi_1$. We write $\sim_{\tlmu/}$ for the equivalence induced by this sublogic. It is easy to see that \tlmu/ is more expressive than \lmu/ in partial order models simply because \tlmu/ includes \lmu/ and can differentiate concurrency from nondeterminism. However, the following counter-example shows that $\sim_{\tlmu/}$ and $\sim_{hpb}$ do not coincide.

\begin{prop}
Neither ${\sim_{hpb}} \subseteq {\sim_{\tlmu/}}$ nor ${\sim_{\tlmu/}} \subseteq {\sim_{hpb}}$.
\end{prop}

\begin{proof}
The two systems at the top in Figure \ref{fig:ceslmu} are hp bisimilar and yet can be distinguished by the formula $\phi = \modald{\otimes}(\modald{a} \modald{c} \mutrue \wedge \modald{b} \modald{d} \mutrue)$. On the other hand, the systems at the bottom are not hp bisimilar and cannot be differentiated by any \tlmu/ formula. This can be verified by exhaustively checking \tlmu/ formulae in the initial state $\circ$.
\qed
\end{proof}
\begin{figure}[ht]
  \begin{center}
    \subfigure
    {\label{fig:ceslmu-hpb} 
\xymatrix @R=5pt @C=10pt 
{
          &                        &                                &                                     & \bullet \\
	        & \ar[ld]_{b} \bullet    &                                & \bullet \ar[rd]^{b} \ar[ru]^{c}     &         \\
	\bullet & \mathit{I} & \ar[lu]_{a} \ar[ld]^{b} \circ \ar[ru]^{a} \ar[rd]_{b} & \mathit{I}             & \bullet \\
	        & \ar[lu]^{a} \bullet    &                                & \bullet \ar[ru]_{a} \ar[rd]_{d}     &         \\
          &                        &                                &                                     & \bullet 
} 
    } \hspace{0.7cm}
    \subfigure
    {\label{fig:ceslmu-nothpb}
\xymatrix @R=5pt @C=10pt 
{
          &                                    &                                &                                     & \bullet \\
	        & \ar[ld]_{b} \bullet                &                                & \bullet \ar[rd]^{b} \ar[ru]^{c}     &         \\
	\bullet & \mathit{I} & \ar[lu]_{a} \ar[ld]^{b} \circ \ar[ru]^{a} \ar[rd]_{b} & \mathit{I}             & \bullet \\
	        & \ar[ld]^{d} \ar[lu]^{a} \bullet    &                                & \bullet \ar[ru]_{a}                 &         \\
  \bullet &                                    &                                &                                     & 
} 
    }
\\
\begin{center}\begin{small} $\sim_{hpb}$ but not $\sim_{\tlmu/}$ \end{small} \end{center}
    \subfigure
    {\label{fig:nothpb1} 
\xymatrix @R=5pt @C=10pt 
{
	                                & \bullet \ar[rd]^{b} &          \\
	\circ \ar[ru]^{a} \ar[rd]_{b} & \mathit{I}          & \bullet  \\
	                                & \bullet \ar[ru]_{a} &          \\
} 
    } \hspace{0.7cm}
    \subfigure
    {\label{fig:nothpb2} 
\xymatrix @R=5pt @C=10pt 
{
	\bullet && \ar[ll]_{b} \bullet \ar[rr]^{b}    &            & \bullet                \\
	        &&                                    & \mathit{I} &                        \\
	        && \circ \ar[uu]^{a} \ar[rr]_{b}    &            & \bullet \ar[uu]_{a}    \\
} 
    }
\begin{center}\begin{small} $\sim_{\tlmu/}$ but not $\sim_{hpb}$ \end{small} \end{center}
\end{center}
\caption{Counter-examples of the coincidence between $\sim_{hpb}$ and $\sim_{\tlmu/}$}
\label{fig:ceslmu}
\end{figure}
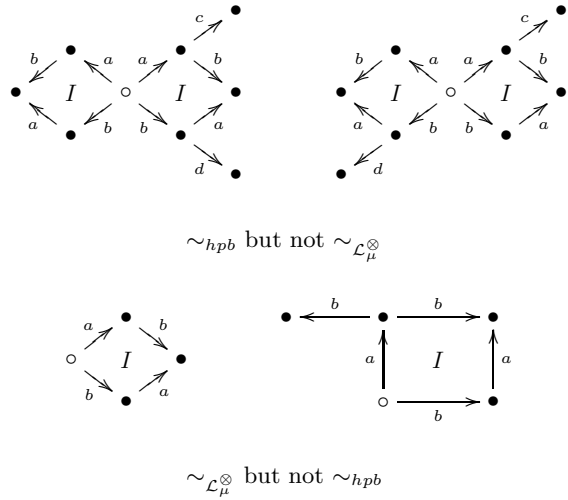 

There is a fundamental reason for the mismatch between $\sim_{hpb}$ and $\sim_{\tlmu/}$. It has to do with a special `sharing' of resources between some of the transitions in the model. This special kind of sharing of resources is characterised by a phenomenon called \emph{confusion}, which is a concept in net theory, and thus, it is useful to think of it directly on nets.

Confusion can be symmetric or asymmetric. Roughly speaking, it is a phenomenon that arises between at least three different actions, say between $t_1$, $t_2$, and $t_3$. In the symmetric case, two of them are independent, e.g., $t_1 \ppar/ t_2$, and at the same time are in conflict with the third action, i.e., ${{^{\bullet}t_1} \cap {^{\bullet}t_3}} \neq \emptyset$ and ${{^{\bullet}t_2} \cap {^{\bullet}t_3}} \neq \emptyset$. On the other hand, in the asymmetric case, two of the actions are independent, e.g., $t_1 \ppar/ t_2$ as before, whereas the third one is in conflict with one of the independent actions, say with $t_1$, and causally depends on the other, i.e., ${{^{\bullet}t_1} \cap {^{\bullet}t_3}} \neq \emptyset$ and ${{t_2^{\bullet}} \cap {^{\bullet}t_3}} \neq \emptyset$, respectively. Confusion is important because, although it is undesirable when analysing the behaviour of a concurrent system, it is also ``inherent to any reasonable net model of a mutual exclusion module'' \cite{confusiontcs-esmith}. Confusion is also present when modelling race conditions in concurrent and distributed systems with shared memory models. These facts show the ubiquity of this phenomenon when analysing real-life models of communicating concurrent systems. Although confusion is a natural concept in net theory, it can also be defined for TSI and event structures, though in these cases the definition is more complicated because it involves sets of transitions and sets of events, respectively, rather than single actions as in the Petri net case.

Confusion appears in the two counter-examples shown in Figure \ref{fig:ceslmu}. In both cases in its asymmetric variant. The problem is that both $\sim_{hpb}$ and $\sim_{\tlmu/}$ can recognise some forms of confusion, but not all of them. However, there is a class of nets called \emph{free-choice} where confusion never arises, and for which coincidence results between $\sim_{hpb}$ and $\sim_{\tlmu/}$ may be possible. Define a free-choice system as follows (see \cite{njc-bradfield} for an equivalent presentation):

\begin{defn} \label{fcsys}
\emph{
A system $\mathfrak{T}$ is \emph{free-choice} iff whenever there are three transitions $t_1$, $t_2$, and $t_3$ such that $t_1$ is in conflict with $t_2$ and any of them is independent of $t_3$, i.e., either $(t_1,t_3) \in I$ or $(t_2,t_3) \in I$, then there must be two transitions $t_4$ and $t_5$, where $t_1 \sim t_4$ and $t_2 \sim t_5$, such that $t_4$ and $t_5$ are also in conflict, and $t_3$ is independent of both $t_4$ and $t_5$.
}
\eos
\end{defn}

Informally, the previous definition means that a choice, i.e., a conflict, cannot be globally affected by the execution of a concurrent transition, since equivalent choices are always possible both before and after that. These facts, along with the observations made before, led us to believe that the following statement holds, although we have so far not been able to prove it. 

\begin{conj}
${\sim_{\tlmu/}} \equiv {\sim_{hpb}}$ on free-choice systems without auto-concurrency.
\end{conj}

Now, let us move to the study of a modal logic that is sensitive to the causal information embodied in partial order systems. In particular, it will be shown that for some classes of systems the \emph{local} duality between concurrency and causality is good enough to capture the full notion of \emph{global} causality defined by $\sim_{hpb}$ on partial order models.

\subsection*{From Local to Global Causality}

In this section we show the first coincidence result of the equivalence induced by one of the sublogics of \tfl/ with a bisimilarity for partial order systems. The result holds for a class of systems whose expressive power lies between that of so-called `free-choice' nets \cite{fcnets-esparza} and that of safe nets, as before with the usual restrictions to systems that are image-finite and have no auto-concurrency.

The coincidence result is with respect to $\sim_{hpb}$. This equivalence is considered to be the standard bisimulation equivalence for causality since it fully captures the interplay between branching and causal behaviour. 
The interesting feature of this coincidence result is that $\sim_{hpb}$ provides a \emph{global} notion of causality whereas the logic we are about to study induces a \emph{local} one, as shown later on. Then, the question we answer here is that of the class of systems for which `local causality' fully captures the standard notion of `global causality'. Such an answer is given by the following modal logic. 

\subsubsection*{The Causal Modal Mu-Calculus.} 
The fourth sublogic to be considered is the `causal modal mu-calculus', \clmu/. This sublogic is obtained from \tfl/ by allowing only the recognition of the duality between concurrency and causality throughout the modal operators on transitions of \tfl/. The syntax of this syntactic fragment is the following: $\phi ::= Z \mid \neg \phi_1 \mid \phi_1 \wedge \phi_2 \mid \modald{a}_{c} \phi_1 \mid \modald{a}_{nc} \phi_1 \mid \mu Z . \phi_1$. 

Clearly, \clmu/ is also more expressive than \lmu/ in partial order models because of the same reasons given for \tlmu/. The naturality of \clmu/ for expressing causal properties is demonstrated by the equivalence it induces, written as $\sim_{\clmu/}$, which coincides with $\sim_{hpb}$, the standard bisimulation equivalence for causal systems, when restricted to systems without auto-concurrency where any 3-tuple of transitions $(t_1,t_2,t_3)$ in confusion is in some sense deterministic. Thus, let us define confusion, a ternary relation on transitions as well as its deterministic variant when considering labelled systems.

\begin{defn} \label{ch3-confusion}
\emph{
\textbf{(Confusion)}
Let $\cfs/$ be a relation on transitions of a system $\mathfrak{T}$ such that ${(t_1,t_2,t_3)} \in {\cfs/}$ iff $t_1 \otimes t_2$ and either $t_1 \# t_3$ and $t_2 \# t_3$ (the symmetric case) or $t_1 \leq t_3$ and $\exists r_{2}.~{t_2 \sim r_{2}} \wedge {r_{2} \# t_3}$ (the asymmetric case). A tuple ${(t_1,t_2,t_3)} \in {\cfs/}$ is deterministic iff either the three transitions have different labels or $\delta(t_1) = \delta(t_3)$ and $t_1 \leq t_3$.
}
\eos
\end{defn}

There are analogous Petri net and event structure definitions for confusion using the basic elements of such models. Those definitions are much better known than the one presented here since confusion is a basic concept in net theory; however, the definition we have just given is equivalent. Perhaps due to this fact is that a very easy way of depicting confusion is using Petri nets. Figure \ref{ch3-conf-fig} shows the two simplest nets featuring confusion, both in their symmetric and asymmetric variants. 

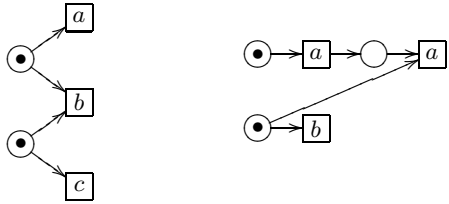
\begin{figure}[ht]
\begin{center}
\subfigure
{ 
\xymatrix @R=6pt @C=12pt 
{
 & \petsqa  \\
\petone \ar[ru] \ar[rd] &\\
& \petsqb   \\
\petone \ar[ru] \ar[rd]&\\
 & \petsqc  
} 
}
\hspace{1.7cm}
\subfigure
{ 
\xymatrix @R=6pt @C=12pt 
{
&&&\\
\petone \ar[r] & \petsqa \ar[r] & \petzero \ar[r] & \petsqa   \\
&&&\\
\petone \ar[r] \ar[rrruu] & \petsqb  & &  \\
&&&
} 
}
\end{center}
\caption{\small Confusion: the Petri net on the left has symmetric confusion and the Petri net on the right has asymmetric confusion. In both cases it is deterministic.}
\label{ch3-conf-fig}
\end{figure} 

Any Petri net that has confusion must have either of these two nets as a subsystem. The statement equivalently holds for TSI and event structures if considering, respectively, the TSI and event structure models corresponding to such nets. This property allows one to define a class of systems for which the logical equivalence induced by \clmu/ captures $\sim_{hpb}$. Such a class contains those systems without auto-concurrency that either are \emph{free-choice} or whose confusion relation has only deterministic tuples. Thus, let us now define the class of free-choice systems. For simplicity, we do so indirectly via the standard definition of free-choice nets, which is well-known in the literature.

\begin{defn} \label{ch3-fcnets}
\emph{
\textbf{(Free-choice nets)}
Let $\mathcal{N}$ be a net. A net is \emph{free-choice} iff for all $s \in P$ we have that ${{\mid s^{\bullet} \mid} \leq {1}}$ or ${\forall t \in s^{\bullet}. {\mid {^{\bullet}t} \mid} = {1}}$. 
}
\eos
\end{defn}

A free-choice Petri net is a free-choice net with an initial marking. A free-choice event structure is an event structure unfolding \cite{pnesdom-winskel} of a free-choice Petri net and a free-choice TSI is the TSI model obtained from a free-choice Petri net. Free-choice nets, and hence free-choice systems, have no confusion as other classes of systems. But what is interesting about this class of nets is that the confusion-freeness property (which is a behavioural characteristic) comes directly from a structural property of these nets. In particular, any free-choice net can be built using the subnets shown in Figure \ref{ch3-freechoicesubnets-fig} (with additional flow arrows after net actions, which can be used unrestrictedly as long as the net is safe).

\begin{figure}[ht]
\begin{center}
\subfigure
{ 
\xymatrix @R=6pt @C=12pt 
{
\petzero \ar[rd] &\\
\vdots & \petsqe   \\
\petzero \ar[ru] &
} 
}
\hspace{1.7cm}
\subfigure
{ 
\xymatrix @R=6pt @C=12pt 
{
                         & \petsqe  \\
\petzero \ar[ru] \ar[rd] & \vdots  \\
                         & \petsqe
}
}
\end{center}
\caption{\small Free-choice nets: two subnets from which any free-choice system can be built.}
\label{ch3-freechoicesubnets-fig}
\end{figure}
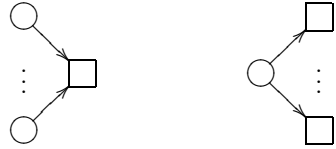 

We are almost ready to show that the two equivalences $\sim_{\clmu/}$ and $\sim_{hpb}$ coincide for a class of systems that we call `fc-structured', and denote by $\Xi$.

\begin{defn} \label{ch3-fcstructuredsys}
\emph{
\textbf{(Fc-structured ($\Xi$) systems)}
The family of \emph{fc-structured} ($\Xi$) systems is the class of systems without auto-concurrency that either are free-choice or whose confusion relation has only deterministic elements. 
}
\eos
\end{defn}

\begin{rem} \label{ch3-fcstructuredsysrem}
\emph{
The family of $\Xi$ systems contains, at least, the following classes of models (without auto-concurrency and with a deterministic conflict relation): Moore and Mealy machines, labelled graphs, synchronous and asynchronous products of sequential systems, free-choice systems, and nondeterministic concurrent systems. More importantly, $\Xi$ systems can model mutual exclusion protocols and the usual synchronization mechanisms of some process calculi.
}
\eos
\end{rem}

Now, back to the issue of relating $\sim_{hpb}$ and $\sim_{\clmu/}$, the proof that $\sim_{hpb}$ and $\sim_{\clmu/}$ coincide for the class of $\Xi$ systems goes by showing that the two inclusions ${\sim_{hpb}} \subseteq {\sim_{\clmu/}}$ and ${\sim_{\clmu/}} \subseteq {\sim_{hpb}}$ hold separately. In fact, the first inclusion holds for any class of systems while the second one requires the restriction to $\Xi$ systems.

\begin{lem}\label{ch3-hpbtoclmu}
\emph{\textbf{(Logical soundness)}}
${\sim_{hpb}} \subseteq {\sim_{\clmu/}}$.
\end{lem}
\begin{proof}
This inclusion can be shown by induction on \clmu/ formulae, which we denote by $\mathfrak{F}_{\clmu/}$. Let $\mathfrak{T}_1$ and $\mathfrak{T}_0$ be two systems and $P \in \mathfrak{S}_1$ and $Q \in \mathfrak{S}_0$ two processes that belong to the process spaces $\mathfrak{S}_1$ and $\mathfrak{S}_0$ associated with $\mathfrak{T}_1$ and $\mathfrak{T}_0$, respectively. If $P \sim_{hpb} Q$ then for all $\phi \in \mathfrak{F}_{\clmu/}$ we have that $P \models^{\mathfrak{T}_1}_{\mathcal{V}_1} \phi \Leftrightarrow Q \models^{\mathfrak{T}_0}_{\mathcal{V}_0} \phi$ given two models $\mathfrak{M}_1 = (\mathfrak{T}_1,\mathcal{V}_1)$ and $\mathfrak{M}_0 = (\mathfrak{T}_0,\mathcal{V}_0)$. Since \clmu/ only considers maximal sets, the process $P = (p,t)$ (resp. the process $Q = (q,t)$) is actually a binary tuple in $S_1 \times \mathfrak{A}_1$ (resp. in $S_0 \times \mathfrak{A}_0$) rather than a tuple in $\mathcal{X}_1 \times \mathfrak{A}_1$ (resp. in $\mathcal{X}_0 \times \mathfrak{A}_0$). Henceforth, let us write $\models$ instead of $\models^{\mathfrak{T}_i}_{\mathcal{V}_i}$, for $i \in \{0,1\}$, since the models will be clear from the context.

The base case of the induction is when $\phi = \mutrue$ or when $\phi = \mufalse$ which is trivial since $\mutrue$ and $\mufalse$ are always true and false, respectively. Now, consider the cases for the boolean operators $\wedge$ and $\vee$; first suppose that $\phi = \psi_1 \wedge \psi_2$ and assume that the result holds for both $\psi_1$ and $\psi_2$. By the definition of the satisfaction relation $P \models \phi$ iff $P \models \psi_1$ and $P \models \psi_2$ iff by the inductive hypothesis $Q \models \psi_1$ and $Q \models \psi_2$, and hence, by the definition of the satisfaction relation $Q \models \phi$. The case for $\vee$ is similar.

Now, consider the cases for the four modalities. First, suppose $\phi = \modalb{a}_{nc} \psi$ and $P \models \phi$. Therefore, for any $P' = (p',t')$, such that $a = \delta(t')$ and $P \xrightarrow{a} P'$ and $t \ominus t'$, it follows that $P' \models \psi$. Now, let $Q \xrightarrow{a} Q'$ such that $a = \delta(t')$ and $t \ominus t'$ since the bisimulation must remain synchronous. Just to recall, synchrony in an hp bisimulation means that the last transition chosen in $\mathfrak{T}_1$ (resp. in $\mathfrak{T}_0$) is concurrent with the former transition also chosen in $\mathfrak{T}_1$ (resp. in $\mathfrak{T}_0$) iff the same pattern holds in the last two transitions chosen in $\mathfrak{T}_0$ (resp. in $\mathfrak{T}_1$), and moreover the two sequences of transitions (i.e., runs) that are generated in this way are the linearisations of isomorphic posets. So, as we know that for some $P'$ there is a $P \xrightarrow{a} P'$, where $t \ominus t'$, and by the inductive hypothesis $P' \sim_{hpb} Q'$, then $Q' \models \psi$, where $t \ominus t'$, and so by the definition of the satisfaction relation $Q \models \phi$. The case when $Q$ satisfies $\phi$ is symmetric, and the case when $\phi = \modalb{a}_c \psi$ is similar (only changing $\ominus$ for $\leq$). The cases for the operators $\modald{a}_c$ and $\modald{a}_{nc}$ are analogous.
\qed
\end{proof}

In order to show the second inclusion, namely ${\sim_{\clmu/}} \subseteq {\sim_{hpb}}$, we first require some lemmas that characterise the set of runs that can be identified by \clmu/ in a partial order system. More specifically, a proof that if two systems $\mathfrak{T}_0$ and $\mathfrak{T}_1$ are \clmu/-equivalent, then for each run of one of the systems there exists a `locally synchronous' run (which is defined below) in the other system. Then, one can use this result to show that for any two $\Xi$ systems $\mathfrak{T}_0$ and $\mathfrak{T}_1$ such that $\mathfrak{T}_0 \sim_{\clmu/} \mathfrak{T}_1$, each pair of locally synchronous runs is moreover induced by two isomorphic posets, and hence, the two systems must be $\sim_{hpb}$ as well since in such a case the pair of runs is synchronous.

Recall the definition of runs and of synchronous runs given before, and let $\pi_0 \in {\Pi_{\mathfrak{T}_0}}$ and $\pi_1 \in {\Pi_{\mathfrak{T}_1}}$ be two runs of two systems $\mathfrak{T}_0$ and $\mathfrak{T}_1$, and $u,v$ two transitions. A pair of runs ($\pi_0 . u , \pi_1 . v$) is inductively defined as \emph{locally synchronous} iff ($\pi_0,\pi_1$) is locally synchronous and ${{({\varrho(\pi_0)},u)} \in I_0} \Leftrightarrow {{({\varrho(\pi_1)},v)} \in I_1}$, where $I_0$ and $I_1$ are the independence relations of $\mathfrak{T}_0$ and $\mathfrak{T}_1$. By definition, the pair of empty runs $(\epsilon,\epsilon)$ is locally synchronous. 
Note that the definitions of locally synchronous runs and synchronous runs is quite similar; the only difference is that synchronous runs must be the linearization of isomorphic posets whereas locally synchronous runs need not be.

\begin{lem}\label{ch3-clmuruns}
Let $\mathfrak{T}_0$ and $\mathfrak{T}_1$ be two systems and ${\Pi_{\mathfrak{T}_0}}$ and ${\Pi_{\mathfrak{T}_1}}$ their sets of runs. If $\mathfrak{T}_0 \sim_{\clmu/} \mathfrak{T}_1$ then for each $\pi_0 \in {\Pi_{\mathfrak{T}_0}}$ (resp. $\pi_1 \in {\Pi_{\mathfrak{T}_1}}$) there exists a run $\pi_1 \in {\Pi_{\mathfrak{T}_1}}$ (resp. $\pi_0 \in {\Pi_{\mathfrak{T}_0}}$) such that the pair of runs $(\pi_0,\pi_1)$ is locally synchronous.
\end{lem}
\begin{proof}
The proof goes by a contradiction argument. Suppose that for all $\phi$ in $\mathfrak{F}_{\clmu/}$ we have that $P \models \phi \Leftrightarrow Q \models \phi$ and there exists a run in one of the systems that is not locally synchronous to any of the runs in the other system. The case where $P$ and $Q$ are the initial processes of $\mathfrak{T}_0$ and $\mathfrak{T}_1$, respectively, is trivially false since, by definition, the pair of empty runs ($\epsilon,\epsilon$) is locally synchronous.

Then, suppose now that $(\pi_0,\pi_1)$ is locally synchronous and that $P$ and $Q$ are two processes reached, respectively, in $\mathfrak{T}_0$ and in $\mathfrak{T}_1$ after following $\pi_0$ and $\pi_1$ in each system (starting from their initial processes). Additionally, suppose that there exists a transition $u$ in one of the systems, say in $\mathfrak{T}_0$, such that there is no transition $v$ in the other system for which the pair of runs $(\pi_0 . u , \pi_1 . v)$ is locally synchronous. Note that $P$ and $Q$ are strongly bisimilar, since \clmu/ includes \lmu/, and thus, the case in which a processes can perform a transition (regardless of its label) and the other cannot do so is impossible as this contradicts the hypothesis that $P \sim_{sb} Q$. 

So, suppose that for some transition $u$ with label $a$, $P = (p,\varrho(\pi_0)) \xrightarrow{a} P' = (p',u)$ and $\varrho(\pi_0) \ominus u$ (resp. $\varrho(\pi_0) \leq u$), but for all transitions $v$ such that $a = \delta(v)$ it holds that $Q = (q,\varrho(\pi_1)) \xrightarrow{a} Q' = (q',v)$ and $\varrho(\pi_1) \leq v$ (resp. $\varrho(\pi_1) \ominus v$) only. However, we know that, by hypothesis, $P \sim_{\clmu/} Q$ and so, it must be true that if $P \models \modald{a}_{nc} \phi$ (resp. if $P \models \modald{a}_{c} \phi$) then $Q \models \modald{a}_{nc} \phi$ (resp. if $Q \models \modald{a}_{c} \phi$), which is a contradiction. Thus, one must be able to match pairs of independent transitions in one of the systems whenever the same happens in the other system for all pairs of processes $P$ and $Q$ satisfying that $P \sim_{\clmu/} Q$.
\qed
\end{proof}


The previous lemma ensures that if two systems satisfy the same set of \clmu/ formulae, then, locally, they have the same causal behaviour. However, in order to show that, globally, they also have the same causal behaviour, one needs some additional information, which is given by the following lemma.

\begin{lem}\label{ch3-labels}
Let $\mathfrak{T}$ be a $\Xi$ system whose conflict relation is $\cfs/$ and let $\pi \in {\Pi_{\mathfrak{T}}}$. If after executing the run $\pi$ in $\mathfrak{T}$ there are two different enabled transitions $u$ and $v$ such that $\delta(u) = \delta(v)$, then the following two statements hold:
\begin{enumerate}
\item $u \# v$.
\item There is at most one transition $t$ in $\pi$ such that $\tau(t) = \sigma(u') = \sigma(v')$ for which $t \leq u'$ and $t \leq v'$ and $u \sim u'$ and $v \sim v'$.
\end{enumerate}
\end{lem}
\begin{proof}
In the same order as in the statement of the lemma:
\begin{enumerate}
\item Because there is no auto-concurrency.
\item As the confusion relation is deterministic there is no ${c} \in {\cfs/}$ such that both $u$ and $v$ belong to $c$; in particular neither transition can be an instance of an action $e$ (at the net level) for which ${\mid {^{\bullet}e} \mid} > 1$. Instead, such transitions are instances of two different actions $e_1$ and $e_2 $ for which ${\mid {^{\bullet}e_1} \mid} = 1 = {\mid {^{\bullet}e_2} \mid}$.\qed
\end{enumerate}
\end{proof}

Finally, the following lemma ensures that for the class of $\Xi$ systems, the notion of locally synchronous runs (associated with local causality) is good enough to capture the stronger notion of synchronous runs (associated with global causality), so long the two systems satisfy the same set of \clmu/ formulae. 

\begin{lem}\label{ch3-localhpbtoglobalhpb}
Let $\mathfrak{T}_0$ and $\mathfrak{T}_1$ be $\Xi$ systems whose sets of runs are ${\Pi_{\mathfrak{T}_0}}$ and ${\Pi_{\mathfrak{T}_1}}$. If for each $\pi_0 \in {\Pi_{\mathfrak{T}_0}}$ (resp. $\pi_1 \in {\Pi_{\mathfrak{T}_1}}$) there exists some $\pi_1 \in {\Pi_{\mathfrak{T}_1}}$ (resp. $\pi_0 \in {\Pi_{\mathfrak{T}_0}}$) such that $(\pi_0,\pi_1)$ is locally synchronous, then $(\pi_0,\pi_1)$ is, moreover, synchronous. 
\end{lem}
\begin{proof}
The proof is based on the fact that if $(\pi_0,\pi_1)$ is a locally synchronous pair, then the posets induced by such locally synchronous runs induce isomorphic posets if the systems are $\Xi$, and hence, the pair of runs is also `globally' synchronous. We proceed by induction on the length of runs. The base case, i.e., when the pair of runs is $(\pi_0,\pi_1) = (\epsilon,\epsilon)$, is trivial since in this case the two posets are empty. 

Then, for the induction step, suppose that there is a non-empty run $\pi_0$ of size $k$ that is locally synchronous to some run $\pi_1$; moreover, suppose that $\pi_0$ and $\pi_1$ induce isomorphic posets. We show that there is not a run $\pi_0 . u$ which induces a poset that is not isomorphic to any of the posets induced by those runs $\pi_1 . v$ for which the pairs of extended runs $(\pi_0 . u , \pi_1 . v)$ are locally synchronous. 

Due to the definition of $\Xi$ systems, one can consider the following three cases: (1) the transition $u$ is the instance of a net action $e$ such that ${\mid{^{\bullet}e}\mid}>1$ and $u$ is not in the conflict relation of $\mathfrak{T}_0$; (2) the transition $u$ is the instance of a net action $e$ such that for some net place $s$ we have that $e \in s^{\bullet}$ and $\forall e \in s^{\bullet}. {{\mid{^{\bullet}e}\mid}=1}$ and $u$ is not in the conflict relation of $\mathfrak{T}_0$; or (3) the transition $u$ is an instance of a net action of either type and is in the conflict relation of $\mathfrak{T}_0$. 

For the first case, let $\pi_0$ be any run such that $\varrho(\pi_0) \leq u$. By hypothesis we have that the posets induced by $\pi_0$ and $\pi_1$ are isomorphic, that $u$ depends only on one transition (namely, on $\varrho(\pi_0)$), and that $\varrho(\pi_1) \leq v$ as well. Then, the only possibility for this case to fail is if $v$, unlike $u$, causally depends on more than only one transition (since it already depends on $\varrho(\pi_1)$). Suppose this could happen; then, there is at least one transition $e_j$ in $\pi_1$ on which $v$ also causally depends and that is independent of $\varrho(\pi_1)$. Then there must exist a run $\pi^{-}_1$ of length $k-1$ that do not contain $e_j$ and where $\varrho(\pi^{-}_1) \leq v'$ for some $v'$ such that $\delta(v') = \delta(v)$. Since $v$ and $v'$ cannot be two instances of the same net action, then they must be in conflict (because there is no auto-concurrency) and moreover belong to some tuple $c$ of the confusion relation $\cfs/$ of $\mathfrak{T}_1$, which is impossible since $\Xi$ systems have a deterministic confusion relation. As a consequence any transition $u$ of this kind can be matched only by a transition $v$ that is the instance of a net action $e$ for which ${\mid{{^{\bullet}}e}\mid} = 1$, and due to Lemma \ref{ch3-labels}, such kind of transitions extend a unique transition of any run, keeping the two extended runs $\pi_0 . u$ and $\pi_1 . v$ not only locally synchronous but also globally synchronous.

For the second case, suppose that $u$ depends on a set $\{e^{i}_0,..., e^{k}_0,..., e^{m}_0 \}$ of elements of the poset induced by $\pi_0$, i.e., $\forall e \in \{e^{i}_0,..., e^{k}_0,..., e^{m}_0 \}. (e,u) \not \in I_0$, and there is at least one $e^{k}_0$ that was related to some $e^{k}_1$ of $\pi_1$ while constructing the two locally synchronous runs, i.e., $e^{k}_0 = \pi_0(k)$ and $e^{k}_1 = \pi_1(k)$ for some natural number $k$, but that is not extended in $\pi_0$ with respect to $u$ as $e^{k}_1$ is extended in $\pi_1$ with respect to $v$, i.e., which makes the two induced posets not isomorphic because $(e^{k}_0,u) \not \in I_0$ whereas $(e^{k}_1,v) \in I_1$. 

For the same reasons given in the first case, $v$ cannot depend on only one transition in $\pi_1$. On the contrary it must depend on at least two transitions, one of which must have the same label as $e^{k}_0$ and $e^{k}_1$; let $e^{n}_1$ be such a transition. As in the first case, w.l.o.g.\ the other transition can be $\varrho(\pi_1)$. Then, we have that $v$ causally depends on $e^{n}_1$ and is independent of $e^{k}_1$, which is independent of $e^{n}_1$. But this is impossible since $\delta(e^{n}_1) = \delta(e^{k}_1)$ and there is no auto-concurrency. Therefore, both runs must be extended in a synchronous way in this case as well.

Finally, for the third case notice that the arguments given before apply here as well, regardless of the kind of transition under consideration since the two properties in the former cases still hold: on the one hand, any two transitions equally labelled are always in conflict and causally depend (locally) on only one transition of any run; and, on the other hand, whenever is enabled a transition that is an instance of a net action whose preset is not a singleton, then that transition is the only one enabled with such a label.

As a consequence, any transition $v$ must extend the poset induced by $\pi_1$ in the same way as $u$ extends the poset induced by $\pi_0$, i.e., $\forall k \in \{1,...,\mid\pi_0\mid\}$ one has that $(\pi_0(k),u) \not \in I_0$ iff $(\pi_1(k),v) \not \in I_1$, making the two posets isomorphic in all cases and for all pairs $(\pi_0,\pi_1)$ of locally synchronous runs of any length.
\qed
\end{proof}

Informally, one can say that the arguments in the proof just given go through because any `extra-concurrency' in one of the systems with respect to the other can be recognised since there is no auto-concurrency, and any `extra-causality' can be recognised since, in $\Xi$ systems, any two transitions enabled at the same time and equally labelled must be in conflict and causally depend on one transition in any run. 

\begin{cor}\label{ch3-clmutohpb}
\emph{\textbf{(Logical completeness)}}
${\sim_{\clmu/}} \subseteq {\sim_{hpb}}$ on $\Xi$ systems.
\end{cor}
\begin{proof}
From Lemmas \ref{ch3-clmuruns} and \ref{ch3-localhpbtoglobalhpb}.
\qed
\end{proof}

\begin{thm}\label{causalsim}
\emph{\textbf{(Full logical definability)}}
${\sim_{\clmu/}} \equiv {\sim_{hpb}}$ on $\Xi$ systems.
\end{thm}
\begin{proof}
Immediate from Lemma \ref{ch3-hpbtoclmu} and Corollary \ref{ch3-clmutohpb}.
\qed
\end{proof}

\begin{cor}\label{corclmudec}
$\sim_{\clmu/}$ is decidable on $\Xi$ systems.
\end{cor}
\begin{proof}
Follows from Theorem \ref{causalsim} and the fact that $\sim_{hpb}$ is decidable \cite{dechpb-meyer}.
\qed
\end{proof}

The previous theorem shows that for the class of $\Xi$ systems the notion of `local' causality defined by \clmu/ captures the stronger notion of `global' causality, which is captured by $\sim_{hpb}$ in arbitrary classes of models of true-concurrency.

This result can have interesting practical applications. For instance, the complexity of deciding whether two systems are hp bisimilar, i.e., that they posses the same causal properties, is EXPTIME-complete \cite{dechpb-meyer}; since verifying that two partial order system systems satisfy the same set of \clmu/ properties requires one to check only related `localities' then the problem may be computationally easier.

\subsubsection*{A Partial Result on the Logical Equivalence Induced by \tfl/.} 
Although the equivalence induced by \tfl/ is analysed in the following section using game-theoretical arguments, we first present a simple preliminary result that relates both $\sim_{\tfl/}$ with $\sim_{hhpb}$, without using any game-theoretical machinery. 

Consider the counter-example given by Fr{\"o}schle \cite{border-froschle} using Petri nets, which provides evidence of the non-coincidence between $\sim_{hpb}$ and $\sim_{hhpb}$ in free-choice systems. Although the systems presented there in Figure 1 and here in the Figure \ref{sfhpbvshhpb} are not hhp bisimilar, they cannot be distinguished by any \tfl/ formula. This result shows that in general $\sim_{hhpb}$ does not coincide with $\sim_{\tfl/}$. However, the precise relation between $\sim_{hhpb}$ and $\sim_{\tfl/}$ is to be defined in the following section using a new form of higher-order logic game for bisimulation. For now, we have the following result:

\begin{prop}\label{tflnotinhhpb}
${\sim_{\tfl/}} \not \subseteq {\sim_{hhpb}}$.
\end{prop}

\begin{figure}[htp]
  \begin{center}
    \subfigure
    {
\xymatrix @R=10pt @C=10pt 
{
    &\petone \ar[ld] \ar[rd] && A && \petone \ar[ld] \ar[rd] & \\
    \petsqa \ar[d] & & \petsqa \ar[d] && \petsqb \ar[d] & & \petsqb \ar[d] \\
    \petzero & & \petzero \ar[rd] && \petzero \ar[ld] & & \petzero \\
 &&&\petsqc &&&\\
} 
}\\
\subfigure
    {
\xymatrix @R=10pt @C=10pt 
{
    &\petone \ar[ld] \ar[d] \ar[rd] && B && \petone \ar[ld] \ar[d] \ar[rd] & \\
    \petsqa \ar[d] & \petsqa \ar[d] & \petsqa \ar[d] && \petsqb \ar[d] & \petsqb \ar[d] & \petsqb \ar[d] \\
    \petzero & \petzero \ar[rrdd] & \petzero \ar[rd] && \petzero \ar[ld] & \petzero \ar[lldd] & \petzero \\
 &&&\petsqc &&&\\
 &&&\petsqc &&&\\
} 
    }
\caption{Not inclusion of $\sim_{\tfl/}$ in $\sim_{hhpb}$: we have that $A {\mathbf{\sim_{\tfl/}}} B$ and $A {\not \sim_{hhpb}} B$.}
\label{sfhpbvshhpb}
\end{center}
\end{figure}
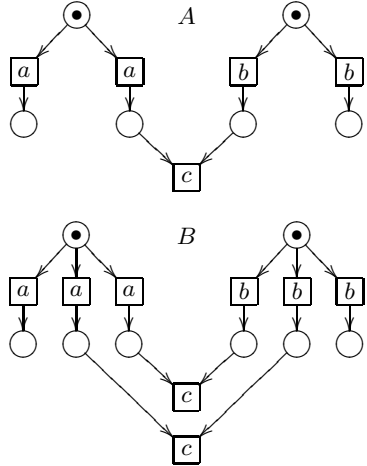

The two systems in Figure \ref{sfhpbvshhpb} are free-choice. We have noted that all counter-examples we have in which non-coincidence from $\sim_{hpb}$ and $\sim_{\tfl/}$ arises are due to anomalies in the concurrent behaviour of the models related to the phenomenon called confusion. For this reason along with the fact the idempotent operator of \tfl/ can recognise conflict-free sets of transitions, we believe, but have no proof, that these two equivalences actually coincide for the class of free-choice systems. So, we finish this section with the following conjecture:

\begin{conj}
${\sim_{\tfl/}} \equiv {\sim_{hpb}}$ on free-choice systems without auto-concurrency.
\end{conj}

\section{Higher-Order Logic Games for Bisimulation} \label{bisgames}
The logic games for bisimulation presented in Section \ref{pre} provide a \emph{first-order} power on the transitions that are picked when playing the game. In this section we introduce a bisimulation game that gives the players \emph{higher-order} power on the sets of transition in the game board. Since such games may be too powerful without restrictions, we consider higher-order games for bisimulation where this higher-order power is restricted to simple characteristic sets of transitions in the board. Moreover, since these games are intended to be used in the analysis of modal logics, then such a higher-order power is also restricted to a local setting. 

In particular, in this section we want to define higher-order logic games for bisimulation that help us understand the bisimulation equivalence induced by \tfl/, and how this logical equivalence relates to the best known history-preserving bisimilarities for concurrency. To this end, we consider games with monadic second-order power on conflict-free sets of transitions, and show that such games can both capture the logical equivalence induced by \tfl/ and be related to the bisimulation games that characterise hpb and hhpb in a very natural way. The higher-order logic game for bisimulation defined here is a refinement of the bisimulation game first presented in \cite{fos09-gut} for SFL, a fixpoint logic similar to \tfl/.

\subsection{Logical Correspondence}

In this section we give a game-theoretical characterisation of the equivalence that \tfl/ induces by defining a characteristic bisimulation game for it. As we already know that the game for \tfl/ must be at least as powerful as the game for hpb, since such a bisimilarity is captured by a syntactic \tfl/ fragment, then it is natural to design a game that extends the (first-order) bisimulation game for hpb. Here we do so. The game presented in this section conservatively extends the hp bisimulation game, and therefore the usual game for modal logic. We show that this bisimulation game, which we call `trace history-preserving' (thp) bisimulation game, characterises the logical equivalence induced \tfl/.

More importantly, we show that, on $\Xi$-systems, the equivalence relation induced by \tfl/ is strictly stronger than $\sim_{hpb}$ and strictly weaker than $\sim_{hhpb}$. We also show that the game characterising $\sim_{\tfl/}$, i.e., the thpb game, is decidable in finite systems, a result that contrasts with hhpb games, which are undecidable on arbitrary finite models \cite{undhhpb-jur}; whether $\sim_{hhpb}$ is decidable on $\Xi$-systems is an open question. These features amongst others make the game introduced here, and consequently the bisimulation equivalence induced by \tfl/, an interesting candidate for an equivalence for systems with partial order semantics. 

But, before presenting the game \tfl/ let us introduce a final definition that is related to the role of support sets as locally identifiable sets of concurrent transitions, i.e., of conflict-free sets of transitions.

\begin{defn}
\emph{
Two sets of transitions $R_1$ and $R_2$ are said to be \emph{history-preserving isomorphic} with respect to a pair of transitions $(t_m,t_n)$ if, and only if, there exists a bijection $\mathcal{B}$ between them such that for every $(t_1,t_2) \in \mathcal{B}$, if $t_m \leq t_1$ (resp. $t_m \ominus t_1$) then $t_n \leq t_2$ (resp. $t_n \ominus t_2$).
}
\eos
\end{defn}

Notice that any infinite play of an hpb game where Eve wins always induces a sequence of history-preserving isomorphic sets, where each set is a singleton. This follows from the fact that if this were not the case then Adam could win by choosing a transition in either $R_1$ or $R_2$ such that the hp bisimulation would no longer be synchronous. We are now ready to define thpb games.

\begin{defn}\label{thpbgame}
\emph{
{\bf (Trace history-preserving bisimulation games)} Let the pair $(\pi_1 , \pi_2)$ be a configuration of the game $\mathcal{G}(\mathfrak{T}_1,\mathfrak{T}_2)$. The initial configuration of the game is $(\epsilon,\epsilon)$. There are two players, Eve and Adam, and Adam always plays first and chooses where to play before using any rule of the game. The equivalence relation $R_{thpb}$ is a trace history-preserving (thp) bisimulation, $\sim_{thpb}$, between $\mathfrak{T}_1$ and $\mathfrak{T}_2$ iff it is an hp bisimulation between $\mathfrak{T}_1$ and $\mathfrak{T}_2$ and:
\begin{itemize}
\item (Base case) The initial configuration $(\epsilon,\epsilon)$ is in $R_{thpb}$.
\item ($\sim_{thpb}$ rule). Before Adam chooses a transition using the $\sim_{hpb}$ rule, he can also restrict the set of available transitions by choosing either in $\pi_1$ or $\pi_2$ a maximal trace to be the new set of available choices. Then, Eve must choose a maximal set in the other component of the configuration.
\end{itemize}
$\mathfrak{T}_1 \sim_{thpb} \mathfrak{T}_2$ iff Eve has a winning strategy for the thpb game  $\mathcal{G}(\mathfrak{T}_1,\mathfrak{T}_2)$.}
\eos
\end{defn}


\begin{lem} \label{duptotfl}
If Eve has a winning strategy for every play in the trace history-preserving bisimulation game $\mathcal{G}(\mathfrak{T}_1,\mathfrak{T}_2)$, then ${\mathfrak{T}_1} \sim_{\tfl/} {\mathfrak{T}_2}$.
\end{lem}
\begin{proof}
By contradiction suppose that Eve has a winning strategy and $P \not \sim_{\tfl/} Q$, where $P = (M,t)$ and $Q = (N,r)$ are two processes of ${\mathfrak{T}_1}$ and ${\mathfrak{T}_2}$, respectively. There are two cases. Suppose that Adam cannot make a move. This means that both $P \models^{\mathfrak{T}_1} \modalb{-} \mufalse$ and $Q \models^{\mathfrak{T}_2} \modalb{-} \mufalse$ only, which is a contradiction. The other case is when Eve wins in an infinite play. Since \clmu/ induces an hp bisimilarity and the thpb game conservatively extends the hpb game, w.l.o.g. we can consider only the case when the rule $\sim_{thpb}$ is necessarily played.

Then, let $P \models^{\mathfrak{T}_1} \modald{\otimes} \phi_1$ that, by hypothesis, is not satisfied by $Q$. By the satisfaction relation either $M$ is already a maximal trace or there is a maximal trace $M'$ such that $M' \sqsubseteq M$. Additionally, such a maximal trace cannot be recognised from $N$. However, this is not possible since Eve can always find such a support set by hypothesis. 

Thus, the only other possibility is that the support set can be constructed but a synchronous transition in it cannot be found. But this also leads to a contradiction because the support sets that Eve chooses are, additionally, history-preserving isomorphic to the ones that Adam chooses. Therefore all properties that include $\modald{\otimes}$ must be satisfied at this stage and the the game has to proceed to the next round. However, since the play will continue forever, this holds for all reachable processes, and therefore, all formulae containing $\modald{\otimes}$ that are satisfied in $P$ must also be satisfied in $Q$, which is again a contradiction.
\qed
\end{proof}

\begin{cor}\label{soundtfl}
\emph{{\bf (Soundness)}.} If ${\mathfrak{T}_1} \not \sim_{\tfl/} {\mathfrak{T}_2}$, then Adam has a winning strategy for every play in the trace history-preserving bisimulation game $\mathcal{G}(\mathfrak{T}_1,\mathfrak{T}_2)$. 
\end{cor}

\begin{lem}  \label{tfltodup}
\emph{{\bf (Completeness)}.}
If ${\mathfrak{T}_1} \sim_{\tfl/} {\mathfrak{T}_2}$, then Eve has a winning strategy for every play in the trace history-preserving bisimulation game $\mathcal{G}(\mathfrak{T}_1,\mathfrak{T}_2)$.
\end{lem}
\begin{proof}
By constructing a winning strategy for Eve based on the fact that ${\mathfrak{T}_1} \sim_{\tfl/} {\mathfrak{T}_2}$. 
For the same reasons given previously, w.l.o.g., it is possible to consider only the case when Adam uses the $\sim_{thpb}$ rule.

So, suppose that Adam is able to choose a maximal set $M$ enabled at $P = (M,t)$, where $P$ is a process in the stateless maximal process space $\mathfrak{S}$ associated with $\mathfrak{T}_1$. This implies that $P \models^{\mathfrak{T}_1} \phi$, where $\phi = \modald{\otimes} \phi_1$ for some formula $\phi$ with support set $M$. By the hypothesis, for some process $Q = (N,r)$ that is thp bisimilar to $P$, it must be true that $Q \models^{\mathfrak{T}_2} \phi$ as well, and therefore Eve can choose a maximal set $N$ which is the support set for $\phi$ in $Q = (N,r)$. Since $P \sim_{\tfl/} Q$ then $M$ and $N$ must be history-preserving isomorphic sets with respect to $(t,r)$; otherwise, there would be a simple modal formula differentiating them. 

Then Adam must choose an element of either set of transitions using the $\sim_{hpb}$ rule, say a transition $t' \in M$. But since $M$ and $N$ are history-preserving isomorphic sets with respect to $(t,r)$, then it is always possible for Eve to find a transition $r' \in N$ that synchronises as $t'$, forcing the game to proceed to a next round. The play, therefore, must either go on forever or stop because Adam cannot make a move. In either case Eve wins the game. The dual case is similar since Adam can always choose where to play, i.e., in which structure, before applying any rule of the game.
\qed
\end{proof}

The soundness and completeness results give a full game-theoretical characterisation to the equivalence induced by \tfl/.

\begin{thm}
{\bf (Game abstraction)} ${\mathfrak{T}_1} \sim_{\tfl/} {\mathfrak{T}_2}$ iff Eve has a winning strategy for the thp bisimulation game $\mathcal{G}(\mathfrak{T}_1,\mathfrak{T}_2)$; conversely, ${\mathfrak{T}_1} \not \sim_{\tfl/} {\mathfrak{T}_2}$ iff Adam has a winning strategy for the thp bisimulation game $\mathcal{G}(\mathfrak{T}_1,\mathfrak{T}_2)$.
\end{thm}

\begin{cor}\label{tflequithpb}
$\sim_{\tfl/} \equiv \sim_{thpb}$.
\end{cor}

\subsection{Decidability and Determinacy}

\begin{thm}\label{tthpbdec}
$\sim_{thpb}$ is decidable on finite systems.
\end{thm}
\begin{proof}
As all other bisimulation games presented in this report, thpb games are two-player zero-sum perfect-information (infinite) games whose winning conditions define `Borel' sets, thus they are determined \cite{boreldet-martin}. Alternatively, we can also say that since thb games are sound and complete, then they must be determined. This means that if Eve does not win a play in the game, then Adam must win it. But, since Eve only wins when the two systems are either \tfl/-equivalent, then Adam must win whenever the two systems are not equivalent. 

Moreover, the number of different configurations of any play is always finite, provided that the systems are finite. This follows from the fact that in order to define a winning strategy for Eve one only needs to analyse the locality of the process space where Eve is playing, rather than the whole history of the game, and clearly the number of localities in a finite system is also finite. In particular, notice that given a state of a partial order model there is always only a finite number of processes and support sets relative to such a state. This feature gives a finite size to the set of elements in the first component of a configuration. On the other hand, finiteness of the elements in the second component follows from the bounded branching property of the models since they are also image-finite.

Now, when constructing a winning strategy for Adam it is important to note that only a finite number of processes must be analysed given a particular state of a partial order model. Firstly, notice that a stateless maximal process space embeds the immediate history of a play in the transition component of a process, and so such information is available locally by exploring a finite number of processes given a particular element in the process space.  Secondly, the support sets that a player can choose given a particular process is also finite and can be explored simply by checking all support sets relative to the same either state or support set of the process in the last configuration of the game, i.e., all those in the same neighbourhood. This analysis must be done for all states of the partial order models being compared, but again these sets of states are also finite. 

Finally, since Eve wins when Adam cannot make a move (a finite play easily decided) or when a finite set of repeated configurations is visited infinitely often (for infinite plays, which are won only by Eve), then it is always possible to compute the winning strategies for Eve, and therefore decidability follows.
\qed
\end{proof}

\begin{cor}
$\sim_{\tfl/}$ is decidable.
\end{cor}
\begin{proof}
Follows from Theorem \ref{tthpbdec} and Corollary \ref{tflequithpb}.
\qed
\end{proof}

The previous results let us relate $\sim_{hhpb}$ with $\sim_{\tfl/}$ using game-theoretical arguments. As these bisimulation games are conservative extensions of the hpb game, they can be compared just by looking at their additional rules with respect to the hpb game. Then, only by showing that the additional rule for the hhpb game is at least as powerful as the additional rule for the thpb game, and taking into account that, by Proposition \ref{tflnotinhhpb}, ${\sim_{hhpb}}$ and ${\sim_{\tfl/}}$ do not coincide, we have:

\begin{thm} \label{hhpbtosfl}
${\sim_{hhpb}} \subset {\sim_{\tfl/}}$. 
\end{thm}
\begin{proof}
Suppose that two systems are hhp bisimilar. We show that they must be thp bisimilar as well. The additional rule for the thpb game simply lets Adam choose a set of runs to be checked by committing a choice made by Eve. Suppose the current configuration of the game is $(\pi_1,\pi_2)$. When Adam uses the additional rule of the thpb game, the extended runs that can be checked by Adam have the form $\pi' = \pi_1.w_1.\alpha$ and $\pi'' = \pi_2.w_2.\beta$, where $w_1$ and $w_2$ are the sequences of transitions in the conflict-free isomorphic sets $M$ and $N$ that Adam and Eve have chosen, and $\alpha$ and $\beta$ are any sequence of transitions after them which are not in $M$ and not in $N$, respectively. Now, since $M$ and $N$ only contain conflict-free transitions at $\pi_1$ and $\pi_2$, i.e., concurrent transitions at $\pi_1$ and $\pi_2$, they are all backwards enabled whenever Adam decides to choose them, and therefore they can also be checked by the additional rule of the hhpb game.
\qed
\end{proof}

This result allows us to define a new hierarchy of equivalences and games for concurrency. The new hierarchy we define here extends the work of Fecher \cite{hier-fecher}.

\subsection{A New Hierarchy of Equivalences for Concurrency}\label{hier}
An interesting problem is that of having logics capturing bisimulation equivalences for concurrency. A hierarchy of so-called `true concurrent' equivalences can be found in \cite{hier-fecher}. Our results define a new hierarchy of equivalences for concurrent systems, where the bisimilarities induced by (the fixpoint-free fragment of) \tfl/ rank above all the decidable equivalences in such a hierarchy.

Prior to this work, we had that $\sim_{hhpb}$ is captured by the Path Logic (PL) of Nielsen and Clausen \cite{pathlogic-nielsen}, as well as the well-known result by Milner and Hennessy that $\sim_{sb}$ is captured by HML \cite{hmljacm-milner}, or what is equivalent, by the fixpoint-free fragment of the mu-calculus. Moreover, in a preliminary work to the one reported here the author also studied another fixpoint logic quite similar to \tfl/. Such a logic, which is called Separation Fixpoint Logic (SFL \cite{fos09-gut}) also induces a bisimulation equivalence strictly stronger than hpb and strictly weaker than hhpb on $\Xi$-systems; such an equivalence is studied in \cite{fos09-gut} under the name of `independence' hpb, $\sim_{ihpb}$. The exact relation between $\sim_{\tfl/}$ and $\sim_{ihpb}$ (or equivalently $\sim_{SFL}$, the bisimilarity induced by SFL) is however still unknown.

Just to recall, here we have shown that \clmu/ captures $\sim_{hpb}$, the standard equivalence for causal systems. Moreover, a new equivalence was introduced and shown to be decidable and strictly between $\sim_{hpb}$ and $\sim_{hhpb}$ in terms of discriminating power. In Figure \ref{hierfig}, $\sim_{PL}$ represents the bisimulation equivalence induced by PL; moreover, $\sim_{eq}$ refers to several other equivalences for concurrency, which are not studied in this document. The original hierarchy can be found in \cite{hier-fecher}.
\begin{figure}[]
  \begin{center}
    \subfigure
    {
\xymatrix @R=10pt @C=10pt 
{
& \sim_{hhpb} \equiv \sim_{PL} \ar[dl] \ar[dr] &  \\
\sim_{ihpb} \equiv \sim_{SFL} \ar[dr] &  & \sim_{thpb} \equiv \sim_{\tfl/} \ar[dl] \\
& \sim_{hpb} \equiv \sim_{\clmu/} \ar[d] & \\
\ar@{--}[dd]&  \ar@{--}[l] \ar@{--}[r] & \ar@{--}[dd]\\
  & \sim_{eq} &   \\
\ar@{--}[r] &  \ar[d] & \ar@{--}[l] \\
& \sim_{sb} \equiv \sim_{HML}  & \\
} 
    }
  \end{center}
\vspace{-0.5cm}
\caption{A hierarchy of equivalences for concurrency. The arrow $\rightarrow$ means inclusion $\subset$.}
\label{hierfig}
\end{figure}
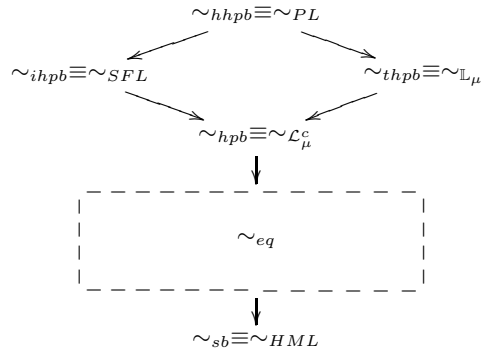

\section{Higher-Order Logic Games for Model-Checking} \label{mcgames}
In this section we introduce higher-order logic games for model-checking that allow \emph{local} second-order power on sets of independent transitions in the underlying partial order models where the games are played. Since the interleaving semantics of such models is not considered, some problems that may arise when using interleaving representations are avoided and new decidability results for partial order models are achieved. The games are shown to be \emph{sound} and \emph{complete}, and therefore determined. While in the interleaving case they coincide with the local model-checking games for the mu-calculus, in a partial order setting they verify properties of a number of fixpoint modal logics that can specify in concurrent systems with partial order semantics, several properties not expressible with \lmu/. 

The games underpin a novel decision procedure for model-checking all temporal properties of a class of infinite and regular event structures, thus improving previous results in the literature. As said before, similar to the case of higher-order logic games for bisimulation, the players in this new game are given local monadic second-order power on conflict-free sets of transitions. The technical details behind the construction of this game follows seminal ideas on local model-checking games for \lmu/ as presented by Stirling \cite{localmc-stirling}. In particular, this higher-order logic game for model-checking is a simple refinement of the model-checking procedure first defined in \cite{con09-gut} by Bradfield and the author.

A higher-order logic game for model-checking $\mathcal{G}(\mathfrak{M},\phi)$ is
played on a model $\mathfrak{M} = (\mathfrak{T},\mathcal{V})$, where
$\mathfrak{T} = (S,s_0,T,I,\Sigma)$ is a system, and on a formula
$\phi$. The game can also be presented as
$\mathcal{G}_{\mathfrak{M}}(H_0,\phi)$, or even as
$\mathcal{G}(H_0,\phi)$ or $\mathcal{G}(s_0,\phi)$, where $H_0 =
(\mathfrak{X}(s_0),t_{\epsilon})$ is the initial process of
$\mathfrak{S}$ in the model $\mathfrak{M}$. The board in which the game is played has the form
$\mathfrak{B} \subseteq \mathfrak{S} \times \Sub(\phi)$, for a process space
$\mathfrak{S} = \mathcal{X} \times \mathfrak{A}$ of 
support sets $R \in \mathcal{X}$ and transitions $t \in \mathfrak{A}$ in the system $\mathfrak{T}$, and where $\Sub (\phi)$ is the subformula set of $\phi$ as given by the Fischer--Ladner closure associated with the logic to which $\phi$ belongs. Since we want to define model-checking games for \tfl/ we now present the Fischer--Ladner closure of \tfl/ formulae.

\begin{defn}
\emph{{\bf (Fischer--Ladner closure of \tfl/ formulae)}
The `subformulae' or subformula set $Sub(\phi)$ of a \tfl/ formula $\phi$ is given in the following way:
\begin{center}
\begin{tabular}{lll}
$Sub(Z)$                     &  $=$ & $\{Z\}$ \\
$Sub(\phi_1 \wedge \phi_2)$  &  $=$ & $\{ \phi_1 \wedge \phi_2 \} \cup Sub(\phi_1) \cup Sub(\phi_2)$ \\
$Sub(\phi_1 \vee \phi_2)$    &  $=$ & $\{ \phi_1 \vee \phi_2 \} \cup Sub(\phi_1) \cup Sub(\phi_2)$ \\
$Sub(\modald{a}_c \phi_1)$   &  $=$ & $\{\modald{a}_c \phi_1 \} \cup Sub(\phi_1)$ \\
$Sub(\modald{a}_{nc} \phi_1)$   &  $=$ & $\{\modald{a}_{nc} \phi_1 \} \cup Sub(\phi_1)$ \\
$Sub(\modalb{a}_c \phi_1)$   &  $=$ & $\{\modalb{a}_c \phi_1 \} \cup Sub(\phi_1)$ \\
$Sub(\modalb{a}_{nc} \phi_1)$   &  $=$ & $\{\modalb{a}_{nc} \phi_1 \} \cup Sub(\phi_1)$ \\
$Sub(\modald{\otimes} \phi_1)$   &  $=$ & $\{\modald{\otimes} \phi_1 \} \cup Sub(\phi_1)$ \\
$Sub(\modalb{\otimes} \phi_1)$   &  $=$ & $\{\modalb{\otimes} \phi_1 \} \cup Sub(\phi_1)$ \\
$Sub(\mu Z . \phi_1)$          &  $=$ & $\{\mu Z . \phi_1 \} \cup Sub(\phi_1)$ \\
$Sub(\nu Z . \phi_1)$          &  $=$ & $\{\nu Z . \phi_1 \} \cup Sub(\phi_1)$ \\
\end{tabular} 
\end{center}
}
\eos
\end{defn}

A play is a possibly infinite sequence of configurations $C_0, C_1, ...$, written as $(R,t) \vdash \phi$ or $H \vdash \phi$ whenever possible; each $C_i$ is an element of the board $\mathfrak{B}$. Every play starts in the configuration $C_0 = H_0 \vdash \phi$, and proceeds according to the rules of the game given in Figure \ref{fig:mcrules}. As usual for model-checking games, Eve tries to prove that $H_0 \models \phi$ whereas Adam tries to show that $H_0 \not \models \phi$.  

\begin{figure}[]
\begin{center}
\small
\begin{tabular}{|lclclcl|}
\hline
\vspace{-0.3cm}
&&&&&& \\
\multicolumn{7}{|c|}{(FP) $~~\myfrac{H \vdash {_{\mu}^{\nu}Z} . \phi}{H \vdash Z}~~$ s.t.\ ${_{\mu}^{\nu}} \in \{\mu,\nu\}$}  \\
\vspace{-0.2cm}
&&&&&& \\
\multicolumn{7}{|c|}{(VAR) $~~\myfrac{H \vdash Z}{H \vdash \phi}~~$ s.t.\ ${_{\mu}^{\nu}Z} . \phi$ } \\ 
\vspace{-0.3cm}
&&&&&& \\
\hline
\vspace{-0.3cm}
&&&&&& \\
\multicolumn{7}{|c|}{($\vee$) $~~\myfrac{H \vdash \phi_0 \vee \phi_1}{H \vdash \phi_i}~~$ $\left[\exists\right] i : ~~ i \in \{0,1\}$} \\
\vspace{-0.2cm}
&&&&&& \\
\multicolumn{7}{|c|}{($\wedge$) $~~\myfrac{H \vdash \phi_0 \wedge \phi_1}{H \vdash \phi_i}~~$ $\left[\forall\right]i: ~~ i \in \{0,1\}$} \\
\vspace{-0.3cm}
&&&&&& \\
\hline
\vspace{-0.3cm}
&&&&&& \\
\multicolumn{7}{|c|}{($\modald{~}_c$) $~~\myfrac{(R,t) \vdash \modald{a}_c \phi}{(\mathfrak{X},r) \vdash \phi}~~$ $\left[\exists\right] a: ~~ a = \delta(r), ~ r \in R, ~ t \leq r$}\\
\vspace{-0.2cm}
&&&&&& \\
\multicolumn{7}{|c|}{($\modald{~}_{nc}$) $~~\myfrac{(R,t) \vdash \modald{a}_{nc} \phi}{(\mathfrak{X},r) \vdash \phi}~~$ $\left[\exists\right] a: ~~ a = \delta(r), ~ r \in R, ~ t \ominus r$}\\
\vspace{-0.2cm}
&&&&&& \\
\multicolumn{7}{|c|}{($\modalb{~}_c$) $~~\myfrac{(R,t) \vdash \modalb{a}_c \phi}{(\mathfrak{X},r) \vdash \phi}~~$ $\left[\forall\right] a: ~~ a = \delta(r), ~ r \in R, ~ t \leq r$}\\
\vspace{-0.2cm}
&&&&&& \\
\multicolumn{7}{|c|}{($\modalb{~}_{nc}$) $~~\myfrac{(R,t) \vdash \modalb{a}_{nc} \phi}{(\mathfrak{X},r) \vdash \phi}~~$ $\left[\forall\right] a: ~~ a = \delta(r), ~ r \in R, ~ t \ominus r$}\\
\vspace{-0.3cm}
&&&&&& \\
\hline
\vspace{-0.3cm}
&&&&&& \\
\multicolumn{7}{|c|}{($\modald{\otimes}$) $~~\myfrac{(R,t) \vdash \modald{\otimes} \phi}{(M,t) \vdash \phi}~~$ $\left[\exists\right] M: ~~  {M \sqsubseteq R}$}\\
\vspace{-0.2cm}
&&&&&& \\
\multicolumn{7}{|c|}{($\modalb{\otimes}$) $~~\myfrac{(R,t) \vdash \modalb{\otimes} \phi}{(M,t) \vdash \phi}~~$ $\left[\forall\right] M: ~~  {M \sqsubseteq R}$}\\
\hline
\end{tabular} 
\end{center}
\caption{Model-checking game rules of \tfl/. Whereas the notation $\left[\forall\right]$ denotes a choice made by Adam, the notation $\left[\exists\right]$ denotes a choice by Eve; $\mathfrak{X}$ is the maximal set at $\tau(r)$.}
\label{fig:mcrules}
\end{figure} 

The rules (FP) and (VAR) control the unfolding of fixpoint operators. Their correctness is based on the fact that ${{_{\mu}^{\nu}}Z} . \phi \equiv \phi \left[ {{_{\mu}^{\nu}}Z} . \phi / Z \right]$, where ${_{\mu}^{\nu}} \in \{\mu,\nu\}$, according to the semantics of the logic. Rules ($\vee$) and ($\wedge$) have the same meaning as the disjunction and conjunction rules, respectively, in a Hintikka game for propositional logic. Rules ($\modald{~}_c$), ($\modald{~}_{nc}$), ($\modalb{~}_c$) and ($\modalb{~}_{nc}$) are like the rules for quantifiers in a standard Hintikka game semantics for first-order (FO) logic, provided that the box and diamond operators behave, respectively, as restricted universal and existential quantifiers sensitive to the causal information in the partial order model.

Finally, the most interesting rules are ($\modald{\otimes}$) and ($\modalb{\otimes}$). Local monadic second-order moves are used to recognize conflict-free sets of transitions in $\mathfrak{M}$, i.e., those in the same `trace'. The use of $\modald{\otimes}$ and $\modalb{\otimes}$ requires a player to make a choice, locally, on a set of transitions rather than on a singleton set as in the traditional games for model-checking. However, such a second-order power is restricted to some conflict-free sets of transitions; more specifically, to maximal traces.

Guided by the semantics of $\modald{\otimes}$ (resp.\ $\modalb{\otimes}$), it is defined that Eve (resp.\ Adam) must look for a maximal set $M$ to be assigned to $\phi_1$ as its support set.

\begin{defn} \label{wc}
\emph{
The following rules are the \emph{winning conditions} that determine a unique winner for every finite or infinite play $C_0, C_1, ... $ in a game $\mathcal{G}(H_0,\phi)$.
}
\end{defn}

Adam wins a finite play $C_0, C_1, ..., C_n$ or an infinite play $C_0, C_1, ...$ iff:
\begin{enumerate}
\item $C_n = H \vdash Z$ and $H \not \in \mathcal{V}(Z)$.
\item $C_n = (R,t) \vdash \modald{a}_{c} \psi$ and $\{(\mathfrak{X},r) : {t \leq {r = {s \xrightarrow{a} s'} \in R}} \} = \emptyset$.
\item $C_n = (R,t) \vdash \modald{a}_{nc} \psi$ and $\{(\mathfrak{X},r) : t \ominus r = s \xrightarrow{a} s' \in R \} = \emptyset$.
\item $C_n = (R,t) \vdash \modald{\otimes} \psi$ and $\{(M,t) : M \sqsubseteq R \} = \emptyset$.
\item The play is infinite and there are infinitely many configurations where $Z$ appears, such that $Z$ is the least fixpoint of some subformula $\mu Z . \psi$ and the syntactically outermost variable in $\phi$ that occurs infinitely often. 
\end{enumerate}

Eve wins a finite play $C_0, C_1, ..., C_n$ or an infinite play $C_0, C_1, ...$ iff:
\begin{enumerate}
\item $C_n = H \vdash Z$ and $H \in \mathcal{V}(Z)$.
\item $C_n = (R,t) \vdash \modalb{a}_{c} \psi$ and $\{(\mathfrak{X},r) : {t \leq {r = {s \xrightarrow{a} s'} \in R}} \} = \emptyset$.
\item $C_n = (R,t) \vdash \modalb{a}_{nc} \psi$ and $\{(\mathfrak{X},r) : t \ominus r = s \xrightarrow{a} s' \in R \} = \emptyset$.
\item $C_n = (R,t) \vdash \modalb{\otimes} \psi$ and $\{(M,t) : M \sqsubseteq R \} = \emptyset$.
\item The play is infinite and there are infinitely many configurations where $Z$ appears, such that $Z$ is the greatest fixpoint of some subformula $\nu Z . \psi$ and the syntactically outermost variable in $\phi$ that occurs infinitely often. 
\eos
\end{enumerate}

\subsection{Soundness and Completeness} \label{soundcomplete}
Let us first give some intermediate results. The statements in this section are all either standard modal mu-calculus statements, or standard statements where additional cases for the new operators of \tfl/ need to be checked. We give the statements in full, and the usual proof outlines, for the sake of being self-contained.

Let $\mathfrak{T}$ be a system and $C = (R,t) \vdash \psi$ a configuration in the game $\mathcal{G}(H_0,\phi)$, as defined before. As usual, the denotation $\denot{\phi}^{\mathfrak{T}}_{\mathcal{V}}$ of a \tfl/ formula $\phi$ in the model $\mathfrak{M} = (\mathfrak{T},\mathcal{V})$ is a subset of $\mathfrak{S}$. We say that a configuration $C$ of $\mathcal{G}(H_0,\phi)$ is \emph{true} if, and only if, we have that $(R,t) \in \denot{\psi}^{\mathfrak{T}}_{\mathcal{V}}$, and \emph{false} otherwise.

\begin{fact}\label{negclosed}
\tfl/ is closed under negation.
\end{fact}

\begin{lem}\label{dualgames}
A game $\mathcal{G}(H_0,\phi)$, where Eve has a winning strategy, has a dual game $\mathcal{G}(H_0,\neg \phi)$ where Adam has a winning strategy, and conversely.
\end{lem}
\begin{proof}
First, note that since \tfl/ is closed under negation, for every rule that requires a player to make a choice on a formula $\psi$ there is a dual rule in which the other player makes a choice on the negated formula $\neg \psi$. Also, note that for every winning condition for one of the players in a formula $\psi$ there is a dual winning condition for the other player in $\neg \psi$. Now, suppose Eve has a winning strategy $\pi$ in the game $\mathcal{G}(H_0,\phi)$. Adam can use $\pi$ in the dual game $\mathcal{G}(H_0,\neg \phi)$ since whenever he has to make a choice, by duality, there is a rule that requires Eve to make a choice in $\mathcal{G}(H_0,\phi)$. In this way, regardless of the choices that Eve makes, Adam can enforce a winning play for himself. The case when Adam has a winning strategy in the game $\mathcal{G}(H_0,\phi)$ is dual.
\qed
\end{proof}

\begin{lem}\label{etrueafalse}
Eve preserves falsity and can preserve truth with her choices. Hence, she cannot choose true configurations when playing in a false configuration. Dually, Adam preserves truth and can preserve falsity with his choices. Then, he cannot choose false configurations when playing in a true configuration.
\end{lem}
\begin{proof}
The cases for the rules ($\wedge$) and ($\vee$) are just as for the Hintikka evaluation games for FO logic. Thus, let us go on to check the rules for the other operators. Firstly, consider the rule $(\modald{~}_c)$ and a configuration $C = (R,t) \vdash \modald{a}_c \psi$, and suppose that $C$ is false. In this case there is no $a$ such that $t \leq {r \in R}$ and $(\mathfrak{X},r) \in \denot{\psi}^{\mathfrak{T}}_{\mathcal{V}}$, where as usual $r$ is some transition $s \xrightarrow{a} q$ and $\mathfrak{X}$ is the maximal set at $\tau(r)$. Hence, the following configurations will be false as well. Contrarily, if $C$ is true, then Eve can make the next configuration $(\mathfrak{X},r) \vdash \psi$ true by choosing a transition $r = s \xrightarrow{a} s' \in R$ such that $t \leq r$. The case for $(\modald{~}_{nc})$ is similar (simply change $\leq$ for $\ominus$), and the cases for $(\modalb{~}_c)$ and $(\modalb{~}_{nc})$ are dual. Now, consider the rule $\modald{\otimes}$ and a configuration $C = (R,t) \vdash \modald{\otimes} \psi$, and suppose that $C$ is false. In this case there is no maximal trace $M$ such that $M \sqsubseteq R$ and $(M,t) \in \denot{\psi}^{\mathfrak{T}}_{\mathcal{V}}$, and hence Eve preserves falsity since the next configuration must be false as well. On the other hand, if $C$ is true, then Eve can make the next configuration true as well by choosing a maximal set $M$ such that $(M,t) \vdash \psi$ is true as well. Finally, the deterministic rules (FP) and (VAR) preserve both truth and falsity because of the semantics of fixpoint operators. Recall that for any process $H$, if $H \in \denot{_{\mu}^{\nu} Z . \psi}$ then $H \in \denot{\psi}_{Z:=\denot{_{\mu}^{\nu} Z . \psi}}$ for all free variables $Z$ in $\psi$. 
\qed
\end{proof}

\begin{lem} \label{outermost}
In any infinite play of a game $\mathcal{G}(H_0,\phi)$ there is a unique syntactically outermost variable that occurs infinitely often.  
\end{lem}
\begin{proof}
By contradiction, assume that the statement is false. Without loss of generality, suppose that there are two variables $Z$ and $Y$ that are syntactically outermost and appear infinitely often. The only possibility for this to happen is that $Z$ and $Y$ are at the same level in $\phi$. However, if this is the case $Z$ and $Y$ cannot occur infinitely often unless there is another variable $X$ that also occurs infinitely often and whose unfolding contains both $Z$ and $Y$. But this means that both $Z$ and $Y$ are syntactically beneath $X$, and therefore neither $Z$ nor $Y$ is outermost in $\phi$, which is a contradiction.
\qed
\end{proof}

\begin{fact} \label{sizec}
Only rule \emph{(VAR)} can increase the size of a formula in a configuration. All other rules decrease the size of formulae in configurations.
\end{fact}

\begin{lem} \label{uwinner}
Every play of a game $\mathcal{G}(H_0,\phi)$ has a uniquely determined winner.
\end{lem}
\begin{proof}
Suppose the play is of finite length. Then, the winner is uniquely determined by one of the winning conditions one to four (Definition \ref{wc}) of either Eve or Adam since such rules cover all possible cases and are mutually exclusive. Now, suppose that the play is of infinite length. Due to Fact \ref{sizec}, rule (VAR) must be used infinitely often in the game, and thus, there is at least one variable that is replaced by its defining fixpoint formula each time it occurs. Therefore, winning condition five of one of the players can be used to uniquely determine the winner of the game since, due to Lemma \ref{outermost}, there is a unique syntactically outermost variable that occurs infinitely often.
\qed
\end{proof}

\begin{defn}\label{approx}
\emph{
\textbf{(Approximants)} Let $Z$ be the least fixpoint of some formula $\phi$ and let $\alpha , \lambda \in \mathbb{O}\mathrm{rd}$ be two ordinals, where $\lambda$ is a limit ordinal. Then:
\begin{center}
\begin{tabular}{ccc}
$Z^0 := \mufalse$, &$~~~~~~~$ $Z^{\alpha + 1} = \phi\left[Z^{\alpha} / Z \right]$, &$~~~~~~~$ $Z^{\lambda} = \bigvee_{\alpha < \lambda} Z^{\alpha}$
\end{tabular}
\end{center}
For greatest fixpoints the approximants are defined dually. Let $Z$ be the greatest fixpoint of some formula $\phi$ and, as before, let $\alpha , \lambda \in \mathbb{O}\mathrm{rd}$ be two ordinals, where $\lambda$ is a limit ordinal. Then:
\begin{center}
\begin{tabular}{ccc}
$Z^0 := \mutrue$, &$~~~~~~~$ $Z^{\alpha + 1} = \phi\left[Z^{\alpha} / Z \right]$, &$~~~~~~~$ $Z^{\lambda} = \bigwedge_{\alpha < \lambda} Z^{\alpha}$
\end{tabular}
\end{center}
}
\vspace{-1.1cm}
\eos 
\end{defn}

We can now show that the analysis for fixpoint modal logics \cite{mucalculi-bradfield} can be extended to this scenario. The proof of soundness uses similar arguments to that in the mu-calculus case, but we present it here in full because it is the basis of the decision procedure for \tfl/ model-checking.

\begin{thm}\label{soundness}
\emph{\textbf{(Soundness)}} Let $\mathfrak{M} = (\mathfrak{T},\mathcal{V})$ be a model of a formula $\phi$ in the game $\mathcal{G}(H_0,\phi)$. If $H_0 \not \in \denot{\phi}^{\mathfrak{T}}_{\mathcal{V}}$ then Adam wins $H_0 \vdash \phi$.
\end{thm}
\begin{proof}
Suppose $H_0 \not \in \denot{\phi}^{\mathfrak{T}}_{\mathcal{V}}$. We construct for Adam a possibly infinite game tree that starts in $H_0 \vdash \phi$. We do so by preserving falsity according to Lemma \ref{etrueafalse}, i.e., whenever a rule requires Adam to make a choice then the tree will contain the successor configuration that preserves falsity. All other choices that are available for Eve are included in the game tree. 

First, consider only finite plays. Since Eve only wins finite plays that end in true configurations, then she cannot win any finite play by using her winning conditions one to four. Hence, Adam wins each finite play in this game tree. 

Now, consider infinite plays. The only chance for Eve to win is to use her winning condition five. So, let the configuration $H \vdash \nu Z . \phi$ be reached such that $Z$ is the syntactically outermost variable that appears infinitely often in the play according to Lemma \ref{outermost}. In the next configuration $H \vdash Z$, variable $Z$ is interpreted as the least approximant $Z^{\alpha}$ such that $H \not \in \denot{Z^{\alpha}}^{\mathfrak{T}}_{\mathcal{V}}$ and $H \in \denot{Z^{\alpha-1}}^{\mathfrak{T}}_{\mathcal{V}}$, by the principle of fixpoint induction. As a matter of fact, by monotonicity and due to the definition of fixpoint approximants it must also be true that $H \in \denot{Z^{\beta}}^{\mathfrak{T}}_{\mathcal{V}}$ for all ordinals $\beta$ such that ${\beta} < {\alpha}$. Note that, also due to the definition of fixpoint approximants, $\alpha$ cannot be a limit ordinal $\lambda$ because this would mean that $H \not \in \denot{Z^{\lambda} = \bigwedge_{\beta < \lambda} Z^{\beta}}^{\mathfrak{T}}_{\mathcal{V}}$ and $H \in \denot{Z^{\beta}}^{\mathfrak{T}}_{\mathcal{V}}$ for all $\beta < \lambda$, which is impossible.

Since $Z$ is the outermost variable that occurs infinitely often and the game rules follow the syntactic structure of formulae, the next time that a configuration $C' = H' \vdash Z$ is reached, $Z$ can be interpreted as $Z^{\alpha-1}$ in order to make $C'$ false as well. And again, if $\alpha-1$ is a limit ordinal $\lambda$, there must be a $\gamma < \lambda$ such that $H' \not \in \denot{Z^{\gamma}}^{\mathfrak{T}}_{\mathcal{V}}$ and $H' \in \denot{Z^{\gamma-1}}^{\mathfrak{T}}_{\mathcal{V}}$. One can repeat this process even until $\lambda = \omega$.

But, since ordinals are well-founded the play must eventually reach a false configuration $C'' = H'' \vdash Z$ where $Z$ is interpreted as $Z^0$. And, according to Definition \ref{approx}, $Z^0 := \mutrue$, which leads to a contradiction since the configuration $C'' = H'' \vdash \mutrue$ should be false, i.e., $H'' \in \denot{\mutrue}^{\mathfrak{T}}_{\mathcal{V}}$ should be false, which is impossible. In other words, if $H$ had failed a maximal fixpoint, then there must have been a descending chain of failures, but, as can be seen, there is not. 

As a consequence, there is no such least $\alpha$ that makes the configuration $H \vdash Z^{\alpha}$ false, and hence, the configuration $H \vdash \nu Z . \phi$ could not have been false either. Therefore, Eve cannot win any infinite play with her winning condition 5 either. Since Eve can win neither finite plays nor infinite ones whenever  $H_0 \not \in \denot{\phi}^{\mathfrak{T}}_{\mathcal{V}}$, then Adam must win all plays of $\mathcal{G}(H_0,\phi)$. 
\qed
\end{proof}

\begin{rem}
\emph{
If only finite state systems are considered $\mathbb{O}\mathrm{rd}$, the set of ordinals, can be replaced by $\mathbb{N}$, the set of natural numbers.
}
\eos
\end{rem}

\begin{thm}\label{completeness}
\emph{\textbf{(Completeness)}} Let $\mathfrak{M} = (\mathfrak{T},\mathcal{V})$ be a model of a formula $\phi$ in the game $\mathcal{G}(H_0,\phi)$ . If $H_0 \in \denot{\phi}^{\mathfrak{T}}_{\mathcal{V}}$ then Eve wins $H_0 \vdash \phi$. 
\end{thm}
\begin{proof}
Suppose that $H_0 \in \denot{\phi}^{\mathfrak{T}}_{\mathcal{V}}$. Due to Fact \ref{negclosed} it is also true that $H_0 \not \in \denot{\neg \phi}^{\mathfrak{T}}_{\mathcal{V}}$. According to Theorem \ref{soundness}, Adam wins $H_0 \vdash \neg \phi$, i.e., has a winning strategy in the game $\mathcal{G}(H_0,\neg \phi)$. And, due to Lemma \ref{dualgames}, Eve has a winning strategy in the dual game $\mathcal{G}(H_0,\phi)$. Therefore, Eve wins $H_0 \vdash \phi$ if $H_0 \in \denot{\phi}^{\mathfrak{T}}_{\mathcal{V}}$.
\qed
\end{proof}

Theorems \ref{soundness} and \ref{completeness} imply that the game is determined. Determinacy and perfect information make the notion of truth defined by this Hintikka game semantics coincide with its Tarskian counterpart.

\begin{cor}\label{determinacy}
\emph{\textbf{(Determinacy)}} For every model-checking game $\mathcal{G}(H_0,\phi)$ either Adam or Eve has a winning strategy to win all plays in the game.
\end{cor}

\subsection{Local Properties and Decidability} \label{locanddec}

We have shown that the higher-order logic games designed to model-check \tfl/ properties of concurrent systems with partial order semantics are still sound and complete even when players are allowed to manipulate sets of independent transitions. Importantly, the power of these games, and also of \tfl/, is that such a second-order quantification is kept both \emph{local} and restricted to transitions in the same \emph{trace}. We now show that such model-checking games enjoy several local properties that in turn make them \emph{decidable} in the finite case. Such a decidability result is used later on to extend the decidability border of model-checking a category of partial order models of concurrency.

\begin{prop}\label{winstr}
\emph{\textbf{(Winning strategies)}} The winning strategies in the higher-order logic games for model-checking \tfl/ properties are history-free. 
\end{prop}
\begin{proof}
Consider a winning strategy $\pi$ for Eve. According to Lemma \ref{etrueafalse} and Theorem \ref{completeness} such a strategy consists of preserving truth with her choices and annotating variables with their approximant indices. But neither of these two tasks depends on the history of a play. Instead they only depend on the current configuration of the game. In particular notice that, of course, this is also the case for the structural operators since the second-order quantification has only a local scope. Similar arguments apply for the winning strategies of Adam.
\qed
\end{proof}

This result is key to achieve decidability of these games in the presence of the local second-order quantification on the traces of the partial order models we consider. Also, from a more practical standpoint, memoryless strategies are desirable as they are easier to synthesize. However, synthesis is not studied here.

\begin{thm}\label{dectflgames}
\emph{\textbf{(Decidability)}} The model-checking game for finite systems against \tfl/ specifications is decidable.
\end{thm}
\begin{proof}
Since the game is determined, finite plays are decided by winning conditions one to four of either player. Now consider the case of plays of infinite length; since the winning strategies of both players are history-free, we only need to look at the set of different configurations in the game, which is finite even for plays of infinite length because of the following simple argument. In a finite system an infinite play can only be possible if the model is cyclic. But, since the model has a finite number of states, there is an upper bound on the number of fixpoint approximants that must be calculated (as well as on the number of configurations of the game board that must be checked) in order to ensure that either a greatest fixpoint is satisfied or a least fixpoint has failed. As a consequence, all possible history-free winning strategies for a play of infinite length can be computed, so that the game can be decided using winning condition five of one of the players. Note that in general computing (pure) history-free winning strategies of determined games played in a finite system is always possible.
\qed
\end{proof}

\subsubsection*{The Interleaving Case.}
Local properties of this higher-order games can also be found in the interleaving case, namely, they coincide with the local model-checking games for the modal mu-calculus defined by Stirling and presented in Section \ref{pre}. Note that interleaving concurrency can be cast using \tfl/ by both syntactic and semantic means. In the former case this is done by considering only its \lmu/ fragment, whereas in the latter case this is achieved by considering directly the (one-step) interleaving semantics of a concurrent system.

The importance of this feature of \tfl/ is that even having constructs for independence and a natural partial order semantics, nothing is lost with respect to the main local approaches to interleaving concurrency. Let us recall that \lmu/ can be obtained from \tfl/ by considering the $\modalt$-free language and using only HML modalities, i.e., making a strict use of the following abbreviations: $\modald{a} \phi = \modald{a}_{c} \phi \vee \modald{a}_{nc} \phi$ and $\modalb{a} \phi = \modalb{a}_{c} \phi \wedge \modalb{a}_{nc} \phi$.

\begin{prop} \label{tmsomcgamesvslmcmugames}
If either a model with an empty independence relation or the syntactic \lmu/ fragment of \tfl/ is considered, then the higher-order model-checking games for \tfl/ degenerate to the local model-checking games for \lmu/. 
\end{prop}
\begin{proof}
Let us consider the case when the syntactic \lmu/ fragment of \tfl/ is considered. The first observation to be made is that the $\{\modald{\otimes},\modalb{\otimes}\}$-free fragment of \tfl/ only considers maximal sets. Hence if a transition can be performed then it is always in the support set currently available. Therefore, support sets in $\mathcal{X}$ can be disregarded and only the set of states of the systems are needed. Also, without loss of generality, consider only the case of the modal operators since the \lmu/ and \tfl/ boolean and fixpoint operators have the same denotation. 
\begin{center}
\begin{tabular}{lll}
$\denot{\modald{a} \phi}^{\mathfrak{T}}_{\mathcal{V}}$ &  $=$ & $\{ (s,t) \in S \times \mathfrak{A} \mid \exists q \in S . ~ t \leq r = s \xrightarrow{a} q \wedge (q,r) \in \denot{\phi}^{\mathfrak{T}}_{\mathcal{V}} \} $\\ 
&  & $\hbox to 0pt{\hss${}\cup{}$}\{ (s,t) \in S \times \mathfrak{A} \mid \exists q \in S . ~  t \ominus r = s \xrightarrow{a} q \wedge (q,r) \in \denot{\phi}^{\mathfrak{T}}_{\mathcal{V}} \} $ \\
\end{tabular} 
\end{center}
The second observation is that when computing the semantics of the combined operator $\modald{a}$, the conditions $t \leq r$, i.e., $(t,r) \not \in I$, and $t \ominus r$, i.e., $(t,r) \in I$, complement each other and become always true since there are no other possibilities. Thus, the second component of every pair in $S \times \mathfrak{A}$ can also be disregarded.
\begin{center}
\begin{tabular}{lll}
$\denot{\modald{a} \phi}^{\mathfrak{T}}_{\mathcal{V}}$ &  $=$ & $\{ s \in S \mid \exists q \in S . ~ s \xrightarrow{a} q \wedge q \in \denot{\phi}_{\mathcal{V}} \}$ \\
\end{tabular} 
\end{center}
The case for the box operator $\modalb{a}$ is similar. Now, note that the new game rules and winning conditions enforced by these restrictions coincide with the ones defined by Stirling for the local model-checking games of \lmu/. In particular, the new game rules and winning conditions for the modalities are as follows.

In a finite play $C_0, C_1, ..., C_n$ of $\mathcal{G}(H_0,\phi)$, where $C_n$ has a modality as a formula component, Adam wins iff $C_n = s \vdash \modald{a} \psi$ and $\{q : s \xrightarrow{a} q \} = \emptyset$, and Eve wins iff $C_n = s \vdash \modalb{a} \psi$ and $\{q : s \xrightarrow{a} q \} = \emptyset$. Since winning conditions for infinite plays do not depend on modalities, they remain the same. Furthermore, the game rules for modal operators reduce to:
\begin{center}
\begin{tabular}{lcl}
\begin{footnotesize}($\modald{~}$)\end{footnotesize} $~~\myfrac{s \vdash \modald{a} \phi}{q \vdash \phi}~~$      \begin{footnotesize}$[\exists] a: ~~ s \xrightarrow{a} q$\end{footnotesize}& $~~~$ & 
\begin{footnotesize}($\modalb{~}$)\end{footnotesize} $~~\myfrac{s \vdash \modalb{a} \phi}{q \vdash \phi}~~$      \begin{footnotesize}$[\forall] a: ~~ s \xrightarrow{a} q$\end{footnotesize}
\end{tabular} 
\end{center}
Clearly, the games just defined are equivalent to the ones defined by Stirling \cite{localmc-stirling}. The reason for this coincidence is that when a modality $\modald{a} \phi$ (resp.\ $\modalb{a} \phi$) is encountered, only Eve (resp.\ Adam) gets to choose both the next subformula and the transition used to verify (resp.\ falsify) the truth value of $\phi$.

Now, let us look at the case when a model with an empty independence relation is considered. In such a case the rule ($\modalb{~}_{nc}$) becomes trivially true and the rule ($\modald{~}_{nc}$) trivially false since in an interleaving model all pairs of transitions $(t,r)$ such that $\tau(t) = \sigma(r)$ are in $\leq$. Moreover, in such models the rules ($\modald{\otimes}$) and ($\modalb{\otimes}$) are ``absorbed'' by rules ($\modald{~}_c$) and ($\modalb{~}_{nc}$), respectively, because all maximal sets are necessarily singletons, and as before because all transitions are in $\leq$, whereas $\ominus$ is empty. For these reasons the elements that belong to the sets $\mathcal{X}$ and $\mathfrak{A}$ need no longer be considered and the rules ($\modalb{~}_{c}$) and ($\modald{~}_{c}$) become ($\modalb{~}$) and ($\modald{~}$), respectively. The other rules remain the same.
\qed
\end{proof}

\subsection{Model-Checking Infinite Posets} \label{mcres}
In this subsection we use the higher-order model-checking games for \tfl/ to push forward the decidability border of the model-checking problem of a particular class of partial order models, namely, of the class of regular event structures defined by Thiagarajan \cite{res-thia}. More precisely, we improve previous results by Penczek \cite{tacas-penczek} and Madhusudan \cite{mceslics-madhu} in terms of temporal expressive power. The solution presented here is the \tfl/ version of the same result obtained in \cite{con09-gut}.

\subsubsection*{\tfl/ on Regular Trace Event Structures.}
As we have shown previously, higher-order logic games for model-checking can be played in either finite or infinite state systems (with finite branching). However, decidability for the games was proved only for finite systems. Therefore, if the system at hand has recursive behaviour and, moreover, is represented by an event structure, then the TSI representation of it may be infinite, and decidability is not guaranteed.

We now analyse the decidability of the model-checking games for \tfl/ against a special class of infinite, but regular, event structures called `regular trace' event structures. This class of systems was introduced by Thiagarajan \cite{res-thia} in order to give a canonical representation to the set of Mazurkiewicz traces modelling the behaviour of a finite concurrent system. The model-checking problem for this class of models has been studied elsewhere \cite{mceslics-madhu,tacas-penczek}, and shown to be rather difficult. In the reminder of this section we show that model-checking \tfl/ properties of this kind of systems is also decidable. 

As described before, an event structure $\mathfrak{E} = (E,\preccurlyeq,\sharp,\eta,\Sigma)$ determines a TSI $\mathfrak{T} = (S,s_0,T,I,\Sigma)$ by means of an inclusion functor from the category $\mathcal{ES}$ of event structures to the category $\mathcal{TSI}$ of TSI. The mapping we presented before was given in a set-theoretic way since such a presentation is more convenient for us. A categorical definition is given by Joyal, Nielsen, and Winskel, which can be found in \cite{open-nielsen}. Let $\lambda: \mathcal{ES} \rightarrow \mathcal{TSI}$ be such a construction.

\begin{defn}
\emph{
A \emph{regular trace event structure} is a labelled event structure $\mathfrak{E} = (E,\preccurlyeq,\sharp,\eta,\Sigma)$ as in Definition \ref{esdef}, where for all configurations $C$ of $\mathfrak{E}$, and for all events $e \in C$, the set of future non-isomorphic configurations rooted at $e$ defines an equivalence relation of finite index.
}
\eos
\end{defn}

Let $\Conf/$ be the set of configurations of $\mathfrak{E}$. Notice that the restriction to image-finite models implies that the partial order $\preccurlyeq$ of $\mathfrak{E}$ is of \emph{finite branching}, and hence for all $C \in \Conf/$, the set of immediately next configurations is bounded. Also notice that the set of states $S$ of the TSI representation of an event structure $\mathfrak{E}$ is isomorphic to the set $\Conf/$ of configurations of $\mathfrak{E}$.

\subsubsection*{A Computable Folding Functor from Event Structures to TSI.}
In order to overcome the problem of dealing with infinite event structures, such as the regular trace event structures just defined, we present a new morphism (a functor) that folds a possibly infinite event structures into a TSI. This way, a finite process space can be constructed so as to give the semantics of \tfl/ formulae, and hence, play a model-checking game for this logic in a finite board. Such a morphism and the procedure to effectively compute it is described below.

\paragraph{The Quotient Set Method.}
Let $Q = ({\Conf/} / {\sim})$ be the \emph{quotient set representation} of $\Conf/$ by $\sim$ in a finite or infinite event structure $\mathfrak{E}$, where $\Conf/$ is the set of configurations in $\mathfrak{E}$ and $\sim$ is an equivalence relation on such configurations. The equivalence class $[X]_{\sim}$ of a configuration $X \in \Conf/$ is the set $\{{C \in {\Conf/}} \mid {C \sim X} \}$. A quotient set $Q$ where $\sim$ is decidable is said to have a decidable characteristic function, and will be called a \emph{computable quotient set}. 

\begin{defn}
\emph{ 
A \emph{regular quotient set} $({\Conf/} / {\sim})$ of an event structure $\mathfrak{E}$ is a computable quotient set representation of $\mathfrak{E}$ with a finite number of equivalence classes.
}
\eos
\end{defn}

Having defined a regular quotient set representation of $\mathfrak{E}$, the morphism $\lambda: \mathcal{ES} \rightarrow \mathcal{TSI}$ above can be modified to defined a new map $\lambda_f: \mathcal{ES} \rightarrow \mathcal{TSI}$ which folds a (possibly infinite) event structure into a TSI:
\begin{center}
\begin{tabular}{lcl}
$S$ & $=$ & $\{[C]_{\sim} \subseteq \Conf/ \mid \exists [X]_{\sim} \in Q = ({\Conf/} / {\sim}) . ~ {C \sim X} \}$ \\
$T$ & $=$ & $\{ ([C]_{\sim},a,[C']_{\sim}) \in S \times \Sigma \times S \mid \exists e \in E . ~ \eta(e) = a , e \not \in C , C' = C \cup \{e\}\}$ \\
$I$ & $=$ & $\{ (([C_1]_{\sim},a,[C'_1]_{\sim}),([C_2]_{\sim},b,[C'_2]_{\sim})) \in T \times T \mid \exists (e_1,e_2) \in {\co/} . ~ $\\&&$\eta(e_1) =a, \eta(e_2)=b, C'_1 = C_1 \cup \{e_1\}, C'_2 = C_2 \cup \{e_2\} \}$
\end{tabular} 
\end{center}

\begin{lem}\label{sattflformulae}
Let $\mathfrak{T}$ be a TSI and $\mathfrak{E}$ an event structure. If $\mathfrak{T} = \lambda_f(\mathfrak{E})$, then the models $(\mathfrak{T},\mathcal{V})$ and $(\mathfrak{E},\mathcal{V})$ satisfy the same set of \tfl/ formulae.
\end{lem}
\begin{proof}
The morphism $\lambda_f: \mathcal{ES} \rightarrow \mathcal{TSI}$ from the category of event structures to the category of TSI has a unique right adjoint $\varepsilon: \mathcal{TSI} \rightarrow \mathcal{ES}$, the unfolding functor that preserves labelling and the independence relation between events, such that for any $\mathfrak{E}$ we have that $\mathfrak{E'} = (\varepsilon \circ \lambda_f ) ~ ( \mathfrak{E} )$, where $\mathfrak{E'}$ is isomorphic to $\mathfrak{E}$. But \tfl/ formulae do not distinguish between models and their unfoldings, and hence cannot distinguish between $(\mathfrak{T},\mathcal{V})$ and $(\mathfrak{E'},\mathcal{V})$. Moreover, \tfl/ formulae do not distinguish between isomorphic models equally labelled, and therefore cannot distinguish between $(\mathfrak{E'},\mathcal{V})$ and $(\mathfrak{E},\mathcal{V})$ either.
\qed
\end{proof}

Having defined a morphism $\lambda_f$ that preserves \tfl/ properties, one can now define a procedure that constructs a TSI model from a given event structure.

\begin{defn}
\emph{ 
Let $\mathfrak{E} = (E,\preccurlyeq,\sharp,\eta,\Sigma)$ be an event structure and $({\Conf/}/{\sim})$ a regular quotient set representation of $\mathfrak{E}$. A \emph{representative set} $E_r$ of $\mathfrak{E}$ is a subset of $E$ such that $\forall C \in \Conf/ . ~ \exists X \subseteq E_r . ~ C \sim X$.
}
\eos
\end{defn}

\begin{lem}\label{sinfiniterep}
Let $\mathfrak{E}$ be an represented as a regular quotient set $({\Conf/}/{\sim})$. Then, a finite representative set $E_r$ of $\mathfrak{E}$ is effectively computable. 
\end{lem}
\begin{proof}
Construct a finite representative set $E_r$ as follows. Start with $E_r = \emptyset$ and $C_j = C_0 = \emptyset$, the initial configuration or root of the event structure. Check $C_j \sim X_i$ for every equivalence class $[X_i]_{\sim}$ in $Q = ({\Conf/}/{\sim})$ and whenever $C_j \sim X_i$ holds define both a new quotient set ${Q' = Q \setminus [X_i]_{\sim}}$ and a new $E_r = E_r \cup C_j$. This subprocedure terminates because there are only finitely many equivalence classes to check and the characteristic function of the quotient set is decidable. Now, do this recursively in a breadth-first search fashion in the partial order defined on $E$ by $\preccurlyeq$, and stop when the quotient set is empty. Since $\preccurlyeq$ is of finite branching and all equivalence classes must have finite configurations, the procedure is bounded both in depth and breath and the quotient set will always eventually get smaller. Hence, such a procedure always terminates. It is easy to see that this procedure only terminates when $E_r$ is a representative set of $\mathfrak{E}$.
\qed
\end{proof}

A finite representative set $E_r$ is big enough to define all states in the TSI representation of $\mathfrak{E}$ when using $\lambda_f$. However, such a set may not be enough to recognize all transitions in the TSI. In particular, cycles cannot be recognized using $E_r$. Therefore, it is necessary to compute a set $E_f$ where cycles in the TSI can be recognized. We call $E_f$ a \emph{complete representative set} of $\mathfrak{E}$. The procedure to construct $E_f$ is similar to the previous one.

\begin{lem}\label{finiterep}
Let $\mathfrak{E} = (E,\preccurlyeq,\sharp,\eta,\Sigma)$ be an event structure and $E_r$ a finite representative set of $\mathfrak{E}$. If $\mathfrak{E}$ is represented as a regular quotient set $({\Conf/}/{\sim})$, then a finite complete representative set $E_f$ of $\mathfrak{E}$ is effectively computable. 
\end{lem}
\begin{proof}
Start with $E_f = E_r$, and set $\mathfrak{C} = \Conf/(E_r)$, the set
of configurations generated by $E_r$. For each $C_j$ in $E_r$ check in $\preccurlyeq$ the set $Next(C_j)$ of next configurations to $C_j$, i.e., those configurations $C'_j$ such that $C'_j = C_j \cup \{e\}$ for some event $e$ in $E \setminus C_j$. Having computed $Next(C_j)$, set $E_f = E_f \cup (\bigcup Next(C_j))$ and $\mathfrak{C} = \mathfrak{C} \setminus \{C_j\}$, and stop when $\mathfrak{C}$ is empty. This procedure behaves as the one described previously. Notice that at the end of this procedure $E_f$ is complete since it contains the next configurations of all elements in $E_r$.
\qed
\end{proof}

\subsubsection*{Temporal Verification of Regular Infinite Event Structures.}
Based on Lemmas \ref{sattflformulae} and \ref{finiterep} and on Theorem \ref{dectflgames}, we can give a decidability result for the class of event structures introduced in \cite{res-thia} against \tfl/ specifications. Such a result, which is obtained by representing a regular event structure as a regular quotient set, is a corollary of the following theorem:

\begin{thm}\label{decmcinfes}
The model-checking problem for an event structure $\mathfrak{E}$ represented as a regular quotient set $({\Conf/}/{\sim})$ against \tfl/ specifications is decidable.
\end{thm}
\begin{proof}
Due to Lemma \ref{finiterep} one can construct a finite complete representative set $E_f$ of $E$. Then a finite TSI that satisfies the same set of \tfl/ formulae as $\mathfrak{E}$ can be defined by using the folding map $\lambda_f$ from event structures to TSI, and using $E_f$ instead of $E$ as the new set of events. Since such a morphism preserves all \tfl/ properties (Lemma \ref{sattflformulae}), the model-checking problem for this kind of event structures can be reduced to solving the model-checking game for finite TSI, and hence for finite systems in general, which due to Theorem \ref{dectflgames} is decidable.
\qed
\end{proof}

\textbf{Regular Event Structures as Finite CCS Processes.} A regular event structure can be generated by a finite concurrent system represented by a finite number of (possibly recursive) CCS processes \cite{ccs-milner,essemccs-winskel}. Syntactic restrictions on CCS that generate only finite systems have been studied. Notice that the combination of the \emph{syntactic} restriction to finite CCS processes and the \emph{semantic} restriction to image-finite models give the requirements for regularity on the event structures that are generated, in particular, of the regular trace event structures defined before.

Now, without any loss of generality, consider only deterministic CCS processes without auto-concurrency. A CCS process is deterministic if whenever $a.M+b.N$, then $a \neq b$, and similarly has no auto-concurrency if whenever $a.M \parallel b.N$, then $a \neq b$. Notice that any CCS process $P$ that either is nondeterministic or has auto-concurrency can be converted into an equivalent process $Q$ which generates an event structure that is isomorphic, up to relabelling of events, to the one generated by $P$. 

Eliminating nondeterminism and auto-concurrency can be done by relabelling events in $\wp(P)$, the powerset of CCS processes of $P$, with an injective map $\theta: \Sigma \rightarrow \Sigma^{*}$ (where $\Sigma^{*}$ is a set of labels and $\Sigma \subseteq \Sigma^{*}$), and by extending the `synchronization algebra' \cite{essemccs-winskel} according to the new labelling of events so as to preserve pairs of (labels of) events that can synchronize. Also notice that the original labelling can always be recovered from the new one, i.e., the one associated with the event structure generated by $Q$, since $\theta$ is injective and hence has inverse $\theta^{-1} : \Sigma^{*} \rightarrow \Sigma$. 

\textbf{Finite CCS Processes as Regular Quotient Sets.} Call $ESProc(P)$ the set of configurations of the event structure generated by a CCS process $P$ of the kind described above. The set $ESProc(P)$ together with an equivalence relation between CCS processes $\equiv_{CCS}$ given simply by syntactic equality between them is a regular quotient set representation $({ESProc(P)} ~/~ {\equiv_{CCS}})$ of the event structure generated by $P$. 

Notice that since there are finitely many different CCS expressions, i.e., $\wp(P)$ is finite, then the event structure generated by $P$ is of finite-branching and the number of equivalence classes is also bounded. Finally, $\equiv_{CCS}$ is clearly decidable because the process $P$ is always associated with the $\emptyset$ configuration and any other configuration in $ESProc(P)$ can be associated with only one CCS expression in $\wp(P)$ as they are deterministic and have no auto-concurrency after relabelling.

The previous simple observations lead to the following result:

\begin{cor}
Model-checking regular trace event structures against \tfl/ specifications is decidable.
\end{cor}

A similar result was given by Bradfield and the present author \cite{con09-gut} using SFL. However, as stated previously, since it is still unknown the exact relationship between the expressivity of SFL and \tfl/, then we are, as to now, unable to decide which of the two results is stronger in terms of temporal expressive power. It may well be that SFL and \tfl/ are actually equi-expressive on $\Xi$-systems.

\section{Discussion and Related Work}\label{relwork}
The work presented here is related to three connected topics: mu-calculi, bisimulation equivalences, and model-checking problems. We were particularly interested in mu-calculi as fixpoint extensions of modal logic, and bisimulation and model-checking problems from a game-theoretic perspective. Our results relate, mainly, to work on these topics with respect to partial order models of concurrency; however, since we also embrace interleaving systems in a very natural way, in some cases pointers to similar work in the interleaving context are given.

{\bf Mu-Calculi.} This work can be related to logics with partial order semantics at large. Formulae of these logics, usually, are given denotations that consider the one-step interleaving semantics of a particular partial order model. Following this approach no new logical constructions have to be introduced; unfortunately, in this case the explicit notion of concurrency in the models is completely lost. 

Thus, the usual approach when defining logics with partial order models is to introduce operators that somehow capture the independence information on the partial order models. In most cases that kind of logical independence is actually a sequential interpretation of concurrency, which is based on the introduction of past operators sensitive to concurrent transitions and a mixture of forwards and backwards reasoning; however, this can lead to undecidability results with respect to the decision problems related to such logics, e.g., with respect to their satisfiability, equivalence, or model-checking problems (cf. \cite{pathlogic-nielsen,undec-penczek,pologics-penczek,logtracesbook-penczek}).

Several logics with the characteristics described above whose semantics are given using partial order models (as well as their related decision problems) can be found in \cite{pologics-penczek,logtracesbook-penczek}, and the references therein. Other logics with partial order semantics that do not appear there can be found, e.g., in \cite{causallogic-alur,pathlogic-nielsen}, but the literature includes many, many more references. However, it is worth saying that not all such logics are extensions of modal logic or even mu-calculi. In some cases, they are variations of the usual linear-time and branching-time temporal logics.

At a more philosophical level, this study is also similar to that in \cite{henkinlogic-bradfield,njc-bradfield}, a work primarily on mathematical logic using game logics for concurrency. In these works the main goal is to explicitly capture what we call `model independence', i.e., explicit concurrency in the models, in a logical way with the use of Henkin quantifiers, which are partial order generalisations of the usual quantifiers in classical logic. More precisely, in \cite{henkinlogic-bradfield} different properties of a number of fixpoint modal logics based on Hintikka and Sandu's `Independence-Friendly' (IF) logic are discussed, and in \cite{njc-bradfield} the bisimulation equivalence induced by one of such logics, namely of IF modal logic (IFML), is thoroughly studied. Their main motivation closely relates to ours, especially because of their interest in the bisimulation equivalences induced by such logics as well as the use of games.

{\bf Bisimulation.} A great deal of work has been done on the study of bisimulation equivalences captured by modal and temporal logics, e.g., as done by Milner and Hennessy \cite{hmljacm-milner} of by De Nicola and Vaandrager \cite{bbjacm-rocco} in and interleaving context. Here, we have proposed a generic approach in which the standard bisimilarity for interleaving concurrency, namely sb, is captured by syntactic and semantics means using \tfl/ and its associated (higher-order) bisimulation game.

Also, the work by Joyal, Nielsen, and Winskel \cite{open-nielsen} relates to ours. Whereas in \cite{open-nielsen} they proposed a \emph{categorical} approach to defining an abstract or model independent notion of bisimulation equivalence for several concurrent systems, here we have proposed a \emph{logical} one, following the way of reasoning used in \cite{hmljacm-milner,bbjacm-rocco}, but in our case in a partial order setting instead of an interleaving one. Moreover, the bisimilarities we studied here are all decidable, a result that contrasts with the work by Nielsen and Clausen \cite{pathlogic-nielsen}, where a concretization of the abstract notion of bisimulation defined in \cite{open-nielsen} is given a logical as well as a game-theoretic characterisation, which turned out to be undecidable.

As briefly mentioned before, Bradfield and Fr{\"o}schle \cite{njc-bradfield} also studied the bisimulation equivalence induced by IFML using game-theoretic techniques. They followed a logical approach, which is in spirit quite close to our work. Unfortunately, the bisimulation equivalences induced by the logic studied there do not coincide, in most cases, with the standard bisimilarities for partial order models of concurrency, not even when restricted to particular classes of systems. 

{\bf Model-checking.} Model-checking games have been an active area of research in the last decades (cf. \cite{mcgames-gradel,games-wal}). They have been studied from both theoretical and practical perspectives. For instance, for the proper definition of their mathematical properties \cite{gamesbook-gradel,phdthesis-lange,mcgames-stirling}, or for the construction of tools for property verification, e.g., see \cite{appgames-ghica,practicalmc-stevens}. Most approaches based on games have considered either only interleaving systems or the one-step interleaving semantics of partial order models. Our work differs from these approaches in that we deal with games played on partial order models without considering interleaving simplifications. Although verification problems in finite partial order models can be undecidable, the games presented here are all \emph{decidable} in the finite case.

Regarding model-checking in a broader sense, many procedures, not only game-theoretic, have been studied elsewhere for concurrent systems both with interleaving models and with partial order semantics. For instance, see \cite{causallogic-alur,mcmit-grumberg,logtracesbook-penczek}, as well as the references therein, for several examples of various techniques and approaches to model-checking concurrent systems. However, since our main motivation was to develop a decision procedure to verify concurrent systems with partial order models, only the techniques considering these kinds of systems relate to our work, though, as said before, such procedures are not game-theoretic.

With respect to the temporal verification of event structures, previous studies have been done on restricted classes. Closer to our work is \cite{mceslics-madhu,tacas-penczek}. Indeed, model-checking regular trace event structures has turned out to be rather difficult and previous work has shown that verifying monadic second-order (MSO) properties on these structures is already undecidable. For this reason weaker (classical, modal, and temporal) logics have been studied. Unfortunately, although very interesting results have been achieved, especially in \cite{mceslics-madhu} where CTL$^*$ temporal properties can be verified, previous approaches have not managed to define decidable theories for a logic with enough expressive power to describe all usual temporal properties as can be done with \lmu/ in the interleaving case, and therefore with \tfl/ when considering partial order models for concurrency.

The difference between the logics and decision procedure presented in \cite{mceslics-madhu} and the approach we presented here is that in \cite{mceslics-madhu} a \emph{global} second-order quantification on conflict-free sets in the partial order model is permitted, whereas only a \emph{local} second-order quantification in the same kind of sets is defined here, but such a second-order power can be embedded into fixpoint specifications, which in turn allows one to express more temporal properties. In this way we are able to improve, in terms of temporal expressive power, previous results on model-checking regular trace event structures against a branching-time temporal logic.

Finally, it is important to point out that most of the results we have presented here with respect to the logics as well as to the bisimulation and model-checking games have been also obtained using SFL (instead of \tfl/ as in this report); such results were first presented in \cite{fos09-gut,con09-gut}. This report is based on the work presented there, but it also contains some modifications, corrections, and refinements.

\section{Concluding Remarks}\label{conc}
We have given a logical characterisation to the \emph{dualities} that can be found when analysing \emph{locally} the relationships between concurrency and conflict as well as concurrency and causality in partial order models of concurrency. This characterisation aims at defining relationships between equivalences that take into account the explicit notion of independence when considering partial order semantics, and which can be defined at the level of the models as well as at the level of the logics. This study led to several positive results with respect to the \emph{bisimulation} and \emph{model-checking} problems associated with such logics and models. It also delivered new forms of logic games for verification where the players are given \emph{higher-order} power on the sets of elements they are allowed to play.

A key ingredient of the work reported here is that we allow a free interplay of \emph{fixpoints} and \emph{local monadic second-order} power in both the mu-calculi and higher-order logic games for verification we have presented. Our results, together with the analysis of some of the related work, suggest that restricting the quantification power to \emph{conflict-free} sets (of transitions) in partial order models of concurrency may be a sensible way of retaining decidability while still having good expressivity. These features, along with the fact that the mu-calculi and logic games we have defined for partial order models generalise those for interleaving concurrency, make our logic-based game-theoretic framework a powerful alternative approach to studying different kinds of concurrent systems \emph{uniformly}, regardless of whether they have an interleaving or a partial order semantics.

\subsection*{Acknowledgements.} I thank Julian Bradfield, Sibylle Fr{\"o}schle, Ian Stark, and Colin Stirling for helpful discussions. I also thank the reviewers of the FOSSACS 2009 and CONCUR 2009 conferences for their comments on preliminary versions of this work. The author was financially supported by an Overseas Research Studentship (ORS) Award and a School of Informatics PhD Studentship at The University of Edinburgh.

\bibliographystyle{abbrv}
\bibliography{bibreport}

\begin{thebibliography}{10}

\bibitem{causallogic-alur}
R.~Alur, D.~Peled, and W.~Penczek.
\newblock Model-checking of causality properties.
\newblock In {\em LICS}, 90--100. IEEE Computer Society, 1995.

\bibitem{jvbthesis-jvb}
J.~V. Benthem.
\newblock {\em Modal Correspondence Theory}.
\newblock PhD thesis, University of Amsterdam, 1977.

\bibitem{lg-jvb}
J.~V. Benthem.
\newblock Logic games, from tools to models of interaction.
\newblock In {\em Logic at the Crossroads}, 283--317. Allied Publishers, 2007.

\bibitem{alt-bradfield}
J.~C. Bradfield.
\newblock The modal $\mu$-calculus alternation hierarchy is strict.
\newblock {\em Theor. Comput. Sci.}, 195(2):133--153, 1998.

\bibitem{henkinlogic-bradfield}
J.~C. Bradfield.
\newblock Independence: logics and concurrency.
\newblock In {\em Truth and Games: Essays in Honour of Gabriel Sandu}, Acta
  Phil. Fennica \textbf{78}, 47--70. Soc. Phil. Fen., 2006.

\bibitem{njc-bradfield}
J.~C. Bradfield and S.~B. Fr{\"o}schle.
\newblock Independence-friendly modal logic and true concurrency.
\newblock {\em Nord. J. Comput.}, 9(1):102--117, 2002.

\bibitem{localtcs-bradfield}
J.~C. Bradfield and C.~Stirling.
\newblock Local model checking for infinite state spaces.
\newblock {\em Theor. Comput. Sci.}, 96(1):157--174, 1992.

\bibitem{mucalculi-bradfield}
J.~C. Bradfield and C.~Stirling.
\newblock Modal mu-calculi.
\newblock In {\em Handbook of Modal Logic}, 721--756. Elsevier, 2007.

\bibitem{mcmit-grumberg}
E.~Clarke, O.~Grumberg, and D.~Peled.
\newblock {\em Model Checking}.
\newblock The MIT Press, 2000.

\bibitem{tltolmu-dam}
M.~Dam.
\newblock {CTL*} and {ECTL*} as fragments of the modal mu-calculus.
\newblock {\em Theor. Comput. Sci.}, 126(1):77--96, 1994.

\bibitem{fcnets-esparza}
J.~Desel and J.~Esparza.
\newblock {\em Free Choice Petri Nets}.
\newblock Cambridge Tracts in Theoretical Computer Science \textbf{40}.
  Cambridge University Press, 1995.

\bibitem{hier-fecher}
H.~Fecher.
\newblock A completed hierarchy of true concurrent equivalences.
\newblock {\em Inf. Process. Lett.}, 89(5):261--265, 2004.

\bibitem{border-froschle}
S.~B. Fr{\"o}schle.
\newblock The decidability border of hereditary history preserving
  bisimilarity.
\newblock {\em Inf. Process. Lett.}, 93(6):289--293, 2005.

\bibitem{appgames-ghica}
D.~R. Ghica.
\newblock Applications of game semantics: From program analysis to hardware
  synthesis.
\newblock In {\em LICS}, 17--26. IEEE Computer Society, 2009.

\bibitem{hpb-glabbeek}
R.~J.~V. Glabbeek and U.~Goltz.
\newblock Refinement of actions and equivalence notions for concurrent systems.
\newblock {\em Acta Inf.}, 37(4/5):229--327, 2001.

\bibitem{mcgames-gradel}
E.~Gr{\"a}del.
\newblock Model checking games.
\newblock {\em Electr. Notes Theor. Comput. Sci.}, 67, 2002.

\bibitem{gamesbook-gradel}
E.~Gr{\"a}del, W.~Thomas, and T.~Wilke, editors.
\newblock {\em Automata, Logics, and Infinite Games}, LNCS \textbf{2500}.
  Springer, 2002.

\bibitem{fos09-gut}
J.~Gutierrez.
\newblock Logics and bisimulation games for concurrency, causality and
  conflict.
\newblock In {\em FOSSACS}, LNCS \textbf{5504}, 48--62. Springer, 2009.

\bibitem{con09-gut}
J.~Gutierrez and J.~C. Bradfield.
\newblock Model-checking games for fixpoint logics with partial order models.
\newblock In {\em CONCUR}, LNCS \textbf{5710}, 354--368. Springer, 2009.

\bibitem{hmljacm-milner}
M.~Hennessy and R.~Milner.
\newblock Algebraic laws for nondeterminism and concurrency.
\newblock {\em J. ACM}, 32(1):137--161, 1985.

\bibitem{dechpb-meyer}
L.~Jategaonkar and A.~R. Meyer.
\newblock Deciding true concurrency equivalences on safe, finite nets.
\newblock {\em Theor. Comput. Sci.}, 154(1):107--143, 1996.

\bibitem{open-nielsen}
A.~Joyal, M.~Nielsen, and G.~Winskel.
\newblock Bisimulation from open maps.
\newblock {\em Inf. Comput.}, 127(2):164--185, 1996.

\bibitem{undhhpb-jur}
M.~Jurdzinski, M.~Nielsen, and J.~Srba.
\newblock Undecidability of domino games and hhp-bisimilarity.
\newblock {\em Inf. Comput.}, 184(2):343--368, 2003.

\bibitem{lmutcs-kozen}
D.~Kozen.
\newblock Results on the propositional mu-calculus.
\newblock {\em Theor. Comput. Sci.}, 27:333--354, 1983.

\bibitem{phdthesis-lange}
M.~Lange.
\newblock {\em Games for Modal and Temporal Logics}.
\newblock PhD thesis, Univeristy of Edinburgh, 2002.

\bibitem{mcgames-stirling}
M.~Lange and C.~Stirling.
\newblock Model checking games for branching time logics.
\newblock {\em J. Log. Comput.}, 12(4):623--639, 2002.

\bibitem{mceslics-madhu}
P.~Madhusudan.
\newblock Model-checking trace event structures.
\newblock In {\em LICS}, 371--380. IEEE Computer Society, 2003.

\bibitem{boreldet-martin}
D.~A. Martin.
\newblock Borel determinacy.
\newblock {\em Ann. Math.}, 102(2):363--371, 1975.

\bibitem{tracesbook-maz}
A.~W. Mazurkiewicz.
\newblock Introduction to trace theory.
\newblock In {\em The Book of Traces}, 3--42. World Scientific, 1995.

\bibitem{ccs80-milner}
R.~Milner.
\newblock {\em A Calculus of Communicating Systems}.
\newblock LNCS \textbf{92}. Springer, 1980.

\bibitem{ccs-milner}
R.~Milner.
\newblock {\em Communication and Concurrency}.
\newblock Prentice-Hall, 1989.

\bibitem{bbjacm-rocco}
R.~D. Nicola and F.~W. Vaandrager.
\newblock Three logics for branching bisimulation.
\newblock {\em J. ACM}, 42(2):458--487, 1995.

\bibitem{pathlogic-nielsen}
M.~Nielsen and C.~Clausen.
\newblock Games and logics for a noninterleaving bisimulation.
\newblock {\em Nord. J. Comput.}, 2(2):221--249, 1995.

\bibitem{pnesdom-winskel}
M.~Nielsen, G.~D. Plotkin, and G.~Winskel.
\newblock Petri nets, event structures and domains, {Part I}.
\newblock {\em Theor. Comput. Sci.}, 13:85--108, 1981.

\bibitem{models-winskel}
M.~Nielsen and G.~Winskel.
\newblock Models for concurrency.
\newblock In {\em Handbook of Logic in Computer Science}, 1--148. Oxford
  University Press, 1995.

\bibitem{fixpv-park}
D.~M.~R. Park.
\newblock Fixpoint induction and proofs of program properties.
\newblock {\em Machine Intelligence}, 5:59--78, Edinburgh University Press,
  1969.

\bibitem{bis-park}
D.~M.~R. Park.
\newblock Concurrency and automata on infinite sequences.
\newblock LNCS \textbf{104}, 167--183. Springer, 1981.

\bibitem{undec-penczek}
W.~Penczek.
\newblock On undecidability of propositional temporal logics on trace systems.
\newblock {\em Inf. Process. Lett.}, 43(3):147--153, 1992.

\bibitem{pologics-penczek}
W.~Penczek.
\newblock Branching time and partial order in temporal logics.
\newblock In {\em Time and Logic: A Computational Approach}, 179--228. UCL
  Press, 1995.

\bibitem{tacas-penczek}
W.~Penczek.
\newblock Model-checking for a subclass of event structures.
\newblock In {\em TACAS}, LNCS \textbf{1217}, 145--164. Springer, 1997.

\bibitem{logtracesbook-penczek}
W.~Penczek and R.~Kuiper.
\newblock Traces and logic.
\newblock In {\em The Book of Traces}, 307--390. World Scientific, 1995.

\bibitem{hpb-rav}
A.~M. Rabinovich and B.~A. Trakhtenbrot.
\newblock Behavior structures and nets.
\newblock {\em Fundam. Inform.}, 11:357--403, 1988.

\bibitem{modelstcs-winskel}
V.~Sassone, M.~Nielsen, and G.~Winskel.
\newblock Models for concurrency: Towards a classification.
\newblock {\em Theor. Comput. Sci.}, 170(1-2):297--348, 1996.

\bibitem{confusiontcs-esmith}
E.~Smith.
\newblock On the border of causality: Contact and confusion.
\newblock {\em Theor. Comput. Sci.}, 153(1{\&}2):245--270, 1996.

\bibitem{practicalmc-stevens}
P.~Stevens and C.~Stirling.
\newblock Practical model-checking using games.
\newblock In {\em TACAS}, LNCS \textbf{1384}, 85--101. Springer, 1998.

\bibitem{localmc-stirling}
C.~Stirling.
\newblock Local model checking games.
\newblock In {\em CONCUR}, LNCS \textbf{962}, 1--11. Springer, 1995.

\bibitem{processes-stirling}
C.~Stirling.
\newblock {\em Modal and Temporal Properties of Processes}.
\newblock Texts in Computer Science. Springer, 2001.

\bibitem{lmctab-stirling}
C.~Stirling and D.~Walker.
\newblock Local model checking in the modal mu-calculus.
\newblock {\em Theor. Comput. Sci.}, 89(1):161--177, 1991.

\bibitem{ktft-tarski}
A.~Tarski.
\newblock A lattice-theoretical fixpoint theorem and its applications.
\newblock {\em Pacific J. Math.}, 5(2):285--309, 1955.

\bibitem{res-thia}
P.~S. Thiagarajan.
\newblock Regular trace event structures.
\newblock Technical report, BRICS, 1996.

\bibitem{games-wal}
I.~Walukiewicz.
\newblock A landscape with games in the background.
\newblock In {\em LICS}, 356--366. IEEE Computer Society, 2004.

\bibitem{essemccs-winskel}
G.~Winskel.
\newblock Event structure semantics for {CCS} and related languages.
\newblock In {\em ICALP}, LNCS \textbf{140}, 561--576. Springer, 1982.

\end{thebibliography}

\end{document}